\documentclass[12pt,reqno]{amsart}

\usepackage{hyperref}
\usepackage{pgf,tikz}
\usetikzlibrary{arrows}
\usepackage{amsmath,amsfonts,amssymb}
\usepackage[latin1]{inputenc}
\usepackage{verbatim}

\newtheorem{theorem}{Theorem}[section]
\newtheorem{lemma}[theorem]{Lemma}
\newtheorem{prop}[theorem]{Proposition}
\newtheorem{cor}[theorem]{Corollary}
\newtheorem{rmk}[theorem]{Remark}
\newtheorem{defn}[theorem]{Definition}
\newtheorem{hyp}[theorem]{Assumption}

\newcounter{defn}

\numberwithin{equation}{section}
\newcommand{\sect}{\vspace{3mm} \setcounter{equation}{0} \setcounter{defn}{0}
\section}

\renewcommand{\Re}{{\rm Re }\;}
\renewcommand{\Im}{{\rm Im }\;}

\newcommand{\be}{\begin{equation}}
\newcommand{\ee}{\end{equation}}
\newcommand{\bea}{\begin{eqnarray}}
\newcommand{\eea}{\end{eqnarray}}
\newcommand{\beas}{\begin{eqnarray*}}
\newcommand{\eeas}{\end{eqnarray*}}

\def\t{\tilde}

\newcommand{\beq}{\begin{equation}}
\newcommand{\eeq}{\end{equation}}
\newcommand{\beqs}{\begin{equation*}}
\newcommand{\eeqs}{\end{equation*}}
\newcommand{\bal}{\begin{aligned}}
\newcommand{\eal}{\end{aligned}}

\begin{document}

\title{TUNNELING EFFECT IN TWO DIMENSIONS WITH VANISHING MAGNETIC FIELDS}

\author{KHALED ABOU ALFA}

\address{ Laboratoire de Math{\'e}matiques Jean Leray\\
Universit{\'e} de Nantes \\
2 rue de la Houssini{\`e}re\\
BP 92208 F-44322 Nantes Cedex 3, France }
\address{\textit{Email-address}: khaled.abou-alfa@univ-nantes.fr}


\begin{abstract} In this paper, we consider the semiclassical 2D magnetic Schr{\"o}dinger operator in the case where the magnetic field vanishes along a smooth closed curve. Assuming that this curve has an axis of symmetry, we prove that semi-classical tunneling occurs. The main result is an expression of the splitting of the first two eigenvalues and an explicit tunneling formula. 
\end{abstract}

\maketitle

\sect{Introduction}

\subsection{Motivation}
We consider two functions $\mathbf{A}:\mathbb{R}_x^d\longrightarrow \mathbb{R}^d$ and $\mathbf{V}:\mathbb{R}_x^d\longrightarrow\mathbb{R}$ corresponding to the magnetic potential and the electric potential respectively. These two potentials provide an electromagnetic field $(E,B)$ defined by
$$ E=\nabla \mathbf{V} \text{  and   } B=\nabla\times \mathbf{A}. $$
Considering the Schr{\"o}dinger equation
\begin{equation}
    \mathrm{i}\hbar \partial_t\Psi =\left(\left( -\mathrm{i}\hbar\nabla+\mathbf{A} \right)^2+\mathbf{V}\right)\Psi,
    \label{Sch1}
\end{equation}
for $t>0$, $x\in\mathbb{R}^d$ and $\Psi$ a normalized solution of \eqref{Sch1}, $\left| \Psi(x,t) \right|^2$ is then the probability density of presence of the particle at point $x$ and at time $t$. Here, the small parameter $\hbar$ stands for the Planck's constant which is $\hbar\sim 10^{-34}J.s$ in physical situations. This parameter is considered here as a strictly positive semi-classical parameter close to $0^+$, in the spirit of the so called semi-classical analysis.
 
A particular solution of equation \eqref{Sch1} is then
$$ \Psi(x,t)=\varphi(x)\mathrm{e}^{-\frac{\mathrm{i}\lambda t}{\hbar}}, $$
where $\lambda$ and $\varphi$ verify 
$$ \left(\left( -\mathrm{i}\hbar\nabla+\mathbf{A} \right)^2+\mathbf{V}\right)\varphi=\lambda\varphi. $$

We are interested here in the determination of such a so-called eigenpair $(\lambda,\varphi)$ in the semiclassical limit ($\hbar\rightarrow0$).

In some cases (where there are symmetries), the first two lowest eigenvalues may be exponentially close (w.r.t $\hbar$), leading to what is called tunneling effect. The tunneling effect is an important physical phenomenon. Mathematically, this phenomenon was studied in particular in the $80$'s by Helffer and Sj\"ostrand in the case where the magnetic potential $\mathbf{A}=0$ and the electric potential has non-degenerate minima \cite{helffer1984multiple,helffer1985multiple,helffer1985puits,helffer1987effet}. They proved that the ground states are concentrated near the minima of the potential $\mathbf{V}$.

This article deals with the same tunneling question, but when $\mathbf{V}=0$ and in a particular geometric situation. A first answer to this type of question was found by Bonnaillie-H{\'e}rau-Raymond \cite{BHR-purely} in the case where the magnetic field is constant in a open, bounded and regular domain of $\mathbb{R}^2$ with the Neumann condition on the boundary. In that work, the authors found an explicit expression of the difference between the first two eigenvalues, leading to the first explicit tunneling formula in a pure magnetic situation. A second case was studied by Fournais-Helffer-Kachmar \cite{fournais2022tunneling} in the case where the magnetic field is a piecewise constant function with a jump discontinuity along a symmetric curve.

In this paper, we work with a variable magnetic field in $\mathbb{R}^2$. We prove that, under some symmetry and small variability conditions on the magnetic field, the tunneling effect also occurs. Note this work is the first one providing tunneling effect results in the case where the magnetic field is variable. 

\subsection{Semiclassical magnetic Laplacian}
The purely magnetic Laplacian in $\mathbb{R}^2$ is defined by
$$ \mathcal{L}_{\hbar}=\left( -\mathrm{i}\hbar\nabla+\mathbf{A} \right)^2,\,\,\,\,\,\,\operatorname{Dom}\left( \mathcal{L}_{\hbar} \right)=\left\{ \psi \in \operatorname{L}^2\left( \mathbb{R}^2\right):\mathcal{L}_{\hbar}\psi\in \operatorname{L}^2\left( \mathbb{R}^2 \right)   \right\}, $$
with $\mathbf{A}=(A_1,A_2)\in\mathcal{C}^{\infty}(\mathbb{R}^2,\mathbb{R}^2)$. Note that by gauge invariance, this self adjoint operator is unitarily equivalent to 
$$ \left( -\mathrm{i}\hbar\nabla+\mathbf{A}+\nabla\phi \right)^2, $$
for any suitable real valued function $\phi$. This gauge transformation ensures that the spectrum of $\mathcal{L}_{\hbar}$ depends only on the magnetic field $B=\nabla\times \mathbf{A}$. We assume that $\underset{\left| x \right|\rightarrow +\infty}{\text{lim}}B(x)=+\infty$, to ensure that the resolvent of $\mathcal{L}_{\hbar}$ is compact. In this case, we can consider the non-decreasing sequence of eigenvalues $\left( \lambda_n(\hbar) \right)_{n\in\mathbb{N}^{\ast}}$.

In this paper, we will focus on a variable magnetic field that vanishes to order $k\geq1$ on a smooth compact connected curve $\Gamma$ of $\mathbb{R}^2$. In Section \ref{section2}, the tubular coordinates $(s,t)$ in the neighborhood of the zero curve $\Gamma$ are defined in detail, where $s$ is the arc length of $\Gamma$ and $t$ is the normal distance to $\Gamma$. With this tubular coordinates and the diffeomorphism $\Phi$ defined in \eqref{Phii}, we define the function $\gamma$ on $\Gamma$ by
\begin{equation}
 \gamma(s):=\frac{1}{k!}\left( \partial^k_{t^k}\left( B\circ\Phi \right)(s,0) \right).
    \label{gammaaa}
\end{equation}
The objective is then to find an explicit approximation of the difference between the first two eigenvalues $\lambda_2(\hbar)-\lambda_1(\hbar)$ of $\mathcal{L}_{\hbar}$, thanks to $\gamma$ and other geometric quantities. 

An important toy model in our context is the so-called generalized Montgomery operator, which is the self-adjoint realization, on $L^2(\mathbb{R},dt)$, of the following operator
$$ \mathfrak{h}_{\xi}^{[k]}=D_{t}^2+\left( \xi-\frac{t^{k+1}}{k+1}\right)^2,\,\,\,\,\,\,k\in \mathbb{N}^{\ast}. $$
The spectrum of this operator was already studied in \cite{montg}, in which it is proven that the function $\mathbb{R}\ni\xi\mapsto\nu^{[k]}(\xi)$ admits a unique non-degenerate minimum at $\xi_0^{[k]}$ and that $\nu^{[k]}(\xi_0^{[k]}) > 0$, where $\nu^{[k]}(\xi)$ is the eigenvalue of $\mathfrak{h}_{\xi}^{[k]}$. This function will be crucial is the monic result.

The spectrum of the magnetic Laplacian $\mathcal{L}_{\hbar}$ has been the subject of many works \cite{bernoff1998onset,bonnaillie2005fundamental,fournais2010spectral,helffer2001magnetic,helffer2003upper,lu1999eigenvalue}, particularly in the context of superconductivity, in which the asymptotic description of the third critical field associated with the Ginzburg-Landau functional is related to the ground state energy of the magnetic Laplacian.

In this paper we will follow the strategy of Helffer and Sj\"ostrand which has been recently applied to understand the tunneling effect for the Neumann realization in a bounded domain. This strategy has been already used in the paper \cite{BHR-purely} by Bonnaillie-H{\'e}rau-Raymond and in paper \cite{fournais2022tunneling} by Fournais-Helffer-Kachmar.

Earlier spectral rigorous results were obtained in the case of the magnetic laplacian with vanishing magnetic field \cite{montgomery1995hearing,helffer2009spectral,helffer1996semiclassical,dombrowski-Raymond}. Helffer and Morame exhibited normal Agmon estimates which allow to show the localization of the eigenfunctions in the neighborhood of the zero curve $\Gamma$ \cite{helffer1996semiclassical}. Helffer and Kordyukov found the first term of the asymptotic expansion of the groundstate energy of $\mathcal{L}_{\hbar}$ in \cite{helffer2009spectral}, and the following asymptotic formula is established
\begin{equation}
\lambda_1(\hbar)=\gamma_0^{\frac{2}{k+2}}\nu^{[k]}\left( \xi_0^{[k]} \right)\hbar^{\frac{2k+2}{k+2}}+o\left( \hbar^{\frac{2k+2}{k+2}} \right),
    \label{int1}
\end{equation}
where $\gamma_0>0$ is the minimum of the function $\gamma$ on $\Gamma$. In \cite{dombrowski-Raymond}, Dombrowski and Raymond also proved that when the function $\gamma$ has a unique and non-degenerate minimum $\gamma_0>0$ at $s=0$ on $\Gamma$,  the first eigenfunctions are somehow localized near the point $s=0$. They established that, for $k=1$ and for all $n\geq1$,
\begin{equation}
\lambda_n(\hbar)=\theta_0^n\hbar^{\frac{4}{3}}+\theta_1^n\hbar^{\frac{5}{3}}+o\left( \hbar^{\frac{5}{3}} \right),
    \label{int2}
\end{equation}
with $\theta_0^n:=\gamma_0^{\frac{2}{3}}\nu^{[1]}\left( \xi_0^{[1]} \right)$ and $\theta_1^n:=\gamma_0^{\frac{2}{3}}C_0+\gamma_0^{\frac{2}{3}}(2n-1)\left( \frac{2\nu^{[1]}\left( \xi_0^{[1]}\right) \left(\nu^{[1]} \right)''\left( \xi_0^{[1]}\right)\gamma_0}{3\gamma''(0)} \right)$, where $C_0$ is a constant.

An open question for the magnetic Laplacian was whether the eigenfunctions have a similar approximation as the eigenvalues in \eqref{int2}, i.e. whether we can approximate the eigenfunctions by asymptotics of the form
\begin{equation}
\mathrm{e}^{-\frac{\Phi(s)}{\hbar^{\alpha}}}\displaystyle\sum_{j\geq1}a_{j}(s,t)\hbar^j,
    \label{int3}
\end{equation}
for some $\alpha>0$. A positive answer to this question was found by Bonnaillie-H{\'e}rau-Raymond in \cite{BHR-BKW}, in which they give a formal WKB expansions for the eigenfunctions of the magnetic Laplacian. The function $\Phi$ that appears in \eqref{int3} is a solution of an equation called ``eikonal equation''. In papers \cite{BHR-purely} and \cite{fournais2022tunneling}, the eikonal equation has explicit solutions and thus the function $\Phi$ can be found explicitly as a function of the curvature of the boundary. This shows that the tunneling effect is linked the curvature. However, in this paper, the situation is different. 

In this work, the eikonal equation is given by (see section \ref{eikn}):
\begin{equation}
\gamma(\sigma)^{\frac{2}{k+2}}\nu^{[k]}\left( \xi_0^{[k]}+\mathrm{i}\Phi'(\sigma) \right)=\gamma_0^{\frac{2}{k+2}}\nu^{[k]}\left( \xi_0^{[k]} \right),
    \label{int4}
\end{equation}
where $\gamma_0:=\underset{s\in\Gamma}{\text{min}}\gamma(s)>0$. We note that in this equation, we implicitly use an holomorphic extension, in a complex neighborhood of $\xi_0^{[k]}$, of the function $\mathbb{R}\ni\xi\mapsto\nu^{[k]}(\xi)$ associated to the Montgomery operator $\mathfrak{h}_{\xi}^{[k]}$. For the solution of this eikonal equation to be a priori well defined, we shall assume that 
$$ \left\| 1-\frac{\gamma_0}{\gamma} \right\|_{\infty} $$
is sufficiently small.

The eikonal equation \eqref{int4} is an implicit complex equation, and a priori its solution is an unknown complex valued function. This induces difficulties not appearing in \cite{BHR-purely} and \cite{fournais2022tunneling}.
\subsection{ Main result}
We work under the following assumptions on the geometry and the potential.
\begin{hyp}
It is assumed that the magnetic field $B$ vanishes out exactly to order $k\geq1$ on a closed, smooth, non empty compact and connected curve $\Gamma\subset\mathbb{R}^2$. It is further assumed that 
\begin{enumerate}
    \item[(i)] $B$ is symmetric with respect to the $x_2$-axis and therefore $\Gamma$ also;
    \item[(ii)]the function $\gamma$ on $\Gamma$ admits a unique non-degenerate minimum $\gamma_0 >0$ which is reached only at two distinct symmetric points $a_1,a_2\in\Gamma$. We suppose that $s_r$ and $s_l$ are the respective arc lengths for $a_1$ and $a_2$;
    \item[(iii)] $ \left\| 1-\frac{\gamma_0}{\gamma} \right\|_{\infty} $ is sufficiently small.
\end{enumerate}
\label{Hyp0}
\end{hyp}

We define the so-called Agmon distance attached to the two wells as
\begin{equation}
    \mathrm{S}=\text{min}\left\{ \mathrm{S}_u,\mathrm{S}_d \right\} ,
    \label{S}
\end{equation}
with ``up" and ``down" constants $\mathrm{S}_u$ and $\mathrm{S}_d$ defined by
\begin{equation}
\mathrm{S}_u=\displaystyle \int_{[s_r,s_l]}\gamma(s)^{\frac{1}{k+2}} \mathfrak{D}(s) ds\,\,\text{and}\,\,\,\mathrm{S}_d=\displaystyle \int_{[s_l,s_r]}\gamma(s)^{\frac{1}{k+2}} \mathfrak{D}(s) ds,
    \label{Sud}
\end{equation}
where, $\mathfrak{D}$ is a positive function defined on $\Gamma$ which will be defined later in \eqref{M}.

Let $L=\frac{\left|\Gamma\right|}{2}$. We define the two constants $\mathrm{A}_u$ and $\mathrm{A}_d$ by
\begin{equation}
\mathrm{A}_u:=\operatorname{exp}\left( -\int_{s_r}^{0}\Re\left(\frac{\mathfrak{V}_r'(s)+2\mathfrak{R}_r(s)-2\delta_{1,1}^{[k]}}{2\mathfrak{V}_r(s)}\right)ds \right),
    \label{Au}
\end{equation}
and
\begin{equation}
\mathrm{A}_d:=\operatorname{exp}\left( -\int_{s_l}^{L}\Re\left(\frac{\mathfrak{V}_l'(s)+2\mathfrak{R}_l(s)-2\delta_{1,1}^{[k]}}{2\mathfrak{V}_l(s)}\right)ds \right),
    \label{Ad}
\end{equation}
where $\delta_{1,1}^{[k]}$ is the second term of the asymptotic decomposition of the ground state energy (see Theorem \ref{BKW}), and the functions $\mathfrak{V}_r$, $\mathfrak{V}_l$, $\mathfrak{R}_r$ and $\mathfrak{R}_l$ are defined in remarks \ref{rmk} and \ref{rmk2}.

Let us state the main theorem of this paper, which gives an optimal estimate of the tunneling effect when the magnetic field vanishes along a curve $\Gamma$.
\begin{theorem}
Under Assumption \ref{Hyp0},  there exists $\varepsilon>0$ such that if 
$$\underset{s\in [-L,L]}{\operatorname{sup}}\left| 1-\frac{\gamma_0}{\gamma} \right|<\varepsilon,$$
then the difference between the first two eigenvalues of $\mathcal{L}_{\hbar}$ is given by
$$ \lambda_2(\hbar)-\lambda_1(\hbar)=2\left| \tilde{w}_{l,r} \right|+\mathrm{e}^{-\frac{\mathrm{S}}{\hbar^{1/(k+2)}}}\mathcal{O}(\hbar^2), $$
with
\begin{equation}
\tilde{w}_{l,r}=\zeta^{1/2}\pi^{-1/2}\hbar^{\frac{2k+3}{k+2}}\left(\overline{\mathfrak{V}_r(0)}\mathrm{A}_u\mathrm{e}^{-\frac{\mathrm{S_u}}{\hbar^{1/(k+2)}}}\mathrm{e}^{\mathrm{i}Lf(\hbar)}+\overline{\mathfrak{V}_r(-L)}\mathrm{A}_d\mathrm{e}^{-\frac{\mathrm{S_d}}{\hbar^{1/(k+2)}}}\mathrm{e}^{-L\mathrm{i}f(\hbar)} \right),
    \label{sum}
\end{equation}
where $\mathrm{S}=\operatorname{min}\left\{ \mathrm{S_u},\mathrm{S}_d \right\}$, $\zeta$ is a constant defined in \eqref{zeta}, and
\begin{enumerate}
    \item[1)]the function $\mathfrak{V}_r$ is introduced in remark \ref{rmk1};
    \item[2)]$\mathrm{A}_u$, $\mathrm{A}_d$ are defined in \eqref{Au}, \eqref{Ad} and $\mathrm{S_u}$, $\mathrm{S_d}$ are defined in \eqref{Sud};
    \item[3)]$f(\hbar)=\beta_0/\hbar- \hbar^{\frac{-1}{k+2}}\displaystyle \int_{-L}^{0}\gamma(s)^{\frac{1}{k+2}}\left( \xi_0^{[k]}-\Im \varphi_r(s) \right)ds-\alpha_0$, with
    \begin{enumerate}
        \item[(i)]the constant $\alpha_0$ defined in \eqref{alpha0};
        \item[(ii)]$\varphi_r$ an exact solution of the eikonal equation for the right well introduced in Lemma \ref{Lemme};
        \item[(ii)]the constant $\beta_0$ defined as
        \begin{equation}
            \beta_0:=\frac{1}{\left|\Gamma \right|}\displaystyle\int_{\Omega}B(x)dx,
            \label{beta0}
        \end{equation}
        where $\Omega$ is the open domain formed by the interior of $\Gamma$.
    \end{enumerate}
    
\end{enumerate}

\label{Thm0}
\end{theorem}
\begin{rmk}
We have two situations in this theorem:
\begin{enumerate}
    \item[1.] If $\mathrm{S}_u\neq\mathrm{S}_d$, only one term in the sum \eqref{sum} defining $\tilde{w}_{l,r}$ is predominent and $\tilde{w}_{l,r}$ is not zero for $\hbar$ small enough.
    \item[2.] If $\mathrm{S}_u=\mathrm{S}_d$, the situation is different: due to the circulation, the interaction term $\tilde{w}_{l,r}$ can vanish for some parameters $\hbar$ and in this case the spectral gap is of order $\mathcal{O}\left( \hbar^2 \mathrm{e^{-\frac{\mathrm{S}}{\hbar^{1/(k+2)}}}} \right)$.
\end{enumerate}
\label{rmk3}
\end{rmk}
The second case in remark \ref{rmk3} occurs when the magnetic field is symmetric with respect to the $x_1$-axis, i.e 
\begin{equation}
B(x_1,x_2)=B(x_1,-x_2)\,\,\,\,\,\text{for all}\,\,\,\,\,\,\,(x_1,x_2)\in\mathbb{R}^2.
    \label{B}
\end{equation}
In this case, we have
$$ \mathrm{A}_u=\mathrm{A}_d:=\mathrm{A}\,\,,\,\,\,\,\,\mathfrak{V}_r(0)=\mathfrak{V}_r(-L):=\mathfrak{V}_0\,\,\,\,\text{and}\,\,\,\,\,\mathrm{S}_u=\mathrm{S}_d:=\mathrm{S}, $$
and we get the following corollary.
\begin{cor}
If $B$ verifies \eqref{B} and under Assumption \ref{Hyp0},  there exists $\varepsilon>0$ such that if
$$\underset{s\in [-L,L]}{\operatorname{sup}}\left| 1-\frac{\gamma_0}{\gamma} \right|<\varepsilon,$$
then the difference between the first two eigenvalues of $\mathcal{L}_{\hbar}$ is given by
$$ \lambda_2(\hbar)-\lambda_1(\hbar)=4\zeta^{1/2}\pi^{-1/2}\hbar^{\frac{2k+3}{k+2}}\left| \mathfrak{V}_0 \right|\mathrm{A}\mathrm{e}^{-\frac{\mathrm{S}}{\hbar^{1/(k+2)}}}\left| \cos \left( f(\hbar) \right) \right|+\hbar^2\mathcal{O}\left( \mathrm{e^{-\frac{\mathrm{S}}{\hbar^{1/(k+2)}}}} \right). $$
\end{cor}
\subsection{Organization of the paper}
In Section \ref{section2}, we explain the spectral reduction scheme, using Agmon's normal estimates and tubular coordinates in the neighborhood of the zero curve $\Gamma$, which allows us to replace the operator $\mathcal{L}_{\hbar}$ by the rescaled operator $\mathcal{N}_h^{[k]}$ with a new semi-classical parameter $h=\hbar^{\frac{1}{k+2}}$. The localization near $\Gamma$ allows us to reduce to the study a straight model, and we introduce reduced left and right ``one well'' models (see Section \ref{section3}). In Section \ref{section4}, we construct the WKB expansions for the ground state of the ``right well'' operator $\mathcal{N}_{h,r}^{[k]}$. In Section \ref{section5}, we conjugate by an exponential and reduce the dimension (at least formally) using a Grushin method. We then choose the exponential weight as a perturbation of  the solution of the eikonal equation and, to keep ellipticity, we have to use the hypothesis of ``soft" variation of the function $\gamma$. With these assumptions, the Agmon weight is uniformly controlled and we are reduced to a perturbation problem near the minimum of the Montgomery operator. To ensure that the frequency variable $\xi$ is bounded, we truncate this variable in a neighborhood of $\gamma_0^{1/(k+2)}\xi_0^{[k]}$ and consider the operator with truncated symbol $\text{Op}_h^\text{w}p_h$. Using the Grushin reduction method, we show tangential coercivity (see Theorem \ref{Theo}) following \cite{keraval}. In Section \ref{section6}, we prove Theorem \ref{theo2}. It consists in particular in removing the cutoff function which was introduced in Section \ref{section5}. In Section \ref{Section 7}, we show optimal tangential estimates using Theorem \ref{theo2} (see Corollary \ref{single}). We also establish tangential estimates for the double-well operator $\mathcal{N}_h^{[k]}$ (see Proposition \ref{double}), and establish WKB approximations of the first eigenfunctions of operator $\mathcal{N}_h^{[k]}$ (see Proposition \ref{App}). In Section \ref{section9}, we prove Theorem \ref{Thm0}. WKB approximations allow the analysis of an interaction matrix whose eigenvalues measure the tunneling effect.

\section{A reduction to a tubular neighborhood of the cancellation curve}
\label{section2}
The following Agmon estimates can be found in \cite[Proposition 5.1]{helffer1996semiclassical}. These estimates show the exponential localization of the eigenfunctions of $\mathcal{L}_{\hbar}$ near the zero curve.

\begin{prop}
Let $E>0$. There exist $C,\hbar_0,\alpha>0$ such that, for all $\hbar\in (0,\hbar_0)$, and all eigenpairs $(\lambda,\psi)$ of $\mathcal{L}_{\hbar}$ with $\lambda  \leq E\hbar^{2\frac{k+1}{k+2}}$,
$$ \int_{\mathbb{R}^2}\mathrm{e}^{2\frac{\alpha \text{dist}(x,\Gamma)^{\frac{k+2}{2}}}{\sqrt{\hbar}}}\left| \psi \right|^2 dx \leq C\left\|\psi\right\|^2, $$
and
$$ \int_{\mathbb{R}^2}\mathrm{e}^{2\frac{\alpha \text{dist}(x,\Gamma)^{\frac{k+2}{2}}}{\sqrt{\hbar}}}\left| \left( -\mathrm{i}\hbar\nabla+\mathbf{A} \right) \psi \right|^2 dx\leq C \hbar^{2\frac{k+1}{k+2}} \left\|\psi\right\|^2. $$
\label{estimations normale d'agmon}
\end{prop}
Since the first eigenfunctions are concentrated in the neighborhood of $\Gamma$ then we can deduce that we can work in a small neighborhood of $\Gamma$ of size $\delta$ small enough. For this reason, we consider the $\delta$-neighborhood of the curve $\Gamma$
$$ \Omega_{\delta}:=\left\{ x\in\mathbb{R}^2:\text{dist}(x,\Gamma)<\delta \right\}. $$
Here $\delta$ normally depends on $\hbar$ which we will specify later. We consider the quadratic form $Q_{\hbar,\delta}$ defined for all $\psi\in \mathcal{V}_{\delta}=H_0^1\left( \Omega_{\delta}\right)$,
$$ Q_{\hbar,\delta}=\int_{\Omega_{\delta}}\left| \left( -\mathrm{i}\hbar\nabla+\mathbf{A} \right) \psi \right|^2dx. $$
The associated self-adjoint operator is
$$ \mathcal{L}_{\hbar,\delta}=\left( -\mathrm{i}\hbar\nabla+\mathbf{A} \right)^2, $$
with domain
$$ \text{Dom}\left( \mathcal{L}_{\hbar,\delta} \right)=\left\{ \psi\in H^2\left( \Omega_{\delta} \right):\psi(x)=0,\,\,\,\text{on}\,\, \,\,\left\{ x\in\mathbb{R}^2 : \text{dist}\left( x,\Gamma \right)=\delta\right\} \right\}. $$
This operator is self-adjoint with compact resolvent, and we can consider the non-decreasing sequence of eigenvalues $ \left( \lambda_n(\hbar,\delta) \right)_{n\geq 1} $. We will follow the same reduction strategy as \cite{BHR-purely}. 
\begin{prop}
Let $n\geq 1$. There exist $C,\hbar_0,\beta>0$ such that, for all $\hbar\in(0,\hbar_0)$ and $\delta\in(0,\delta_0)$,
$$ \lambda_n(\hbar)\leq\lambda_n(\hbar,\delta)\leq \lambda_n(\hbar)+C\mathrm{e}^{-\frac{\beta\delta^{\frac{k+2}{2}}}{\sqrt{\hbar}}}. $$
\label{R1}
\end{prop}
\begin{proof}
The proof is similar to that of \cite{BHR-purely}, except the power of $\hbar$. We first then prove first inequality. Let $\psi_n\in\mathcal{V}_{\delta}$ the eigenfunction associated with $\lambda_n(\hbar,\delta)$. Since $\psi_n=0$ on $\left\{ x\in\mathbb{R}^2:\text{dist}(x,\Gamma)=\delta \right\}$, we can extend it by $0$ on $\mathbb{R}^2$ to obtain a function $\tilde{\psi_n}$ defined on $\mathbb{R}^2$ which satisfies
$$ Q_{\hbar}\left( \tilde{\psi}_n \right)=Q_{\hbar,\delta}\left( \psi_n \right)=\lambda_n(\hbar,\delta). $$
Then by min-max principle $\lambda_n(\hbar)\leq \lambda_n(\hbar,\delta)$.

We now show the second inequality in proposition \ref{R1}. Let $\left( \psi_j \right)_{1\leq j \leq n }$ be an orthonormal family of eigenfunctions associated with $\left( \lambda_j(\hbar) \right)_{1\leq j \leq n}$ and let
$$ \chi_{\delta}(x):=\chi\left( \frac{\text{dist}(x,\Gamma)}{\delta} \right), $$
where $\chi$ is a smooth cut off function, which is equal to $0$ on $[1,+\infty[$, and is equal to $1$ on $[0,1/2[$. We define
$$ \mathcal{E}(\hbar,\delta):=\underset{1\leq j \leq n}{\text{Span}}\chi_{\delta}\psi_j\subset \mathcal{V}_{\delta}. $$
Let $\tilde{\psi}$ be a function of $\mathcal{E}(\hbar,\delta)$. This function is written in the forme
$$ \tilde{\psi}=\chi_{\delta}\sum_{j=1}^n\beta_j\psi_j=\chi_{\delta}\psi. $$
We have
\begin{align*}
Q_{\hbar,\delta}\left( \chi_{\delta}\psi \right)&=\int_{\Omega_{\hbar,\delta}}\left| \left( -\mathrm{i}\hbar\nabla+\mathbf{A} \right)\left( \chi_{\delta}\psi \right) \right|^2dx\\&\leq \left\| \left( -\mathrm{i}\hbar\nabla+\mathbf{A} \right)\psi \right\|^2+2\hbar\left\| \left( -\mathrm{i}\hbar\nabla+\mathbf{A} \right)\psi \right\|_{L^2\left( \mathbb{R}^2\setminus\Omega_{\delta/2} \right)}\left\| \left| \nabla\chi_{\delta} \right|\psi  \right\|+\hbar^2\left\| \left| \nabla\chi_{\delta}\right|\psi \right\|^2.
\end{align*}
Since the family $\left( \psi_j \right)_{1\leq j\leq n}$ is orthogonal, then
$$ \langle \left(-\mathrm{i}\hbar\nabla+\mathbf{A}\right)\psi_j,\left(-\mathrm{i}\hbar\nabla+\mathbf{A}\right)\psi_k\rangle=0,\,\,\,\forall j\neq k, $$
which implies
$$ \left\| \left( -\mathrm{i}\hbar\nabla+\mathbf{A} \right)\psi \right\|^2\leq \lambda_n(\hbar)\left\|\psi\right\|^2. $$
Using Proposition \ref{estimations normale d'agmon}, we have
$$ \left\| \left| \nabla\chi_{\delta} \right|\psi  \right\|\leq C\delta^{-1}  \mathrm{e}^{-\frac{\alpha\left(\frac{\delta}{2}\right)^{(k+2)/2}}{\sqrt{\hbar}}}\left\| \psi \right\|, $$
and
$$ \left\| \left( -\mathrm{i}\hbar\nabla+\mathbf{A} \right)\psi \right\|_{L^2\left( \mathbb{R}^2\setminus\Omega_{\delta/2} \right)}\leq C\hbar^{\frac{k+1}{k+2}}  \mathrm{e}^{-\frac{\alpha\left(\frac{\delta}{2} \right)^{(k+2)/2}}{\sqrt{\hbar}}}\left\| \psi \right\|. $$
Therefore,
$$ Q_{\hbar,\delta}\left(\tilde{\psi}\right)\leq\left( \lambda_n(\hbar)+C\left( \hbar^{\frac{2k+3}{k+2}}\delta^{-1}+\hbar^2\delta^{-2} \right)  \mathrm{e}^{-\frac{\beta\delta^{\frac{k+2}{2}}}{\sqrt{\hbar}}} \right)\left\|\tilde{\psi}\right\|^2,\,\,\,\,\,\,\forall\tilde{\psi}\in\mathcal{E}_n(\hbar,\delta), $$
with $\beta=\frac{\alpha}{2^{k/2}}$. Then we get
$$ \lambda_n(\hbar,\delta)\leq\lambda_n(\hbar)+  C\mathrm{e}^{-\frac{\beta\delta^{\frac{k+2}{2}}}{\sqrt{\hbar}}}.$$
\end{proof}

Proposition \ref{R1} allows to replace the initial operator $\mathcal{L}_{\hbar}$ by the operator $\mathcal{L}_{\hbar,\delta}$ with dirichlet conditions in a $\delta$-neighborhood of the curve $\Gamma$. We will make a change of coordinates in the neighborhood of the zero curve $\Gamma$. This change of coordinates can be found in detail in \cite{fournais2010spectral} and we see all it here. Let
$$ M:\mathbb{R}/(\left|\Gamma\right|\mathbb{Z} )\ni s\longmapsto M(s)\in \Gamma $$
be the arc-length parametrization of $\Gamma$ (see figure \ref{Tubular coordinates}) so that 
$$ \Gamma\cap\left\{(x,y)\in \mathbb{R}^2:x=0 \right\}=\left\{ M(0):=(0,y_0),M(L):=(0,y_1) \right\} \,\,\,\,\,\,\text{with}\,\,\,\,\,\,y_1<y_0. $$
Let $\nu(s)$ the unit normal to $\Gamma$ at the point $M(s)$. We choose the orientation of the parametrization $M$ so that
$$ det\left( M'(s),\nu(s) \right)=1. $$
The curvature $\kappa(s)$ of $\Gamma$ at point $M(s)$ is given by the parametrization
$$ M''(s)=\kappa(s)\nu(s). $$
Since we are working with $2L$-periodic functions, then we can consider the restriction of these functions on the interval $]-L,+L]$.

We consider the function $\Phi:\mathbb{R}/\left( \left|\Gamma\right|\mathbb{Z}\right)\times(-\delta_0,\delta_0)\longrightarrow\Omega_{\delta_0}$ defined by
\begin{equation}
    \Phi(s,t)=M(s)+t\nu(s),\,\,\,\,\forall(s,t)\in\mathbb{R}/\left( \left|\Gamma\right|\mathbb{Z}\right)\times(-\delta_0,\delta_0),
    \label{Phii}
\end{equation}
where $t=\text{dist}(x,\Gamma)$ and $\delta_0>0$ small enough, so that $\Phi$ is diffeomorphism with image $\Omega_{\delta_0}:=\left\{ x\in \mathbb{R}^2:\text{dist}(x,\Gamma)<\delta_0) \right\}$.

\begin{figure}[ht]

\begin{center}

\definecolor{tttttt}{rgb}{0.2,0.2,0.2}
\definecolor{qqqqff}{rgb}{0,0,1}
\definecolor{qqwuqq}{rgb}{0,0.39,0}
\definecolor{ttqqcc}{rgb}{0.2,0,0.8}
\definecolor{uququq}{rgb}{0.25,0.25,0.25}
\begin{tikzpicture}[line cap=round,line join=round,>=triangle 45,x=1.0cm,y=1.0cm]
\clip(-5.35,-4.37) rectangle (8.45,3.71);
\draw[color=qqwuqq,fill=qqwuqq,fill opacity=0.1] (2.03,1.48) -- (2.22,1.51) -- (2.19,1.7) -- (2,1.67) -- cycle; 
\draw [shift={(-2.76,0)},dash pattern=on 2pt off 2pt]  plot[domain=2.43:3.14,variable=\t]({1*1.12*cos(\t r)+0*1.12*sin(\t r)},{0*1.12*cos(\t r)+1*1.12*sin(\t r)});
\draw [shift={(-2.59,-0.17)},dash pattern=on 2pt off 2pt]  plot[domain=1.79:2.42,variable=\t]({1*1.37*cos(\t r)+0*1.37*sin(\t r)},{0*1.37*cos(\t r)+1*1.37*sin(\t r)});
\draw [shift={(-2.45,-0.79)},dash pattern=on 2pt off 2pt]  plot[domain=1.39:1.79,variable=\t]({1*2*cos(\t r)+0*2*sin(\t r)},{0*2*cos(\t r)+1*2*sin(\t r)});
\draw [shift={(-3,-3.49)},dash pattern=on 2pt off 2pt]  plot[domain=1.18:1.38,variable=\t]({1*4.75*cos(\t r)+0*4.75*sin(\t r)},{0*4.75*cos(\t r)+1*4.75*sin(\t r)});
\draw [shift={(0.18,3.92)},dash pattern=on 2pt off 2pt]  plot[domain=4.29:4.66,variable=\t]({1*3.32*cos(\t r)+0*3.32*sin(\t r)},{0*3.32*cos(\t r)+1*3.32*sin(\t r)});
\draw [shift={(-0.18,3.92)},dash pattern=on 2pt off 2pt]  plot[domain=4.77:5.13,variable=\t]({1*3.32*cos(\t r)+0*3.32*sin(\t r)},{0*3.32*cos(\t r)+1*3.32*sin(\t r)});
\draw [shift={(3,-3.49)},dash pattern=on 2pt off 2pt]  plot[domain=1.76:1.96,variable=\t]({1*4.75*cos(\t r)+0*4.75*sin(\t r)},{0*4.75*cos(\t r)+1*4.75*sin(\t r)});
\draw [shift={(2.45,-0.79)},dash pattern=on 2pt off 2pt]  plot[domain=1.35:1.76,variable=\t]({1*2*cos(\t r)+0*2*sin(\t r)},{0*2*cos(\t r)+1*2*sin(\t r)});
\draw [shift={(2.59,-0.17)},dash pattern=on 2pt off 2pt]  plot[domain=0.72:1.35,variable=\t]({1*1.37*cos(\t r)+0*1.37*sin(\t r)},{0*1.37*cos(\t r)+1*1.37*sin(\t r)});
\draw [shift={(2.59,0.17)},dash pattern=on 2pt off 2pt]  plot[domain=4.93:5.57,variable=\t]({1*1.37*cos(\t r)+0*1.37*sin(\t r)},{0*1.37*cos(\t r)+1*1.37*sin(\t r)});
\draw [shift={(2.45,0.79)},dash pattern=on 2pt off 2pt]  plot[domain=4.53:4.93,variable=\t]({1*2*cos(\t r)+0*2*sin(\t r)},{0*2*cos(\t r)+1*2*sin(\t r)});
\draw [shift={(3,3.49)},dash pattern=on 2pt off 2pt]  plot[domain=4.32:4.52,variable=\t]({1*4.75*cos(\t r)+0*4.75*sin(\t r)},{0*4.75*cos(\t r)+1*4.75*sin(\t r)});
\draw [shift={(-0.18,-3.92)},dash pattern=on 2pt off 2pt]  plot[domain=1.15:1.52,variable=\t]({1*3.32*cos(\t r)+0*3.32*sin(\t r)},{0*3.32*cos(\t r)+1*3.32*sin(\t r)});
\draw [shift={(0.18,-3.92)},dash pattern=on 2pt off 2pt]  plot[domain=1.63:1.99,variable=\t]({1*3.32*cos(\t r)+0*3.32*sin(\t r)},{0*3.32*cos(\t r)+1*3.32*sin(\t r)});
\draw [shift={(-3,3.49)},dash pattern=on 2pt off 2pt]  plot[domain=4.91:5.11,variable=\t]({1*4.75*cos(\t r)+0*4.75*sin(\t r)},{0*4.75*cos(\t r)+1*4.75*sin(\t r)});
\draw [shift={(-2.45,0.79)},dash pattern=on 2pt off 2pt]  plot[domain=4.49:4.9,variable=\t]({1*2*cos(\t r)+0*2*sin(\t r)},{0*2*cos(\t r)+1*2*sin(\t r)});
\draw [shift={(-2.59,0.17)},dash pattern=on 2pt off 2pt]  plot[domain=3.86:4.49,variable=\t]({1*1.37*cos(\t r)+0*1.37*sin(\t r)},{0*1.37*cos(\t r)+1*1.37*sin(\t r)});
\draw [shift={(-2.77,0)},dash pattern=on 2pt off 2pt]  plot[domain=3.15:3.85,variable=\t]({1*1.12*cos(\t r)+0*1.12*sin(\t r)},{0*1.12*cos(\t r)+1*1.12*sin(\t r)});
\draw [shift={(-2.76,0)},dash pattern=on 2pt off 2pt]  plot[domain=2.44:3.14,variable=\t]({1*2.12*cos(\t r)+0*2.12*sin(\t r)},{0*2.12*cos(\t r)+1*2.12*sin(\t r)});
\draw [shift={(-2.59,-0.17)},dash pattern=on 2pt off 2pt]  plot[domain=1.8:2.43,variable=\t]({1*2.36*cos(\t r)+0*2.36*sin(\t r)},{0*2.36*cos(\t r)+1*2.36*sin(\t r)});
\draw [shift={(-2.45,-0.79)},dash pattern=on 2pt off 2pt]  plot[domain=1.39:1.79,variable=\t]({1*3*cos(\t r)+0*3*sin(\t r)},{0*3*cos(\t r)+1*3*sin(\t r)});
\draw [shift={(-2.98,-3.44)},dash pattern=on 2pt off 2pt]  plot[domain=1.18:1.38,variable=\t]({1*5.69*cos(\t r)+0*5.69*sin(\t r)},{0*5.69*cos(\t r)+1*5.69*sin(\t r)});
\draw [shift={(0.21,4.02)},dash pattern=on 2pt off 2pt]  plot[domain=4.27:4.63,variable=\t]({1*2.42*cos(\t r)+0*2.42*sin(\t r)},{0*2.42*cos(\t r)+1*2.42*sin(\t r)});
\draw [shift={(-0.21,4.02)},dash pattern=on 2pt off 2pt]  plot[domain=4.8:5.15,variable=\t]({1*2.42*cos(\t r)+0*2.42*sin(\t r)},{0*2.42*cos(\t r)+1*2.42*sin(\t r)});
\draw [shift={(2.98,-3.44)},dash pattern=on 2pt off 2pt]  plot[domain=1.76:1.96,variable=\t]({1*5.69*cos(\t r)+0*5.69*sin(\t r)},{0*5.69*cos(\t r)+1*5.69*sin(\t r)});
\draw [shift={(2.45,-0.79)},dash pattern=on 2pt off 2pt]  plot[domain=1.35:1.75,variable=\t]({1*3*cos(\t r)+0*3*sin(\t r)},{0*3*cos(\t r)+1*3*sin(\t r)});
\draw [shift={(2.59,-0.17)},dash pattern=on 2pt off 2pt]  plot[domain=0.71:1.35,variable=\t]({1*2.36*cos(\t r)+0*2.36*sin(\t r)},{0*2.36*cos(\t r)+1*2.36*sin(\t r)});
\draw [shift={(2.59,0.17)},dash pattern=on 2pt off 2pt]  plot[domain=4.94:5.57,variable=\t]({1*2.36*cos(\t r)+0*2.36*sin(\t r)},{0*2.36*cos(\t r)+1*2.36*sin(\t r)});
\draw [shift={(2.45,0.79)},dash pattern=on 2pt off 2pt]  plot[domain=4.53:4.94,variable=\t]({1*3*cos(\t r)+0*3*sin(\t r)},{0*3*cos(\t r)+1*3*sin(\t r)});
\draw [shift={(2.98,3.44)},dash pattern=on 2pt off 2pt]  plot[domain=4.32:4.52,variable=\t]({1*5.69*cos(\t r)+0*5.69*sin(\t r)},{0*5.69*cos(\t r)+1*5.69*sin(\t r)});
\draw [shift={(-0.21,-4.02)},dash pattern=on 2pt off 2pt]  plot[domain=1.13:1.49,variable=\t]({1*2.42*cos(\t r)+0*2.42*sin(\t r)},{0*2.42*cos(\t r)+1*2.42*sin(\t r)});
\draw [shift={(0.21,-4.02)},dash pattern=on 2pt off 2pt]  plot[domain=1.66:2.01,variable=\t]({1*2.42*cos(\t r)+0*2.42*sin(\t r)},{0*2.42*cos(\t r)+1*2.42*sin(\t r)});
\draw [shift={(-2.98,3.44)},dash pattern=on 2pt off 2pt]  plot[domain=4.9:5.1,variable=\t]({1*5.69*cos(\t r)+0*5.69*sin(\t r)},{0*5.69*cos(\t r)+1*5.69*sin(\t r)});
\draw [shift={(-2.45,0.79)},dash pattern=on 2pt off 2pt]  plot[domain=4.49:4.89,variable=\t]({1*3*cos(\t r)+0*3*sin(\t r)},{0*3*cos(\t r)+1*3*sin(\t r)});
\draw [shift={(-2.59,0.17)},dash pattern=on 2pt off 2pt]  plot[domain=3.85:4.49,variable=\t]({1*2.36*cos(\t r)+0*2.36*sin(\t r)},{0*2.36*cos(\t r)+1*2.36*sin(\t r)});
\draw [shift={(-2.77,0)},dash pattern=on 2pt off 2pt]  plot[domain=3.14:3.85,variable=\t]({1*2.12*cos(\t r)+0*2.12*sin(\t r)},{0*2.12*cos(\t r)+1*2.12*sin(\t r)});
\draw [->,dash pattern=on 2pt off 2pt] (2.28,0) -- (1.82,2.73);
\draw [->,dash pattern=on 2pt off 2pt] (-2,1.67) -- (-1.92,2.16);
\draw [->,dash pattern=on 2pt off 2pt] (-2,1.67) -- (-2.08,1.17);
\draw [shift={(2.78,0)},dash pattern=on 2pt off 2pt]  plot[domain=-0.71:0.71,variable=\t]({1*2.11*cos(\t r)+0*2.11*sin(\t r)},{0*2.11*cos(\t r)+1*2.11*sin(\t r)});
\draw [shift={(2.78,0)},dash pattern=on 2pt off 2pt]  plot[domain=-0.71:0.71,variable=\t]({1*1.11*cos(\t r)+0*1.11*sin(\t r)},{0*1.11*cos(\t r)+1*1.11*sin(\t r)});
\draw[smooth,samples=100,domain=-4.3:4.3] plot(\x,{(-(\x)^2-9+(36*(\x)^2+3.2^4)^0.5)^0.5});
\draw[smooth,samples=100,domain=-4.3:4.3] plot(\x,{0-(-(\x)^2-9+(36*(\x)^2+3.2^4)^0.5)^0.5});
\draw [shift={(2.79,0)}] plot[domain=-0.33:0.33,variable=\t]({1*1.59*cos(\t r)+0*1.59*sin(\t r)},{0*1.59*cos(\t r)+1*1.59*sin(\t r)});
\draw [shift={(-2.79,0)}] plot[domain=2.81:3.47,variable=\t]({1*1.59*cos(\t r)+0*1.59*sin(\t r)},{0*1.59*cos(\t r)+1*1.59*sin(\t r)});
\draw (2.78,-1.64) node[anchor=north west] {$ \Gamma $};
\draw (4.01,1.41) node[anchor=north west] {$ s_r $};
\draw (-4.43,1.46) node[anchor=north west] {$ s_l $};
\draw (2.29,0.18) node[anchor=north west] {$ -\infty $};
\draw (1.88,2.89) node[anchor=north west] {$ +\infty $};
\draw [line width=1.2pt,dash pattern=on 2pt off 2pt,color=ttqqcc] (0,3.07)-- (0,-3.07);
\draw [dash pattern=on 2pt off 2pt] (0.94,1.49)-- (3.06,1.85);
\draw (-1.84,2.09) node[anchor=north west] {$ \Omega_{\delta} $};
\draw (-0.05,1.6) node[anchor=north west] {$ 0 $};
\draw [shift={(3.87,0.51)},dash pattern=on 2pt off 2pt,color=qqqqff]  plot[domain=0.13:0.82,variable=\t]({1*1.43*cos(\t r)+0*1.43*sin(\t r)},{0*1.43*cos(\t r)+1*1.43*sin(\t r)});
\draw [->,color=qqqqff] (4.95,1.45) -- (4.85,1.55);
\draw [color=tttttt](0,2.87) node[anchor=north west] {$ x_1=0 $};
\draw [color=qqqqff](5.2,1.39) node[anchor=north west] {+};
\draw (1.94,2.27) node[anchor=north west] {$ M(s) $};
\draw (0.02,-0.73) node[anchor=north west] {$ -L $};
\draw (-0.82,-0.73) node[anchor=north west] {$ +L $};
\begin{scriptsize}
\fill [color=uququq] (2,1.67) circle (1.5pt);
\fill [color=uququq] (4,1.05) circle (1.5pt);
\fill [color=uququq] (-4,1.05) circle (1.5pt);
\fill [color=uququq] (0,1.11) circle (1.5pt);
\fill [color=uququq] (0,-1.11) circle (1.5pt);
\end{scriptsize}
\end{tikzpicture}

\caption{Tubular coordinates in the neighborhood of $\Gamma$}
\label{Tubular coordinates}
\end{center}
\end{figure}
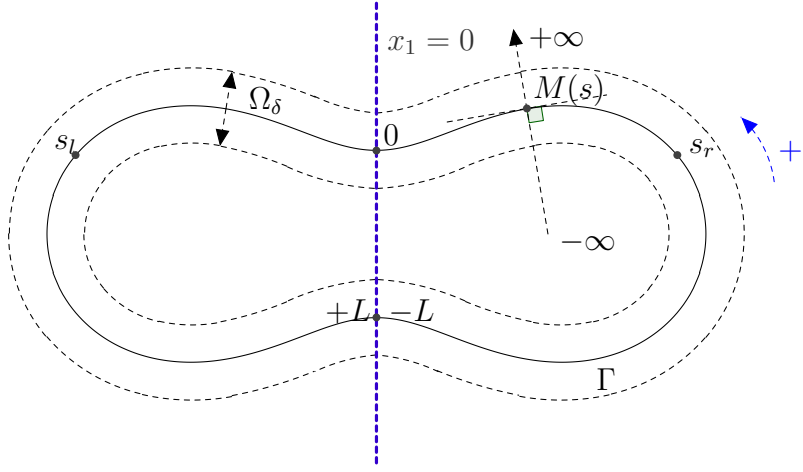

With the change of coordinates $\Phi^{-1}$ defined in , the determinant of the Jacobian matrix of this transformation is given by 
$$ m(s,t)=1-t\kappa(s), $$
and the quadratic form $Q_{\hbar,\delta}$ can be rewritten as follows
\begin{align*}
Q_{\hbar,\delta}(u)&=\int_{\Omega_{\delta}}\left| \left( -\mathrm{i}\hbar\nabla+\mathbf{A} \right)u \right|^2dx\\&=\int_{\Phi^{-1}(\Omega_{\delta})}\left\{ \left( 1-t\kappa(s) \right)^{-2}\left| \left( -\mathrm{i}\hbar\partial_s+\overline{A}_1 \right)v \right|^2+ \left| \left( -\mathrm{i}\hbar\partial_t+\overline{A}_2 \right)v \right|^2 \right\}(1-t\kappa(s))dsdt,
\end{align*}
and
$$ \int_{\Omega_{\delta}}\left| u(x) \right|^2dx=\int_{\Phi^{-1}(\Omega_{\delta})}\left| v(s,t) \right|^2(1-t\kappa(s))dsdt, $$
for all $u\in \mathcal{V}_{\delta}$, with $v=u\circ\Phi$ and 
$$ \overline{A}_1(s,t)=(1-t\kappa(s))\left( \mathbf{A} \circ\Phi\right).M'(s),\,\,\,\,\,\,\overline{A}_2(s,t)=\left( \mathbf{A} \circ\Phi\right).\nu(s). $$
The magnetic field associated with the new magnetic potential $\overline{\mathbf{A}}$ is given by
\begin{align*}
    \beta(s,t)&:=\nabla_{(s,t)}\times\overline{\mathbf{A}}(s,t)\\&=m(s,t)\left( \nabla\times \mathbf{A} \right)\circ\Phi(s,t)\\&=m(s,t)B\circ\Phi(s,t).
\end{align*}
To eliminate the normal component of $\overline{\mathbf{A}}$, we now use the gauge transformation which corresponds to the conjugation of the operator by $\mathrm{e}^{\mathrm{i}\frac{\phi}{\hbar}}$, with $\phi$ is given by
$$ \phi(s,t)=-\beta_0 s+\int_0^t\overline{A}_2(s,t')dt'+\int_0^s\overline{A}_1(s',0)ds', $$
where $\beta_0$ is defined in \eqref{beta0}.

The new magnetic potential is given by $\tilde{\mathbf{A}}(s,t)=\overline{\mathbf{A}}(s,t)-\nabla_{(s,t)}\phi$. Then for all $u\in\mathcal{V}_{\delta}$, we have
\begin{align*}
Q_{\hbar,\delta}(u)&=\int_{\Omega_{\delta}}\left| \left( -\mathrm{i}\hbar\nabla+\mathbf{A} \right)u \right|^2dx\\&=\int_{\Phi^{-1}(\Omega_{\delta})}\left\{ \left( 1-t\kappa(s) \right)^{-2}\left| \left( -\mathrm{i}\hbar\partial_s+\tilde{A}_1 \right)w \right|^2+ \left| \left( -\mathrm{i}\hbar\partial_t \right)w \right|^2 \right\}(1-t\kappa(s))dsdt,
\end{align*}
where $w=\mathrm{e}^{\mathrm{i}\frac{\phi}{\hbar}}v$ and $v=u\circ\Phi$.\\
After this change of coordinates, the operator $\mathcal{L}_{\hbar,\delta}$ is unitarily equivalent to $\tilde{\mathcal{L}}_{\hbar,\delta}$ the
self-adjoint realization on $L^2\left( \Gamma\times (-\delta,\delta);m(s,t)dsdt \right)$, of the differential operator
$$ (1-t\kappa(s))^{-1}\hbar D_t (1-t\kappa(s)) \hbar D_t+(1-t\kappa(s))^{-1}\left( \hbar D_s+\tilde{A}_1(s,t) \right)(1-t\kappa(s))^{-1}\left( \hbar D_s+\tilde{A}_1(s,t) \right), $$
where $D=\frac{1}{\mathrm{i}}\partial$ and
\begin{equation}
\tilde{A}_1(s,t)=\beta_0-\int_0^tm(s,t')B\circ\Phi(s,t')dt'=\beta_0-\int_0^t(1-t'\kappa(s))B\circ\Phi(s,t')dt',
\label{POT1}
\end{equation}
with boundary dirichlet conditions.

Using Assumption \ref{Hyp0}, magnetic field $B$ vanishes exactly at order $k\geq 1$ on $\Gamma$, with $B$ and the $(k-1)$ first normal derivatives of $B$ vanish on $\Gamma$. So
$$ B\circ\Phi,\partial_t\left(B\circ\Phi\right),\partial_{t^2}^2\left(B\circ\Phi\right),...,\partial_{t^{k-1}}^{k-1}\left(B\circ\Phi\right) $$
vanish at $t=0$.\\
Since we work for $t$ small enough ($-\delta< t <\delta$), writing the asymptotic expansion of $B\circ\Phi$ near $t=0$ (for $s$ fixed) gives
$$ B\circ\Phi(s,t)=\frac{t^k}{k!}\left( \partial_{t^k}^k\left(B\circ\Phi\right)(s,0) \right)+\frac{t^{k+1}}{{(k+1)}!}\left( \partial_{t^{k+1}}^{k+1}\left(B\circ\Phi\right)(s,0) \right)+\mathcal{O}(t^{k+2}). $$
We recall and define
\begin{equation}
    \gamma(s):=\frac{1}{k!}\left( \partial_{t^k}^k\left(B\circ\Phi\right)(s,0) \right)\,\,\,\,\,and\,\,\,\,\,\,\delta(s):=\frac{1}{{(k+1)}!}\left( \partial_{t^{k+1}}^{k+1}\left(B\circ\Phi\right)(s,0) \right).
    \label{gamma}
\end{equation}

Using Assumption \ref{Hyp0}, the function $s\mapsto \gamma(s)$ has a non-degenerate minimum $\gamma_0>0$ at $s=s_r<0$ and $s=s_l=-s_r>0$, with
$$M(s_r)=a_1,\,\,\,\,M(s_l)=a_2,\,\,\,\,-L<s_r<0,\,\,\,\,0<s_l<+L,$$
and
$$ \gamma(s_l)=\gamma(s_r)=\gamma_0,\,\,\,\,\,\gamma'(s_l)=\gamma'(s_r)=0,\,\,\,\,\,\gamma''(s_l),\gamma''(s_r)>0. $$
By computing the integral in \eqref{POT1}, expression of the magnetic potential $\tilde{A}_1$ is given by
$$ \tilde{A}_1(s,t)=\beta_0-\gamma(s)\frac{t^{k+1}}{k+1}-\tilde{\delta}(s)\frac{t^{k+2}}{k+2}+\mathcal{O}(t^{k+3}), $$
where $\tilde{\delta}(s)=\delta(s)-\gamma(s)\kappa(s)$.
\subsection{Truncated operator and rescaled operator}
In this section we follow the same spectral reduction method as in \cite{BHR-purely}. First, we truncate the variable $t$ to work on the domain $]-L,+L]\times\mathbb{R}$ instead of $]-L,+L]\times(-\delta,\delta)$. After the truncation, we use the fact that the first eigenfunctions of operator $\tilde{\mathcal{L}}_{\hbar,\delta}$ decay exponentially away from the cancellation curve $\Gamma$ at the length scale $h=\hbar
^{\frac{1}{k+2}}$. This localization allows us to consider the partial rescaling $(s,t)=(\sigma,h\tau)$ with $h=\hbar^{\frac{1}{k+2}}$.

We start by truncating in the variable $t$. Let $c$ be a smooth truncation function equal to $1$ on $[-1,1]$ and $0$ for $\left| t\right|\geq 2$.

We define
$$ \underline{m}(s,t)=1-tc\left( \frac{t}{\delta}\right)\kappa, $$
and
$$ \underline{A}(s,t)=\beta_0-\gamma(s)\frac{t^{k+1}}{k+1}-\tilde{\delta}(s)c\left(\frac{t}{\delta}\right)\frac{t^{k+2}}{k+2}+c\left(\frac{t}{\delta}\right)\mathcal{O}\left(t^{k+3}\right). $$

We introduced here the truncation function $c$ to ensure that the terms are bounded when $t$ is large. This truncation function is found only in front of $t^{k+2}$ and $t^{k+3}$ in $\underline{A}(s,t)$. Then, we define $\underline{\mathcal{M}}_{\hbar,\delta}$ as self-adjoint realization on the space $L^2\left( \Gamma\times\mathbb{R},\underline{m}(s,t)dsdt \right)$ of the differential operator
$$ \underline{m}^{-1}\hbar D_t\underline{m}\hbar D_t+\underline{m}^{-1}\left( \hbar D_s+\underline{A}(s,t) \right)\underline{m}^{-1}\left( \hbar D_s+\underline{A}(s,t) \right). $$
We denote by $\left(\underline{\lambda}(\hbar,\delta)\right)_{n\geq 1}$ the increasing sequence of eigenvalues of operator $\underline{\mathcal{M}}_{\hbar,\delta}$.

Using the same method of the proof of Proposition \ref{R1}, Agmon estimates of $\underline{\mathcal{M}}_{\hbar,\delta}$ in coordinates $(s,t)$ and the min-max principle we can obtain the following proposition.
\begin{prop}
Let $n\geq 1$. There exist $C,\hbar_0,\beta>0$ such that, for all $\hbar\in(0,\hbar_0)$ and $\delta\in(0,\delta_0)$,
$$ \underline{\lambda}_n(\hbar,\delta)\leq \lambda_n(\hbar,\delta)\leq \underline{\lambda}_n(\hbar,\delta)+ C\mathrm{e}^{-\frac{\beta\delta^{\frac{k+2}{2}}}{\sqrt{\hbar}}}. $$
\label{R3}
\end{prop}
From now on we fix
\begin{equation}
    \delta=\hbar^{\frac{k}{(k+2)^2}-\frac{2\eta}{k+2}},
    \label{delta}
\end{equation}

for some fixed $ 0< \eta < \frac{k}{2(k+2)} $, which verifies that
$$ \delta^{\frac{k+2}{2}}=\hbar^{\frac{k}{2(k+2)}-\eta} >> \hbar^{\frac{k}{2(k+2)}}. $$
Now the $\hbar^{\frac{1}{k+2}}$-scale normal localization invites us to make the following change of variable
$$ (s,t)=(\sigma,h\tau), $$
where $h=\hbar^{\frac{1}{k+2}}$ is the new semi-classical parameter. With this change of variable, the metric and the new magnetic potential are given by
$$ \mathfrak{a}_h(\sigma,\tau)=1-h\tau c_{\mu}(\tau)\kappa(\sigma), $$
and
$$ \mathcal{A}_h^{[k]}(\sigma,\tau)=-h^{-k-1}\beta_0+\gamma(\sigma)\frac{\tau^{k+1}}{k+1}+h\tilde{\delta}(\sigma)c_{\mu}(\tau)\frac{\tau^{k+2}}{k+2}+h^2c_{\mu}\mathcal{O}\left(\tau^{k+3}\right), $$
with $c_{\mu}(\tau)=c(\mu\tau)$ where $\mu >0 $. In the following, the estimates are uniform with respect to the parameter $\mu$ and then we will fix after $\mu=h^{\frac{2}{k+2}+2\eta}$. Dividing $\underline{\mathcal{M}}_{\hbar,\delta}$ by $h^{2k+2}$, we get the rescaled operator
\begin{equation}
    \mathcal{N}_{h}^{[k]}=\mathfrak{a}_h^{-1}D_{\tau}\mathfrak{a}_hD_{\tau}+\mathfrak{a}_h^{-1}\left( hD_{\sigma}-\mathcal{A}_h^{[k]}(\sigma,\tau)\right)\mathfrak{a}_h^{-1}\left( hD_{\sigma}-\mathcal{A}_h^{[k]}(\sigma,\tau)\right).
    \label{N}
\end{equation}

We denote by $\left( \nu_n(h) \right)_{n\geq 1}$ the sequence of eigenvalues of $\mathcal{N}_h^{[k]}$. Then for all $n\geq 1$ we have
$$ \lambda_n\left(\underline{\mathcal{M}}_{\hbar,\delta} \right)=h^{2k+2} \nu_n(h)=\hbar^{2\frac{k+1}{k+2}} \nu_n(h). $$
\begin{prop}
Let $n\geq 1$. There exist $D > S, C,\hbar_0>0$ such that, for all $\hbar\in(0,\hbar_0)$
$$ \lambda_n(\hbar)-C\mathrm{e}^{-\frac{D}{\hbar^{1/(k+2)}}}\leq h^{2k+2}\nu_n(h)\leq \lambda_n(\hbar)+C\mathrm{e}^{-\frac{D}{\hbar^{1/(k+2)}}}. $$
\label{R4}
\end{prop}
\begin{proof}
Using Propositions \ref{R1} and \ref{R3}, we can deduce that
$$ \lambda_n(\hbar)-C\mathrm{e}^{-\frac{\beta\delta^{\frac{k+2}{2}}}{\sqrt{\hbar}}}\leq h^{2k+2}\nu_n(h)\leq \lambda_n(\hbar)+C\mathrm{e}^{-\frac{\beta\delta^{\frac{k+2}{2}}}{\sqrt{\hbar}}}. $$
With the choice of $\delta$, we have
$$ \mathrm{e}^{-\frac{\beta\delta^{\frac{k+2}{2}}}{\sqrt{\hbar}}}=\mathrm{e}^{-\frac{\beta\hbar^{-\eta}}{\hbar^{1/(k+2)}}}. $$
Therefore, there exist $D > S$ such that
$$ \lambda_n(\hbar)-C\mathrm{e}^{-\frac{D}{\hbar^{1/(k+2)}}}\leq h^{2k+2}\nu_n(h)\leq \lambda_n(\hbar)+C\mathrm{e}^{-\frac{D}{\hbar^{1/(k+2)}}}.  $$
\end{proof}

\section{Single well}
\label{section3}
The function $\gamma$ admits two non-degenerate minima in $s_l$ and $s_r$ on $\Gamma$. We will now consider two operators $\mathcal{N}_{h,l,\beta_0}^{[k]}$ and $\mathcal{N}_{h,r,\beta_0}^{[k]}$ which represent the left well operator and the right well operator respectively.
\subsection{Right well operator}
\label{Right well}

\begin{figure}[ht]
\begin{center}
\definecolor{tttttt}{rgb}{0.2,0.2,0.2}
\definecolor{uququq}{rgb}{0.25,0.25,0.25}
\begin{tikzpicture}[line cap=round,line join=round,>=triangle 45,x=1.0cm,y=1.0cm]
\clip(-6.44,-4.12) rectangle (7.36,3.97);
\draw[smooth,samples=100,domain=-4.3:4.3] plot(\x,{0-(-(\x)^2-9+(36*(\x)^2+3.2^4)^0.5)^0.5});
\draw [shift={(2.79,0)}] plot[domain=-0.33:0.33,variable=\t]({1*1.59*cos(\t r)+0*1.59*sin(\t r)},{0*1.59*cos(\t r)+1*1.59*sin(\t r)});
\draw [shift={(-2.79,0)}] plot[domain=2.81:3.47,variable=\t]({1*1.59*cos(\t r)+0*1.59*sin(\t r)},{0*1.59*cos(\t r)+1*1.59*sin(\t r)});
\draw (2.79,-1.6) node[anchor=north west] {$ \Gamma $};
\draw (4.01,1.41) node[anchor=north west] {$ s_r $};
\draw (-4.55,1.35) node[anchor=north west] {$ s_l $};
\draw [shift={(3.87,0.51)},dash pattern=on 2pt off 2pt,color=tttttt]  plot[domain=0.13:0.82,variable=\t]({1*1.43*cos(\t r)+0*1.43*sin(\t r)},{0*1.43*cos(\t r)+1*1.43*sin(\t r)});
\draw [->,color=tttttt] (4.95,1.45) -- (4.85,1.55);
\draw [color=tttttt](5.2,1.39) node[anchor=north west] {+};
\draw (0.04,-0.68) node[anchor=north west] {$ -L $};
\draw (-0.82,-0.68) node[anchor=north west] {$ +L $};
\draw[dash pattern=on 2pt off 2pt, smooth,samples=100,domain=-4.3:-3.7] plot(\x,{(-(\x)^2-9+(36*(\x)^2+3.2^4)^0.5)^0.5});
\draw[smooth,samples=100,domain=-3.7:4.3] plot(\x,{(-(\x)^2-9+(36*(\x)^2+3.2^4)^0.5)^0.5});
\draw [shift={(-3.94,1.58)},dash pattern=on 2pt off 2pt]  plot[domain=3.45:5.46,variable=\t]({1*0.35*cos(\t r)+0*0.35*sin(\t r)},{0*0.35*cos(\t r)+1*0.35*sin(\t r)});
\draw [shift={(-1.24,2.53)},dash pattern=on 2pt off 2pt]  plot[domain=3.24:3.48,variable=\t]({1*3.21*cos(\t r)+0*3.21*sin(\t r)},{0*3.21*cos(\t r)+1*3.21*sin(\t r)});
\draw [shift={(-4.49,0.58)},dash pattern=on 2pt off 2pt]  plot[domain=-0.32:1.69,variable=\t]({1*0.2*cos(\t r)+0*0.2*sin(\t r)},{0*0.2*cos(\t r)+1*0.2*sin(\t r)});
\draw [shift={(-4.44,-0.11)},dash pattern=on 2pt off 2pt]  plot[domain=1.65:2.66,variable=\t]({1*0.89*cos(\t r)+0*0.89*sin(\t r)},{0*0.89*cos(\t r)+1*0.89*sin(\t r)});
\draw [->] (-4.41,2.06) -- (-4.43,2.2);
\draw (-5.5,0.4) node[anchor=north west] {$ -\infty $};
\draw (-4.41,2.43) node[anchor=north west] {$ +\infty $};
\draw (-3.64,1.42) node[anchor=north west] {$ s_l-\hat{\eta} $};
\draw (-4.27,0.58) node[anchor=north west] {$ s_l+\hat{\eta} $};
\draw (-0.53,1.61) node[anchor=north west] {$ s=0 $};
\begin{scriptsize}
\fill [color=uququq] (4,1.05) circle (1.5pt);
\fill [color=uququq] (-4,1.05) circle (1.5pt);
\fill [color=uququq] (0,1.11) circle (1.5pt);
\fill [color=uququq] (0,-1.11) circle (1.5pt);
\fill [color=uququq] (-3.7,1.33) circle (1.5pt);
\fill [color=uququq] (-4.3,0.52) circle (1.5pt);
\end{scriptsize}
\end{tikzpicture}

\caption{One well domain attached to the right well}
\label{One well}
\end{center}
\end{figure}
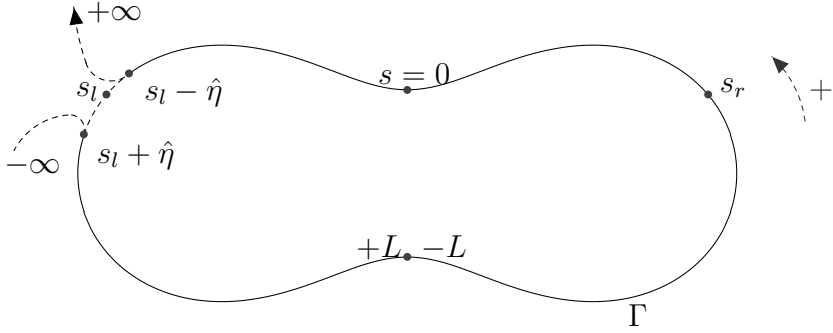

This operator is attached to the right well $s_r$. We will work on $\mathbb{R}\times\mathbb{R}$ instead of $\Gamma\times\mathbb{R}$ and with only one well. For this we will remove the left well by removing a small neighborhood of $s_l$, and gluing an infinite strip: precisely we start by identifying $\Gamma$ with $(s_l-2L,s_l]$. We fix $\hat{\eta}$ so that
\begin{equation}
    0< \hat{\eta}< \text{min}\left\{ \frac{1}{4},\frac{L}{4} \right\}. 
    \label{etaa}
\end{equation}

We consider the following right well differential operator in $L^2\left(\mathbb{R}\times\mathbb{R};\mathfrak{a}_{h,r}d\sigma d\tau \right)$,
\begin{equation}
    \mathcal{N}_{h,r,\beta_0}^{[k]}=\mathfrak{a}_{h,r}^{-1}D_{\tau}\mathfrak{a}_{h,r}D_{\tau}+\mathfrak{a}_{h,r}^{-1}\left( hD_{\sigma}-\mathcal{A}_{h,r,\beta_0}^{[k]}(\sigma,\tau)\right)\mathfrak{a}_{h,r}^{-1}\left( hD_{\sigma}-\mathcal{A}_{h,r,\beta_0}^{[k]}(\sigma,\tau)\right),
    \label{puit droite}
\end{equation}
with
$$ \mathfrak{a}_{h,r}(\sigma,\tau)=1-h\tau c_{\mu}(\tau)\kappa_r(\sigma), $$
and
$$ \mathcal{A}_{h,r,\beta_0}^{[k]}(\sigma,\tau)=-h^{-k-1}\beta_0+\gamma_r(\sigma)\frac{\tau^{k+1}}{k+1}+h\tilde{\delta}_r(\sigma)\frac{\tau^{k+2}}{k+2}c_{\mu}(\tau)+c_{\mu}h^2\mathcal{O}\left( \tau^{k+3}\right), $$
where the functions $\delta_r$ and $\kappa_r$ are respective extensions of $\delta$ and $\kappa$ such that
$$ \delta_r(\sigma)=\delta(\sigma)\,\,\,\,\text{and}\,\,\,\,\,\kappa_r(\sigma)=\kappa(\sigma)\,\,\,\,\,\text{on}\,\,\,\,\,I_{r,\hat{\eta}}:=(s_l-2L+\hat{\eta},s_l-\hat{\eta}), $$
and are zero functions on $(-\infty,s_l-2L)\cup (s_l,+\infty)$. On the other hand, the extension $\gamma_r$ of $\gamma$ is chosen so that $\gamma_r=\gamma$ on $I_{r,\hat{\eta}}$ and $\gamma_r=\gamma_{\infty}$ on $(-\infty,s_l-2L)\cup (s_l,+\infty)$, with $\gamma_{\infty}\in\mathbb{R}_+^{\ast}$ such that $\gamma_{\infty}>\underset{\sigma\in\Gamma}{\text{max}}\,\gamma(\sigma)$.
This extension can be chosen so that $\gamma_r$ admits a unique non-degenerate minimum $\gamma_0>0$ at $s_r<0$ and that $\left\| 1-\frac{\gamma_0}{\gamma_r} \right\|_{\infty}$ is small enough.

 Since we are now working with a simply connected operator, then the two operators $\mathcal{N}_{h,r,\beta_0}^{[k]}$  and $\mathcal{N}_{h,r,0}^{[k]}$ are equivalent.
We denote by $u_{h,r}^{[k]}$ a normalized ground state of the operator $\mathcal{N}_{h,r}^{[k]}:=\mathcal{N}_{h,r,0}^{[k]}$ in $L^2\left(\mathbb{R}\times\mathbb{R};\mathfrak{a}_{h,r}d\sigma d\tau\right)$, and the normalized ground state of $\mathcal{N}_{h,r,\beta_0}^{[k]}$ is
given by: 
\begin{equation}
   \Check{\phi}_{h,r}^{k}(\sigma,\tau)=\mathrm{e}^{-\mathrm{i}\beta_0\sigma/{h^{k+2}}}u_{h,r}^{[k]}(\sigma,\tau).
    \label{aaa1}
\end{equation}

\subsection{Left well operator}
\label{Left well}
To define the left well operator, we consider the symmetry operator
\begin{equation}
    Uf(\sigma,\tau):=\overline{f(-\sigma,\tau)},
    \label{sym}
\end{equation}

and define the left well operator on $L^2\left(\mathbb{R}\times\mathbb{R};\mathfrak{a}_{h,l}d\sigma d\tau\right)$ by
$$ \mathcal{N}_{h,l,\beta_0}^{[k]}=U^{-1}\mathcal{N}_{h,r,\beta_0}^{[k]}U, $$
where
$$ \mathfrak{a}_{h,l}(\sigma,\tau)=\mathfrak{a}_{h,r}(-\sigma,\tau). $$

Note that this operator corresponds to the following construction. We identify $\Gamma$ with $[s_r,s_r+2L)$, we can define on $\mathbb{R}$ the functions $\gamma_l$, $\delta_l$ and $\kappa_l$ by $\gamma_l(\sigma)=\gamma_r(-\sigma)$, $\delta_l(\sigma)=\delta_r(-\sigma)$ and $\kappa_l(\sigma)=\kappa_r(-\sigma)$.\\
Then the functions $\delta_l$ and $\kappa_l$ verify that
$$ \delta_l(\sigma)=\delta(\sigma)\,\,\,\text{and}\,\,\,\kappa_l(\sigma)=\kappa(\sigma)\,\,\,\,\,\,\,\,\,\,\,\,\,\text{on}\,\,\,\,\, I_{l,\hat{\eta}}:=(s_r+\hat{\eta},s_r+2L-\hat{\eta}), $$
and are zero functions on $(-\infty,s_r)\cup (s_r+2L,+\infty)$. On the other hand, the extension $\gamma_l$ of $\gamma$ is chosen so that $\gamma_l=\gamma$ on $I_{l,\hat{\eta}}$ and $\gamma_l=\gamma_{\infty}$ on $(-\infty,s_r)\cup (s_r+2L,+\infty)$.
In this way, $\gamma_l$ admits a unique non-degenerate minimum $\gamma_0>0$ at $s_l>0$, and verify that $\left\| 1-\frac{\gamma_0}{\gamma_l} \right\|_{\infty}$ is small enough.

The normalized ground state of the operator $\mathcal{N}_{h,l,\beta_0}^{[k]}$ on $L^2\left(\mathbb{R}\times\mathbb{R};\mathfrak{a}_{h,l}d\sigma d\tau\right)$ is given by
\begin{equation}
    \check{\phi}_{h,l}^{[k]}(\sigma,\tau):=U\check{\phi}_{h,r}^{[k]}(\sigma,\tau)=\mathrm{e}^{-\mathrm{i}\beta_0\sigma/{h^{k+2}}}u_{h,l}^{[k]}(\sigma,\tau),
    \label{aaa2}
\end{equation}
where $u_{h,l}^{[k]}=Uu_{h,r}^{[k]}$.

\section{WKB expansions of the right well operator}
\label{section4}
In this section, we will construct an approximation of the eigenvalues and the associated eigenfunctions for the right well operator $\mathcal{N}_{h,r}^{[k]}:=\mathcal{N}_{h,r,0}^{[k]}$ by WKB expansions, and the construction for the left well operator $\mathcal{N}_{h,l}^{[k]}:=\mathcal{N}_{h,l,0}^{[k]}$ is obtained by symmetry. These WKB constructions are inspired by \cite{BHR-BKW}.

\subsection{Generalized Montgomery operator}
\label{GMO}
For $(x,\xi)\in\mathbb{R}^2$, we consider the symbol operator
$$ \mathcal{M}_{x,\xi}^{[k]}=D_t^2+\left(\xi-\gamma_r(x)\frac{t^{k+1}}{k+1}\right)^2. $$
In the sense of Kato's perturbation theory (see \cite{kato2013perturbation}), the family of self-adjoint operators $\left( \mathcal{M}_{x,\xi}^{[k]} \right)_{(x,\xi)\in\mathbb{R}^2}$ is analytic of type (A). Then, for $(x_0,\xi_0)\in\mathbb{R}^2$, the family $\left( \mathcal{M}_{x,\xi}^{[k]} \right)_{(x,\xi)\in\mathbb{R}^2}$ can be extended into a family of closed operators $\left( \mathcal{M}_{x,\xi}^{[k]} \right)_{(x,\xi)\in\mathcal{V}}$ where $\mathcal{V}$ is a $\mathbb{C}$-neighborhood of $(x_0,\xi_0)$. The lowest eigenvalue of $\mathcal{M}_{x,\xi}^{[k]}$, denoted by $\mu^{[k]}(x,\xi)$, satisfies
\begin{equation}
    \mu^{[k]}(x,\xi)=\gamma_r(x)^{\frac{2}{k+2}}\nu^{[k]}\left(\gamma_r(x)^{-\frac{1}{k+2}}\xi\right).
    \label{VP}
\end{equation}
Since the function $\mathbb{R}\ni\sigma\mapsto\gamma_r(\sigma)$ admits a unique non-degenerate minimum $\gamma_0>0$ at $s_r$, then the function $\mathbb{R}^2\ni(x,\xi)\mapsto\mu^{[k]}(x,\xi)$ admits a unique non-degenerate minimum at $(s_r,\gamma_0^{\frac{1}{k+2}}\xi_0^{[k]})$ give by
\begin{equation}
    \mu_0^{[k]}:=\gamma_0^{\frac{2}{k+2}}\nu^{[k]}(\xi_0^{[k]})>0.
    \label{cte thm}
\end{equation}

We denote by $u_{x,\xi}^{[k]}$ the eigenfunction of $\mathcal{M}_{x,\xi}^{[k]}$ associated with the eigenvalue $\mu^{[k]}(x,\xi)$.\\
By differentiating with respect to $\xi$ equation $\left( \mathcal{M}_{x,\xi}^{[k]}-\mu^{[k]}(x,\xi)\right)u_{x,\xi}^{[k]}=0$ and taking the inner product with $u_{\overline{x},\overline{\xi}}^{[k]}$ in $L^2(\mathbb{R})$, we get
\begin{equation}
    \partial_{\xi}\mu^{[k]}(x,\xi)=\displaystyle\int_{\mathbb{R}}\left( \left( \partial_{\xi}\mathcal{M}_{x,\xi}^{[k]}\right)u_{x,\xi}^{[k]}(\tau)\right)u_{x,\xi}^{[k]}(\tau)d\tau.  
    \label{FH2}
\end{equation}
By differentiating the function $\mu^{[k]}$ with respect to $x$ and $\xi$, the Hessian matrix of $\mu^{[k]}$ at $(s_r,\gamma_0^{\frac{1}{k+2}}\xi_0^{[k]})$ is given by
\begin{equation}
    \text{Hess} \, \mu^{[k]}(s_r,\gamma_0^{\frac{1}{k+2}}\xi_0^{[k]})=\left(\begin{array}{cc}
\frac{2}{k+2}\gamma''(s_r)\gamma_0^{-\frac{k}{k+2}}\nu^{[k]}\left(\xi_0^{[k]}\right) & 0 \\ 0 & \left(\nu^{[k]}\right)''\left( \xi_0^{[k]} \right) \end{array} \right).
\label{matrice hessienne}
\end{equation}
\subsection{Eikonal equation}
\label{eikn}
We consider the following equation
\begin{equation}
    \nu\left( \mathrm{i}\varphi_r(\sigma) \right)=F_r(\sigma),
    \label{EIK1}
\end{equation}
where $\nu(\xi)=\nu^{[k]}\left(\xi_0^{[k]}+\xi \right)-\nu^{[k]}\left(\xi_0^{[k]}\right)$ and $F_r(\sigma)=\nu^{[k]}\left(\xi_0^{[k]}\right) \left( \left( \frac{\gamma_0}{\gamma_r(\sigma)} \right)^{2/{(k+2)}}-1 \right).$\\
This eikonal equation can be found in \cite[Section 4]{BHR-BKW}. The following lemma is the same as the one of \cite[Lemma 4.4]{BHR-BKW}. Since $\left\|1-\frac{\gamma_0}{\gamma_r}\right\|_{\infty}$ is small enough, the solution of this equation is defined for all $\sigma\in (s_l+\hat{\eta}-2L,s_l-\hat{\eta})$ where $\hat{\eta}>0$ is introduced in \eqref{etaa}.
\begin{lemma}
Equation \eqref{EIK1} admits a smooth solution $\varphi_r$ defined on $(s_l+\hat{\eta}-2L,s_l-\hat{\eta})$ such that $\varphi_r(s_r)=0$ and $\varphi_r'(s_r)=\sqrt{\frac{2}{k+2}\frac{\gamma''(s_r)\nu^{[k]}\left(\xi_0^{[k]}\right)}{\gamma_0\left(\nu^{[k]}\right)''\left(\xi_0^{[k]}\right)}}>0$.
\label{Lemme}
\end{lemma}
Using the proof of Lemma $4.4$ in \cite[Section 4]{BHR-BKW}, the function $\varphi$ is given by
$$  \varphi_r(\sigma)=-\mathrm{i}\tilde{\nu}^{-1}\left( \mathrm{i}\mathfrak{f}_r(\sigma) \right), $$
where $\tilde{\nu}$ is a holomorphic function in a neighborhood of $0$ such that $\tilde{\nu}^2=\nu$ and $\tilde{\nu}'(0)=\sqrt{\frac{\nu''(0)}{2}}$ and the function $\mathfrak{f}_r$ is defined by
$$ \mathfrak{f}_r(\sigma)=\begin{cases}
\,\,\,\,\sqrt{\nu^{[k]}(\xi_0^{[k]})}\sqrt{1-\left( \frac{\gamma_0}{\gamma_r(\sigma)} \right)^{2/(k+2)}} \,\,\,\,\,\,\text{if}\,\,\,\,\sigma\geq s_r, \\
-\sqrt{\nu^{[k]}(\xi_0^{[k]})}\sqrt{1-\left( \frac{\gamma_0}{\gamma_r(\sigma)} \right)^{2/(k+2)}} \,\,\,\,\,\text{if}\,\,\,\,\sigma\leq s_r.
\end{cases} $$
The function $\mathfrak{f}_r$ is derivable at $s_r$ and $\mathfrak{f}_r'(0)=\sqrt{-\frac{F_r''(0)}{2}}>0$. The taylor series of $\tilde{\nu}^{-1}$ at $0$ gives 
$$ \varphi_r(\sigma)=\Re\left(\varphi_r(\sigma)\right)+\mathrm{i}\Im\left(\varphi_r(\sigma)\right), $$
with
\begin{equation}
    \Re\left(\varphi_r(\sigma)\right)=\sqrt{\frac{2}{\left( \nu^{[k]} \right)''(\xi_0^{[k]})}}\left|\mathfrak{f}_r(\sigma)\right|+\mathcal{O}\left( \left|\mathfrak{f}_r(\sigma)\right|^3 \right),
    \label{R}
\end{equation}
and
\begin{equation}
    \Im\left(\varphi_r(\sigma)\right)=\frac{\left(\tilde{\nu}^{-1}\right)''(0)}{2}\mathfrak{f}_r(\sigma)^2 +\mathcal{O}\left(\mathfrak{f}_r(\sigma)^4\right).
    \label{I}
\end{equation}
\begin{rmk}
Concerning the left well operator $\mathcal{N}_{h,l}^{[k]}$ defined in \eqref{Left well}, equation
$$ \nu\left( \mathrm{i}\varphi_l(\sigma) \right)=F_l(\sigma) $$
admits also a smooth solution $\varphi_l$ defined on $(s_r+\hat{\eta},s_r+2L-\hat{\eta})$ such that $\varphi_l(s_l)=0$ and $$\varphi_l'(s_l)=\sqrt{\frac{2}{k+2}\frac{\gamma''(s_r)\nu^{[k]}\left(\xi_0^{[k]}\right)}{\gamma_0\left(\nu^{[k]}\right)''\left(\xi_0^{[k]}\right)}}>0,$$ 
where
$$ F_l(\sigma)=\nu^{[k]}\left(\xi_0^{[k]}\right) \left( \left( \frac{\gamma_0}{\gamma_l(\sigma)} \right)^{2/{(k+2)}}-1 \right). $$
As the construction of $\varphi_r$, the function $\varphi_l$ is defined by
$$  \varphi_l(\sigma)=-\mathrm{i}\tilde{\nu}^{-1}\left( \mathrm{i}\mathfrak{f}_l(\sigma) \right), $$
where
$$\mathfrak{f}_l(\sigma)=\begin{cases}
\,\,\,\,\sqrt{\nu^{[k]}(\xi_0^{[k]})}\sqrt{1-\left( \frac{\gamma_0}{\gamma_l(\sigma)} \right)^{2/(k+2)}} \,\,\,\,\,\,\text{if}\,\,\,\,\sigma\geq s_l, \\
-\sqrt{\nu^{[k]}(\xi_0^{[k]})}\sqrt{1-\left( \frac{\gamma_0}{\gamma_l(\sigma)} \right)^{2/(k+2)}} \,\,\,\,\,\text{if}\,\,\,\,\sigma\leq s_l.
\end{cases} $$
\label{rmk1}
\end{rmk}
By construction of $\varphi_r$ and $\varphi_l$, we can define the two even smooth functions $\mathfrak{V}$ and $\mathfrak{I}$ on $\Gamma\equiv [-L,+L]$ by
\begin{equation}
    \mathfrak{D}(\sigma):= \begin{cases}
-\Re \varphi_r(\sigma) & \,\,\,\text{if}\,\,\,\sigma\in[-L,s_r],\\
\Re \varphi_r(\sigma) & \,\,\,\text{if}\,\,\,\sigma\in[s_r,0],\\
-\Re \varphi_l(\sigma) & \,\,\,\text{if}\,\,\,\sigma\in[0,s_l],\\
\Re \varphi_l(\sigma) & \,\,\,\text{if}\,\,\,\sigma\in[s_l,+L].
\end{cases}
\label{M}
\end{equation}
and
\begin{equation}
    \mathfrak{I}(\sigma):= \begin{cases}
\Im \varphi_r(\sigma) & \,\,\,\text{if}\,\,\,\sigma\in[-L,0],\\
\Im \varphi_l(\sigma) & \,\,\,\text{if}\,\,\,\sigma\in[0,+L].
\end{cases}
\label{varphi}
\end{equation}

\subsection{WKB expansions}
The WKB expansions of operator $\mathcal{N}_{h,r}^{[k]}$ are inspired from \cite[Section 5, Theorem 5.2]{BHR-BKW}. In the following theorem, we will construct these approximations and specify the Agmon distance adapted to our case which will be a positive real function.

Let us introduce the Agmon distance related to the ``right well''
$$ \Phi_r(\sigma)=\int_{s_r}^{\sigma}\gamma_r(\tilde{\sigma})^{1/(k+2)}\Re\left( \varphi_r(\tilde{\sigma}) \right) d\tilde{\sigma}, $$
which verifies that $\Phi_r''(s_r)>0$ where $\varphi_r$ is the function defined in Lemma \ref{Lemme}.
\begin{theorem}
There exist a sequence of smooth functions $\left( a_{n,j}^{[k]} \right)_{j\geq 0}\subset Dom \left( \mathcal{N}_{h,r}^{[k]} \right)$, a sequence of real numbers $\left( \delta_{n,j}^{[k]}\right)_{j\geq 0}\subset \mathbb{R}$, a family of functions $\left( \Psi_{h,r}^{[k]}\right)_{h\in(0,h_0]}\subset L^2\left(\mathbb{R}^2\right)$ and a family of real numbers $\left( \delta_n^{[k]}(h)\right)_{h\in(0,h_0]}$ such that
$$ \Psi_{h,r}^{[k]}(\sigma,\tau)\sim h^{-1/4}\mathrm{e}^{-\frac{\Phi_r(\sigma)}{h}}\mathrm{e}^{\mathrm{i}\frac{\mathfrak{g}_r(\sigma)}{h}}\sum_{j\geq 0}a_{n,j}^{[k]}(\sigma,\tau)h^j, $$
$$ \delta_n^{[k]}(h)\sim \sum_{j\geq 0}\delta_{n,j}^{[k]}h^j, $$
and
$$ \left( \mathcal{N}_{h,r}^{[k]}-\delta_n^{[k]}(h) \right)\Psi_{h,r}^{[k]} =\mathcal{O}(h^{\infty})\mathrm{e}^{-\Phi_r/h}, $$
with 
\begin{equation}
    \mathfrak{g}_r(\sigma)=\int_0^{\sigma}\gamma_r(\tilde{\sigma})^{\frac{1}{k+2}} \left(\xi_0^{[k]} - \Im \left( \varphi_r(\tilde{\sigma}) \right)\right) d\tilde{\sigma}.
    \label{gsigma}
\end{equation}
Furthermore
\begin{enumerate}
\item[1)]$\delta_{n,0}^{[k]}=\gamma_0^{\frac{2}{k+2}}\nu^{[k]}(\xi_0^{[k]})$ and $\delta_{n,1}^{[k]}=\frac{\left(\nu^{[k]}\right)''\left(\xi_0^{[k]}\right)}{2}(2n-1)\zeta+\mathfrak{R}_r(s_r),$ \\
\item[2)]$a_{n,0}^{[k]}(\sigma,\tau)=f_{n,0}(\sigma) u_{\sigma ,\mathfrak{w}_r(\sigma)}^{[k]}(\tau),$ where $f_{n,0}$ solves the effective transport equation
\begin{equation}
\frac{1}{2}\left(D_{\sigma}\partial_{\xi}\mu^{[k]}(\sigma,\mathfrak{w}_r(\sigma))+\partial_{\xi}\mu^{[k]}(\sigma,\mathfrak{w}_r(\sigma))D_{\sigma}\right)f_{n,0}+R_r^{[k]}(\sigma)f_{n,0}=\delta_{n,1}f_{n,0},
\label{equation transport}
\end{equation}
\end{enumerate}
with
$$ \zeta=\sqrt{\frac{2}{k+2}\frac{\gamma''(0)\nu^{[k]}\left(\xi_0^{[k]}\right)}{\gamma_0^{\frac{k}{k+2}}\left(\nu^{[k]} \right)''\left(\xi_0^{[k]}\right)}}, $$
\begin{align*}
\mathfrak{R}_r(\sigma)&=2\gamma_r(\sigma)\left(\delta_r(\sigma)+\frac{\kappa_r(\sigma)\gamma_r(\sigma)}{k+1}\right)\int c_{\mu}  \frac{\tau^{2k+3}}{(k+1)(k+2)}\left(u_{\sigma,\mathfrak{w}_r(\sigma)}^{[k]}(\tau)\right)^2 d\tau\\&+\kappa_r(\sigma)\int \left( c_{\mu} +c_{\mu} ' \tau \right) \partial_{\tau}u_{\sigma,\mathfrak{w}(\sigma)}^{[k]}(\tau)u_{\sigma,\mathfrak{w}_r(\sigma)}^{[k]}(\tau)d\tau \\&-2\mathfrak{w}_r(\sigma)\left(\delta_r(\sigma)+\frac{(k+3)\kappa_r(\sigma)\gamma_r(\sigma)}{k+1}\right)\int c_{\mu} \frac{\tau^{k+2}}{k+2}\left(u_{\sigma,\mathfrak{w}_r(\sigma)}^{[k]}(\tau)\right)^2d\tau\\&+2 \mathfrak{w}_r (\sigma)^2 \kappa_r(\sigma)\int c_{\mu} \tau u_{\sigma,\mathfrak{w}_r(\sigma)}^{[k]}(\tau)d\tau,
\end{align*}

where 
\begin{equation}
\mathfrak{w}_r(\sigma):=\mathrm{i}\Phi_r'(\sigma)+\mathfrak{g}_r'(\sigma) \,\,\,\text{and} \,\,\, \mathfrak{w}_r(s_r)=\gamma_0^{\frac{1}{k+2}}\xi_0^{[k]}\,. 
\label{nsita1}
\end{equation}

\label{BKW}
\end{theorem}
\begin{proof}
For some real function $\Phi_r=\Phi_r(\sigma)$ to be determined, we introduce the conjugate operator
$$\tilde{\mathcal{N}}_{h,r}^{[k]}=\mathrm{e}^{\frac{\Phi_r (\sigma)-\mathrm{i}\mathfrak{g}_r(\sigma)}{h}}\mathcal{N}_{h,r}^{[k]}\mathrm{e}^{-\frac{\Phi_r (\sigma)-\mathrm{i}\mathfrak{g}a_r(\sigma)}{h}},$$
and expand it formally as follows
$$ \tilde{\mathcal{N}}_{h,r}^{[k]}\sim \sum_{j\geq 0}\mathcal{N}_jh^j, $$
with
$$ \mathcal{N}_0=D_{\tau}^2+\left( \mathfrak{w}_r(\sigma)-\gamma_r(\sigma)\frac{\tau^{k+1}}{k+1}\right)^2,$$
$$ \mathcal{N}_1=D_{\sigma}\left( \mathfrak{w}_r(\sigma)-\gamma_r(\sigma)\frac{\tau ^{k+1}}{k+1}\right)+\left( \mathfrak{w}_r(\sigma)-\gamma_r(\sigma)\frac{\tau ^{k+1}}{k+1}\right)D_{\sigma}+\mathcal{R}_r(\sigma,\tau), $$
where
\begin{align*}
    \mathcal{R}_r(\sigma,\tau)=&2\tau c_{\mu} \kappa_r(\sigma) \Bigg(\mathfrak{w}_r(\sigma)-\gamma_r(\sigma)\frac{\tau^{k+1}}{k+1}\Bigg) ^2-2\tilde{\delta}_r(\sigma)c_{\mu}\frac{\tau^{k+2}}{k+2}\Bigg(\mathfrak{w}_r(\sigma)-\gamma_r(\sigma)\frac{\tau^{k+1}}{k+1}\Bigg)\\&+c_{\mu} \kappa_r(\sigma)\partial_{\tau}+c_{\mu} '\kappa_r(\sigma) \tau \partial_{\tau},
\end{align*}
$\mathfrak{w}_r(\sigma)=\mathrm{i}\Phi_r'(\sigma)+\mathfrak{g}_r'(\sigma)$, and the function $\mathfrak{g}_r$ is defined in \eqref{gsigma}.\\
Let $a^{[k]}(\sigma,\tau;h)=\sum_{j\geq 0}a_{n,j}^{[k]}(\sigma,\tau)h^j$ and let us formally solve equation
$$ \left( \tilde{\mathcal{N}}_{h,r}^{[k]} -\delta_n^{[k]}(h) \right)a^{[k]}(\sigma,\tau;h)=\mathcal{O}(h^{\infty}). $$
Identifying the coefficient of each $h^j$, $j\geq 0$, gives us first
$$ \left( \mathcal{N}_0-\delta_{n,0}^{[k]} \right)a_{n,0}^{[k]}=0, \,\,\,\,\,\,\,\,\,\,\,\,\,\,\,\,\,\,\,\,\,\,\,\,\,\,\,\,\,\,\,\,\,\,\,\,\,\,\,\,\,\,\,\,\,\,\,\text{(A)} $$
$$ \left( \mathcal{N}_0-\delta_{n,0}^{[k]} \right)a_{n,1}^{[k]}=\left( \delta_{n,1}^{[k]}- \mathcal{N}_1\right) a_{n,0}^{[k]}.\,\,\,\,\,\,\,\,\,\,(B) $$

Noticing that $\mathcal{N}_0=\mathcal{M}_{\sigma,\mathfrak{w}_r(\sigma)}^{[k]}$, we get that the first equation allows us to choose the function $\Phi_r$ such that 
\begin{equation}
    \delta_{n,0}^{[k]}=\mu_0^{[k]}=\mu^{[k]}(\sigma,\mathfrak{w}_r(\sigma)),
    \label{Eik2}
\end{equation}
and $a_{n,0}^{[k]}(\sigma,\tau)=f_{n,0}(\sigma)u_{\sigma,\mathfrak{w}_r(\sigma)}^{[k]}(\tau)$ where $f_{n,0}$ is to be determined at a later stage.\\
Indeed, using \eqref{VP}, and that $\mu_0^{[k]}=\gamma_0^{\frac{2}{k+2}}\nu^{[k]}\left( \xi_0^{[k]}\right)$, the eikonal equation \eqref{Eik2} is given by
\begin{equation}
    \gamma_r(\sigma)^{\frac{2}{k+2}}\nu^{[k]}\left( \xi_0^{[k]}+\mathrm{i}\left( \gamma_r(\sigma)^{-\frac{1}{k+2}}\Phi_r'(\sigma)+\mathrm{i}\Im \left(\varphi_r(\sigma)\right)\right)\right)=\gamma_0^{\frac{2}{k+2}}\nu^{[k]}\left(\xi_0^{[k]}\right),
    \label{Eikonale equation}
\end{equation}
which is equivalent to
$$ \nu^{[k]}\left( \xi_0^{[k]}+\mathrm{i}\left( \gamma_r(\sigma)^{-\frac{1}{k+2}}\Phi_r'(\sigma)+\mathrm{i}\Im \left(\varphi_r(\sigma)\right)\right)\right)-\nu^{[k]}\left(\xi_0^{[k]}\right)=\nu^{[k]}\left(\xi_0^{[k]}\right)\left(\left(\frac{\gamma_0}{\gamma_r}\right)^{\frac{2}{k+2}}-1\right).  $$
Therefore, using Lemma \ref{Lemme}, we choose the function $\Phi_r$ such that
$$ \gamma_r(\sigma)^{-\frac{1}{k+2}}\Phi_r'(\sigma)+\mathrm{i}\Im \left(\varphi_r(\sigma)\right)=\varphi_r(\sigma), $$
which is equivalent to
$$  \gamma_r(\sigma)^{-\frac{1}{k+2}}\Phi_r'(\sigma)=\Re \left(\varphi_r(\sigma)\right). $$
Then we get 
$$ \Phi_r(\sigma)=\int_{s_r}^{\sigma}\gamma_r(\tilde{\sigma})^{1/(k+2)}\Re\left( \varphi_r(\tilde{\sigma}) \right) d\tilde{\sigma}. $$
This function $\Phi_r$ verifies that
$$ \Phi_r(s_r)=\Phi_r'(s_r)=0\,\,\,\,and\,\,\,\,\,\Phi_r''(s_r)=\gamma_0^{1/(k+2)}\varphi_r'(s_r)=\gamma_0^{\frac{1}{k+2}}\sqrt{\frac{2}{k+2}\frac{\gamma''(s_r)\nu^{[k]}\left(\xi_0^{[k]}\right)}{\gamma_0\left(\nu^{[k]}\right)''\left(\xi_0^{[k]}\right)}}>0. $$

The second equation (B) can be solved if the following Fredholm condition holds
$$\left( \delta_{n,1}^{[k]}- \mathcal{N}_1\right) a_{n,0}^{[k]}\in \left(  \text{Ker} \left( \mathcal{N}_0-\delta_{n,0}^{[k]} \right)^{\ast} \right)^{\perp}=\text{span} \left( u_{\sigma,\overline{\mathfrak{w}_r(\sigma)}}^{[k]} \right)^{\perp}. $$
Taking the inner product with $u_{\sigma,\overline{\mathfrak{w}_r(\sigma)}}^{[k]}$ in $L^2\left( \mathbb{R}\right)$, the Fredholm condition will be given by
$$ \langle \mathcal{N}_1a_{n,0}^{[k]},u_{\sigma,\overline{\mathfrak{w}_r(\sigma)}}^{[k]} \rangle _{L^2\left( \mathbb{R},d\tau \right)}=\delta_{n,1}^{[k]}f_{n,0}(\sigma). $$
Noticing that $\left(\partial_{\xi}\mathcal{M}_{x,\xi}^{[k]}\right)_{\sigma,\mathfrak{w}_r(\sigma)}=2\left( \mathfrak{w}_r(\sigma)-\gamma_r(\sigma)\frac{\tau^{k+1}}{k+1} \right)$, $\mathcal{N}_1$ can be written as
$$ \mathcal{N}_1=\frac{1}{2}\left(D_{\sigma}\left(\partial_{\xi}\mathcal{M}_{x,\xi}^{[k]}\right)_{\sigma,\mathfrak{w}_r(\sigma)}+\left(\partial_{\xi}\mathcal{M}_{x,\xi}^{[k]}\right)_{\sigma,\mathfrak{w}_r(\sigma)}D_{\sigma} \right)+\mathcal{R}_r(\sigma,\tau). $$
Using \eqref{FH2} with $x=\sigma$ and $\xi=\mathfrak{w}_r(\sigma)$, we have
\begin{equation}
    \left( \partial_{\xi}\mu^{[k]}(x,\xi)\right)_{\sigma,\mathfrak{w}_r(\sigma)}=\int_{\mathbb{R}} \left( \left( \partial_{\xi}\mathcal{M}_{x,\xi}^{[k]}\right)_{\sigma,\mathfrak{w}_r(\sigma)}u_{\sigma,\mathfrak{w}_r(\sigma)}^{[k]}(\tau)\right) u_{\sigma,\mathfrak{w}_r(\sigma)}^{[k]}(\tau)d\tau.
    \label{FH}
\end{equation}
Multiplying \eqref{FH} by $f_{n,0}(\sigma)$ and differentiating with respect to $\sigma$, we get
\begin{align}
     &D_{\sigma}\left( f_{n,0}(\sigma)\left( \partial_{\xi}\mu^{[k]}(x,\xi)\right)_{\sigma,\mathfrak{w}_r(\sigma)} \right)\\&=\langle \left( D_{\sigma}\left(\partial_{\xi}\mathcal{M}_{x,\xi}^{[k]}\right)_{\sigma,\mathfrak{w}_r(\sigma)}+\left(\partial_{\xi}\mathcal{M}_{x,\xi}^{[k]}\right)_{\sigma,\mathfrak{w}_r(\sigma)}D_{\sigma} \right)a_{n,0}^{[k]},u_{\sigma,\overline{\mathfrak{w}_r(\sigma)}}^{[k]} \rangle\\&-\left( \partial_{\xi}\mu^{[k]}(x,\xi)\right)_{\sigma,\mathfrak{w}_r(\sigma)}D_{\sigma}f_{n,0}(\sigma),
\end{align}
which implies
\begin{align}
    \langle \mathcal{N}_1a_{n,0}^{[k]},u_{\sigma,\overline{\mathfrak{w}_r(\sigma)}}^{[k]} \rangle _{L^2\left( \mathbb{R},d\tau \right)}&=\frac{1}{2}\Bigg( D_{\sigma}\left( \partial_{\xi}\mu^{[k]}(x,\xi)\right)_{\sigma,\mathfrak{w}_r(\sigma)} +\\& \left( \partial_{\xi}\mu^{[k]}(x,\xi)\right)_{\sigma,\mathfrak{w}_r(\sigma)} D_{\sigma} \Bigg)f_{n,0}+\mathfrak{R}_r(\sigma)f_{n,0},
\end{align}
with
$$ \mathfrak{R}_r(\sigma)=\langle\mathcal{R}_r(\sigma,\tau)u_{\sigma,\mathfrak{w}_r(\sigma)}^{[k]},u_{\sigma,\overline{\mathfrak{w}_r(\sigma)}}^{[k]}\rangle_{L^2\left( \mathbb{R},d\tau\right)}. $$
Therefore, $f_{n,0}$ verifies the transport equation
\begin{equation}
    \frac{1}{2}\left( D_{\sigma}\left( \partial_{\xi}\mu^{[k]}(x,\xi)\right)_{\sigma,\mathfrak{w}_r(\sigma)} +\left( \partial_{\xi}\mu^{[k]}(x,\xi)\right)_{\sigma,\mathfrak{w}_r(\sigma)} D_{\sigma} \right)f_{n,0}+\mathfrak{R}_r(\sigma)=\delta_{n,1}^{[k]}f_{n,0}.
    \label{Equation du transport}
\end{equation}
Considering the linearized equation near $\sigma=s_r$, we are led to choose $\delta_{n,1}^{[k]}$ in the set
$$ \text{sp}\left(  \frac{1}{2}\text{Hess} \mu^{[k]}(s_r,\xi_0^{[k]}\gamma_0^{\frac{1}{k+2}})(\sigma,D_{\sigma})+\mathfrak{R}_r(s_r) \right).$$
Using \eqref{matrice hessienne}, the Hessian matrix of $\mu^{[k]}$ at $(s_r,\xi_0^{[k]}\gamma_0^{\frac{1}{k+2}})$ is given by
$$ \text{Hess} \, \mu^{[k]}(s_r,\gamma_0^{\frac{1}{k+2}}\xi_0^{[k]})=\left(\begin{array}{cc}
\frac{2}{k+2}\gamma''(s_r)\gamma_0^{-\frac{k}{k+2}}\nu^{[k]}\left(\xi_0^{[k]}\right) & 0 \\ 0 & \left(\nu^{[k]}\right)''\left( \xi_0^{[k]} \right) \end{array} \right), $$
which gives us
\begin{align*}
    &\frac{1}{2}\text{Hess} \mu^{[k]}(s_r,\xi_0^{[k]}\gamma_0^{\frac{1}{k+2}})(\sigma,D_{\sigma})\\&=\frac{1}{2}\left(\begin{array}{cc}
\frac{2}{k+2}\gamma''(s_r)\gamma_0^{-\frac{k}{k+2}}\nu^{[k]}\left(\xi_0^{[k]}\right) & 0 \\ 0 & \left(\nu^{[k]}\right)''\left( \xi_0^{[k]} \right) \end{array} \right)\left(\begin{array}{c} \sigma \\ D_{\sigma}
\end{array} \right).\left(\begin{array}{c} \sigma \\ D_{\sigma}
\end{array} \right)\\&=\frac{1}{2}\left(\nu^{[k]}\right)''\left( \xi_0^{[k]} \right)\left(D_{\sigma}^2+\left( \zeta \sigma\right)^2\right),
\end{align*}
with $\zeta$ is given by
\begin{equation}
    \zeta= \sqrt{\frac{2}{k+2}\frac{\gamma''(s_r)\nu^{[k]}\left(\xi_0^{[k]}\right)}{\gamma_0^{\frac{k}{k+2}}\left(\nu^{[k]} \right)''\left(\xi_0^{[k]}\right)}}.
    \label{zeta}
\end{equation}
Recalling that the spectrum of the harmonic oscillator $D_{\sigma}^2+\left( \zeta \sigma\right)^2$ is given by 
$$ \left\{ (2n-1)\zeta,\,\,n\in\mathbb{N}^{\ast} \right\}, $$
we get 
$$ \delta_{n,1}^{[k]}=\left(n-\frac{1}{2} \right)\left(\nu^{[k]}\right)''\left( \xi_0^{[k]} \right) \sqrt{\frac{2}{k+2}\frac{\gamma''(s_r)\nu^{[k]}\left(\xi_0^{[k]}\right)}{\gamma_0^{\frac{k}{k+2}}\left(\nu^{[k]} \right)''\left(\xi_0^{[k]}\right)}}+\mathfrak{R}_r(s_r). $$

Let us come back to
$$  \left( \mathcal{N}_0-\delta_{n,0}^{[k]} \right)a_{n,1}^{[k]}=\left( \delta_{n,1}^{[k]}- \mathcal{N}_1\right) a_{n,0}^{[k]}, $$
where $a_{n,0}^{[k]}(\sigma,\tau)=f_{n,0}(\sigma)u_{\sigma,\mathfrak{w}_r(\sigma)}^{[k]}(\tau)$. Then we take $a_{n,1}^{[k]}$ as 
$$ a_{n,1}^{[k]}(\sigma,\tau)=f_{n,1}(\sigma)u_{\sigma,\mathfrak{w}_r(\sigma)}^{[k]}(\tau)+\tilde{a}_{n,1}^{[k]}(\sigma,\tau), $$
where
$$ \tilde{a}_{n,1}^{[k]}\in \left( \text{Ker} \left( \mathcal{N}_0-\mu_0^{[k]} \right) \right)^{\perp}. $$
The procedure can be continued by induction.
\end{proof}

\begin{rmk}
By \eqref{Eikonale equation} and using the fact that $\xi_0^{[k]}-\Im\left(\varphi(\sigma) \right)$ is bounded below and that $\gamma(\sigma)^{-\frac{1}{k+2}}\Phi_r'(\sigma)=\Re \varphi_r(\sigma)$ is sufficiently small, we can apply \cite[Theorem 1.2]{BHR-holomorphic} to the function $\nu^{[k]}$ and we obtain that the exact solution of the eikonal equation verifies that
\begin{equation}
    \Phi_r'(\sigma)\geq \nu^{[k]}\left( \xi_0^{[k]} \right)\left( \gamma(\sigma)^{\frac{2}{k+2}}-\gamma_0^{\frac{2}{k+2}} \right).
    \label{rmk00}
\end{equation}

\end{rmk}

\begin{rmk}\textbf{(Solving \eqref{Equation du transport} and normalization of $\Psi_{h,r}^{[k]}$)}
In the expression of the tunneling effect that we will write at the end, we need to find the explicit form (a priori in terms of $\varphi_r$) of solution $f_{1,0}$ of the transport equation \eqref{Equation du transport}. This equation can be written as follows
\begin{equation}
    \partial_{\sigma}f_{1,0}+\frac{\mathfrak{V}_r'(\sigma)+ 2\mathfrak{R}_r(\sigma)-2\delta_{n,1}^{[k]}}{2\mathfrak{V}_r(\sigma)}f_{1,0}=0.
    \label{transpor}
\end{equation}
where
\begin{equation}
    \mathfrak{V}_r(\sigma):=-\mathrm{i} \partial_{\xi}\mu^{[k]}(\sigma,\mathfrak{w}_r(\sigma)).
\end{equation}
We may write $f_{1,0}$ in the form $f_{1,0}(\sigma)=\mathrm{e}^{\mathrm{i}\alpha_{1,0}(\sigma)}\tilde{f}_{1,0}(\sigma)$ with $\tilde{f}_{1,0}$ and $\alpha_{1,0}$ are real-valued functions such that $\tilde{f}_{1,0}(0)>0$. From \eqref{transpor}, $\tilde{f}_{n,0}$ solves the real classical transport equation
\begin{equation*}
    \partial_{\sigma}\tilde{f}_{1,0}+\Re \left( \frac{\mathfrak{V}_r'(\sigma)+ 2\mathfrak{R}_r(\sigma)-2\delta_{n,1}^{[k]}}{2\mathfrak{V}_r(\sigma)}\right)\tilde{f}_{1,0}=0.
    \label{transport3}
\end{equation*}
Then, we get
\begin{equation*}
    \tilde{f}_{1,0}(\sigma)=K_0\operatorname{exp}\left( -\int_{s_r}^{\sigma}\Re \left( \frac{\mathfrak{V}_r'(s)+ 2\mathfrak{R}_r(s)-2\delta_{n,1}^{[k]}}{2\mathfrak{V}_r(s)}\right)ds \right),
\end{equation*}
and the constant $K_0$ is chosen so that the WKB solution $\Psi_{h,r}^{[k]}$ in Theorem \ref{BKW} is almost normalized. Following e.g.\cite[Lemma 2.1]{BHR-circle}, we choose $K_0$ so that $1=K_0^2\sqrt{\frac{\pi}{\Phi''(s_r)}}$, which allows us to choose $K_0$ as
$$ K_0=\left( \frac{\Phi_r''(s_r)}{\pi} \right)^{1/4}=\left( \frac{\zeta}{\pi} \right)^{1/4}, $$
with $\zeta$ is defined in \eqref{zeta}. Therefore,
\begin{equation}
    \tilde{f}_{1,0}^2(0)=\sqrt{\frac{\zeta}{\pi}}\mathrm{A}_u \,\,\,\,\text{and}\,\,\,\,\mathrm{A}_u:=\operatorname{exp}\left( -\int_{s_r}^{0}\Re\left(\frac{\mathfrak{V}_r'(s)+2\mathfrak{R}_r(s)-2\delta_{n,1}^{[k]}}{2\mathfrak{V}_r(s)}\right)ds \right).
\end{equation}

From \eqref{transpor}, the phase shifts $\alpha_{1,0}$ are chosen so that
$$ \alpha_{1,0}'(s)=-\Im \left( \frac{\mathfrak{V}_r'(\sigma)+ 2\mathfrak{R}_r(\sigma)-2\delta_{n,1}^{[k]}}{2\mathfrak{V}_r(\sigma)}\right). $$
Noticing that $\mathfrak{V}_r'(s_r)+ 2\mathfrak{R}_r(s_r)-2\delta_{n,1}^{[k]}=0$ and $\mathfrak{V}_r$ vanishes linearly at $s_r$, the function $\alpha_{1,0}'(s)$ can be considered as a smooth function at $s_r$. This shows that we have determined the phase shift $\alpha_{1,0}$ up to an additive constant. Then, we define
\begin{equation}
    \alpha_0:=\frac{\alpha_{1,0}(0)-\alpha_{1,0}(-L)}{L}.
    \label{alpha0}
\end{equation}
\label{rmk}
\end{rmk}
\begin{rmk}
By the symmetry defined in \eqref{sym}, we define the functions attached to the left well by
$$ \mathfrak{w}_l(\sigma):=\overline{\mathfrak{w}_r(-\sigma)}\,\,\,\,\text{and}\,\,\,\,\,\mathfrak{R}_l(\sigma):=\overline{\mathfrak{R}_r(-\sigma)}, $$
and the function $\mathfrak{V}_l$ by
$$ \mathfrak{V}_l(\sigma)=-\mathrm{i}\partial_{\xi}\mu^{[k]}\left( \sigma,\mathfrak{w}_l(\sigma) \right). $$

\label{rmk2}
\end{rmk}

\section{A Grushin problem}
\label{section5}

In this section, we introduce pseudo-differential calculus with operator-valued symbols and perform a pseudo-differential dimensional reduction using Grushin's method. This method is already used in \cite[Chapter 3]{keraval}, \cite{BHR-purely}, and the importance of this method is that it gives optimal decay estimates consistent with the WKB expansions.

In this section, we consider again the right-well operator $\mathcal{N}_{h,r}^{[k]}$, introduced in \eqref{puit droite}. To simplify the notations, we will omit the reference to ``right well" in the notation and write $\mathcal{N}_h^{[k]},\gamma,\delta,\tilde{\delta},\kappa$ instead of $\mathcal{N}_{h,r}^{[k]},\gamma_r,\delta_r,\tilde{\delta}_r,\kappa_r$. We also denote $\varphi$ instead of $\varphi_r$, which has been defined in Lemma \ref{Lemme}.

\subsection{Sub-solution of the eikonal equation}
To obtain the optimal estimates for the ground states of $\mathcal{N}_h^{[k]}$, we will consider an exponential weight defined as a sub-solution of the eikonal equation \eqref{Eikonale equation}. For this we consider a non-negative Lipchitzian function, $\sigma\mapsto\Phi(\sigma)$, satisfying the following hypothesis:
\begin{hyp}
For all $M>0$ there exist $h_0,C,R>0$ such that, for all $h\in(0,h_0)$, the function $\Phi$ satisfies
\begin{enumerate}
    \item[(i)]For all $\sigma\in \mathbb{R}$, we have
    $$ \Re \left( \gamma(\sigma)^{\frac{2}{k+2}}\nu^{[k]}\left( \xi_0^{[k]}-\Im\left(\varphi(\sigma) \right)+\mathrm{i}\gamma(\sigma)^{-\frac{1}{k+2}}\Phi'(\sigma)\right)-\gamma_0^{\frac{2}{k+2}}\nu^{[k]}(\xi_0^{[k]})\right) \geq 0, $$
    \item[(ii)]For all $\sigma\in\mathbb{R}$ such that $|\sigma-s_r|\geq Rh^{1/2}$, we have
    $$ \Re \left( \gamma(\sigma)^{\frac{2}{k+2}}\nu^{[k]}\left( \xi_0^{[k]}-\Im\left(\varphi(\sigma) \right)+\mathrm{i}\gamma(\sigma)^{-\frac{1}{k+2}}\Phi'(\sigma)\right)-\gamma_0^{\frac{2}{k+2}}\nu^{[k]}(\xi_0^{[k]})\right)\geq Mh, $$
    \item[(iii)]For all $\sigma\in\mathbb{R}$ such that $|\sigma-s_r|\leq Rh^{1/2}$, we have
    $$ \left| \Phi(\sigma) \right|\leq Mh. $$
\end{enumerate}
\label{hyp2}
\end{hyp}
\begin{rmk}
The function 
    $$ \Phi(\sigma)=\sqrt{\frac{\nu^{[k]}(\xi_0^{[k]})}{2}}\displaystyle\int_{s_r}^{\sigma}\sqrt{\gamma_r(\tilde{\sigma})^{\frac{2}{k+2}}-\gamma_0^{\frac{2}{k+2}}}d\tilde{\sigma} $$
    verifies Hypothesis \ref{hyp2}. Indeed, using the fact that $\xi_0^{[k]}-\Im\left(\varphi(\sigma) \right)$ is bounded below and that $\gamma(\sigma)^{-\frac{1}{k+2}}\Phi'(\sigma)$ is sufficiently small, we can apply \cite[Theorem 1.2]{BHR-holomorphic} to the function $\nu^{[k]}$ and we obtain
    \begin{align*}
        &\Re \left( \gamma(\sigma)^{\frac{2}{k+2}}\nu^{[k]}\left( \xi_0^{[k]}-\Im\left(\varphi(\sigma) \right)+\mathrm{i}\gamma(\sigma)^{-\frac{1}{k+2}}\Phi'(\sigma)\right)-\gamma_0^{\frac{2}{k+2}}\nu^{[k]}(\xi_0^{[k]})\right)\\& \geq \frac{1}{2}\nu^{[k]}(\xi_0^{[k]})\left(\gamma_r(\sigma)^{\frac{2}{k+2}}-\gamma_0^{\frac{2}{k+2}}\right),
    \end{align*}
    and Hypothesis \ref{hyp2} is well verified using the fact that the function $\gamma_r$ has a unique non-degenerate minimum at $s_r$. But it should be noted that this does not give us the optimal Agmon estimates (and after the optimal approximations of the eigenfunctions), it is necessary to construct weight functions that are related to the exact solution of the eikonal equation. Much more useful solutions will be presented in the following proposition.
\end{rmk}
The following proposition shows the weight functions which satisfy Assumption \ref{hyp2}.
\begin{prop}
We consider the function $\mathfrak{v}_r$ defined on $\mathbb{R}$ by
$$ \mathfrak{v}_r(\sigma)=\frac{1}{2}\nu^{[k]}(\xi_0^{[k]})\left(\gamma_r(\sigma)^{\frac{2}{k+2}}-\gamma_0^{\frac{2}{k+2}}\right).  $$
By the hypothesis on $\gamma_r$, we can choose $c_0>0$ such that 
$$ \mathfrak{v}_r(\sigma)\geq c_0(\sigma-s_r)^2 \text{  and  }\Phi_r(\sigma)\geq c_0(\sigma-s_r)^2\text{  for all  }\sigma\in\operatorname{B}_l(L-\eta). $$
The following functions verify Assumption \ref{hyp2} :
\begin{enumerate}
\item[(a)]For $\epsilon\in (0,1)$, 
$$\Phi_{r,\epsilon}=\sqrt{1-\epsilon}\Phi_r \,\,\, \text{with} \,\,\,\,R>0\,\,\,\, \text{and} \,\,\,\,M=c_0\epsilon R^2.$$
\item[(b)]For $N\in \mathbb{N}^\ast$ and $h\in (0,1)$,
$$ \tilde{\Phi}_{r,N,h}=\Phi_{r,R}-Nh\ln\left( \max \left( \frac{\Phi_r}{h},N\right)\right)\,\,\,\text{with}\,\,\,\,R=\sqrt{\frac{N}{c_0}}\,\,\,\,\text{and}\,\,\,\,M=N\inf\frac{\mathfrak{v}_r}{\Phi_r}. $$
\item[(c)]For $\epsilon \in (0,1)$, $N\in \mathbb{N}$ and $h\in (0,1)$,
$$ \hat{\Phi}_{r,N,h}(s)=\min\left\{ \tilde{\Phi}_{r,N,h}(s), \sqrt{1-\epsilon} \underset{t\in \operatorname{supp} \chi_r '}{\inf}\left(\Phi_r(t)+\int_{[s_r,t]} \gamma(\tilde{\sigma})^{\frac{1}{k+2}}\Re \varphi_r(\tilde{\sigma}) d\tilde{\sigma}\right)\right\},  $$
with $R=\sqrt{\frac{N}{c_0}}$ and $M=N\min\left(\epsilon,\inf\frac{\mathfrak{v}_r}{\Phi_r}\right)$, where $\operatorname{supp}\chi_r ' \subset I_{\eta,r}\setminus I_{2\eta,r}$.
\end{enumerate}
\label{poidss}
\end{prop}
\begin{proof}
Since $\Phi_r$ verifies \eqref{rmk00} and the function $\gamma_r$ admits a unique non-degenerate minimum at $s_r$, the existence of $c_0>0$ is well guaranteed.

We recall that the function $\Phi_r$ verifies the eikonal equation \eqref{Eikonale equation}, and by Lemma \ref{Lemme}, the function $\Phi_r$ is defined by
\begin{equation}
    \Phi_r(\sigma)=\displaystyle \int_{s_r}^{\sigma}\gamma(\tilde{\sigma})^{\frac{1}{k+2}}\Re \varphi_r(\tilde{\sigma})d\tilde{\sigma},
    \label{vv}
\end{equation}
where $\varphi_r$ verify \eqref{R} and \eqref{I}, with
$$ \left|\mathfrak{f}_r(\sigma)\right|=\sqrt{\nu^{[k]}(\xi_0^{[k]})}\sqrt{1-\left( \frac{\gamma_0}{\gamma_r} \right)^{\frac{2}{k+2}}}. $$
According to the chosen hypothesis on $\gamma_r$, $\left|\mathfrak{f}_r(\sigma)\right|$ is small enough for all $\sigma\in\mathbb{R}$ and so, by \eqref{R} and \eqref{I}, $\left|\varphi_r(\sigma)\right|$ is small enough for all $\sigma\in\mathbb{R}$.
\begin{enumerate}
    \item [(a)]By \eqref{Eikonale equation} and \eqref{vv}, $\varphi_r$ verify that
    $$ \gamma_r(\sigma)^{\frac{2}{k+2}}\nu^{[k]}\left( \xi_0^{[k]}+\mathrm{i}\varphi_r(\sigma) \right)=\gamma_0^{\frac{2}{k+2}}\nu^{[k]}(\xi_0^{[k]}).  $$
Then, by the expression of $\Phi_{r,\epsilon}$, we get
    \begin{align*}
        &\Re\left( \gamma_r(\sigma)^{\frac{2}{k+2}}\nu^{[k]}\left( \xi_0^{[k]}-\Im\left( \varphi_r(\sigma) \right)+\mathrm{i}\gamma(\sigma)^{-\frac{1}{k+2}}\Phi_{r,\epsilon}'(\sigma)\right)-\gamma_0^{\frac{2}{k+2}}\nu^{[k]}(\xi_0^{[k]})\right)\\&=\gamma(\sigma)^{\frac{2}{k+2}}\Re\left\{ \nu^{[k]}\left( \xi_0^{[k]}-\Im\left( \varphi_r(\sigma)\right) +\mathrm{i}\sqrt{1-\epsilon}\Re \varphi_r(\sigma)\right)-\nu^{[k]}\left( \xi_0^{[k]}+\mathrm{i}\varphi_r(\sigma) \right) \right\}.
    \end{align*}
Using the Taylor expansion for the function $\nu^{[k]}$ in a neighborhood of $\xi_0^{[k]}$, we get
\begin{align*}
    &\Re \left( \nu^{[k]}\left( \xi_0^{[k]}-\Im\left( \varphi_r(\sigma)\right)+\mathrm{i}\sqrt{1-\epsilon}\Re \varphi_r(\sigma)\right)-\nu^{[k]}\left( \xi_0^{[k]}+\mathrm{i}\varphi_r(\sigma)\right)\right) \\&=\Re \left\{ \displaystyle\sum_{n\geq2}\frac{\left( \nu^{[k]} \right)^{(n)}\left( \xi_0^{[k]} \right)}{n!}\left( \left(-\Im\left( \varphi_r(\sigma)\right)+\mathrm{i}\sqrt{1-\epsilon}\Re \varphi_r(\sigma)  \right)^n- \left( \mathrm{i}\varphi_r(\sigma) \right)^n \right) \right\}\\&=\displaystyle\sum_{n\geq2}\frac{\left( \nu^{[k]} \right)^{(n)}\left( \xi_0^{[k]} \right)}{n!}\Re \left\{ \left(-\Im\left( \varphi_r(\sigma)\right)+\mathrm{i}\sqrt{1-\epsilon}\Re \varphi_r(\sigma)  \right)^n- \left(-\Im\left( \varphi_r(\sigma) \right)+ \mathrm{i}\Re\varphi_r(\sigma) \right)^n \right\}.
\end{align*}
Recall that, for all $a,b_1,b_2\in\mathbb{R}$ and for all $n\in\mathbb{N}\setminus \{0,1\}$, we have
\begin{equation}
\Re \left\{ (a+\mathrm{i}b_1)^n-(a+\mathrm{i}b_2)^n \right\}=(b_2^2-b_1^2)\displaystyle \sum_{j=1}^{\lfloor \frac{n}{2} \rfloor}(-1)^{j+1}\mathrm{C}_n^{2j}a^{n -2j}\left( \displaystyle\sum_{l=0}^{j-1}b_1^{2l}b_2^{2j-2l-2} \right).
    \label{ab1b2}
\end{equation}
Using \eqref{ab1b2}, \eqref{R}, \eqref{I} and the fact that $\left| \mathfrak{f}_r(\sigma) \right|$ is small enough, we get
 \begin{align*}
     &\Re \left\{ \left(-\Im\left( \varphi_r(\sigma)\right)+\mathrm{i}\sqrt{1-\epsilon}\Re \varphi_r(\sigma)  \right)^n- \left(-\Im\left( \varphi_r(\sigma) \right)+ \mathrm{i}\Re\varphi_r(\sigma) \right)^n \right\}\\&=\begin{cases}
     \epsilon\Re \varphi_r^2(\sigma) \,\,\,\,\,\,\,\,\,\,\,\,\,\,\,\,\,\,\,\,\,\,\,\,\,\,\,\,\,\,\,\,\,\text{if}\,\,\,\,\,n=2,\\
     \epsilon\Re \varphi_r^2(\sigma)\mathcal{O}\left(  \mathfrak{f}_r(\sigma) ^2 \right)\,\,\,\,\text{if}\,\,\,\,\,n\geq2,
     \end{cases}
 \end{align*}
which implies that
\begin{align*}
    &\Re \left( \nu^{[k]}\left( \xi_0^{[k]}-\Im\left( \varphi_r(\sigma)\right)+\mathrm{i}\sqrt{1-\epsilon}\Re \varphi_r(\sigma)\right)-\nu^{[k]}\left( \xi_0^{[k]}+\mathrm{i}\varphi_r(\sigma)\right)\right)\\&=\epsilon\mathfrak{f}_r(\sigma)^2+\epsilon\mathfrak{f}_r(\sigma)^2\mathcal{O}\left( \mathfrak{f}_r(\sigma) ^2 \right).
\end{align*}

    Therefore, 
    \begin{align*}
        \gamma(\sigma)^{\frac{2}{k+2}}&\Re\left\{ \nu^{[k]}\left( \xi_0^{[k]}-\Im\left( \varphi_r(\sigma)\right) +\mathrm{i}\sqrt{1-\epsilon}\Re \varphi_r(\sigma)\right)-\nu^{[k]}\left( \xi_0^{[k]}+\mathrm{i}\varphi_r(\sigma) \right) \right\}\\&=\epsilon \gamma(\sigma)^{\frac{2}{k+2}}\mathfrak{f}_r(\sigma)^2\left( 1+\mathcal{O}\left( \mathfrak{f}_r(\sigma)^2 \right)  \right) \\&\geq \epsilon \mathfrak{v}_r(\sigma),
    \end{align*}
    and, for all $\sigma\in\mathbb{R}$ such that $\left| \sigma-s_r \right|\geq Rh^{1/2}$, we have
    $$ \Re\left( \gamma(\sigma)^{\frac{2}{k+2}}\nu^{[k]}\left( \xi_0^{[k]}-\Im\left( \varphi_r(\sigma) \right)+\mathrm{i}\gamma(\sigma)^{-\frac{1}{k+2}}\Phi_{r,\epsilon}'(\sigma)\right)-\gamma_0^{\frac{2}{k+2}}\nu^{[k]}(\xi_0^{[k]})\right)\geq \epsilon c_0R^2h. $$
    \item[(b)]For all $N\in\mathbb{N}$ and $h\in(0,h_0)$, we have
    $$ \tilde{\Phi}_{r,N,h}' = \begin{cases}
\Phi_r '\left( 1-\frac{Nh}{\Phi_r} \right) & \,\,\,\text{if}\,\,\,\frac{\Phi_r}{h}\geq N \,,\\
\Phi_r ' & \,\,\,\text{if}\,\,\,\frac{\Phi_r}{h}<N\,. \end{cases} $$
Then, on $\left\{ \Phi_r\geq Nh \right\}$, we have
\begin{align*}
    &\Re\left( \gamma(\sigma)^{\frac{2}{k+2}}\nu^{[k]}\left( \xi_0^{[k]}-\Im\left( \varphi_r(\sigma) \right)+\mathrm{i}\gamma(\sigma)^{-\frac{1}{k+2}}\tilde{\Phi}_{r,N,h}'(\sigma)\right)-\gamma_0^{\frac{2}{k+2}}\nu^{[k]}(\xi_0^{[k]})\right)\\&=\gamma(\sigma)^{\frac{2}{k+2}}\Re\left\{ \nu^{[k]}\left( \xi_0^{[k]}-\Im\left( \varphi_r(\sigma)\right) +\mathrm{i}\left( 1-\frac{Nh}{\Phi_r} \right)\Re \varphi_r(\sigma)\right)-\nu^{[k]}\left( \xi_0^{[k]}+\mathrm{i}\varphi_r(\sigma) \right) \right\}.
\end{align*}
Similarly as part (a), on $\left\{ \Phi_r\geq Nh \right\}$, we get
\begin{align*}
    \Re&\left( \gamma(\sigma)^{\frac{2}{k+2}}\nu^{[k]}\left( \xi_0^{[k]}-\Im\left( \varphi_r(\sigma) \right)+\mathrm{i}\gamma(\sigma)^{-\frac{1}{k+2}}\tilde{\Phi}_{r,N,h}'(\sigma)\right)-\gamma_0^{\frac{2}{k+2}}\nu^{[k]}(\xi_0^{[k]})\right)\\&\geq \frac{Nh}{\Phi_r}\left( 2-\frac{Nh}{\Phi_r} \right)\mathfrak{v}_r(\sigma)\\&\geq  \frac{Nh}{\Phi_r}\mathfrak{v}_r(\sigma)\geq c_1Nh,
\end{align*}
with $c_1=\underset{\sigma\in \mathbb{R}}{\text{inf}}\frac{\mathfrak{v}_r}{\Phi_r}>0$.\\
Let $R\geq R_0=\sqrt{\frac{N}{c_0}} $. we have
$$ \left| \sigma-s_r\right| \geq Rh^{1/2} \Rightarrow \Phi_{r,R}\geq c_0 R^2h \geq Nh, $$
which implies that for all $\sigma\in \mathbb{R}$ such that $\left| \sigma-s_r\right| \geq Rh^{1/2}$, we have
$$ \Re\left( \gamma(\sigma)^{\frac{2}{k+2}}\nu^{[k]}\left( \xi_0^{[k]}-\Im\left( \varphi_r(\sigma) \right)+\mathrm{i}\gamma(\sigma)^{-\frac{1}{k+2}}\tilde{\Phi}_{r,N,h}'(\sigma)\right)-\gamma_0^{\frac{2}{k+2}}\nu^{[k]}(\xi_0^{[k]})\right)\geq Mh. $$
\item[(c)]It exists $t_0\in \text{supp}(\chi_r ')$ such that
$$ \underset{t\in \text{supp} \chi_r '}{\text{inf}}\left(\Phi_r(t)+\int_{s_r}^{t}\gamma(\tilde{\sigma})^{\frac{1}{k+2}}\Re\varphi_r(\tilde{\sigma}) d\tilde{\sigma}\right)=\Phi_r(t_0)+\int_{s_r}^{t_0}\gamma(\tilde{\sigma})^{\frac{1}{k+2}}\Re\varphi_r(\tilde{\sigma}) d\tilde{\sigma}. $$
Then,
$$ \left| \hat{\Phi}_{r,N,h} ' \right|=\left| \tilde{\Phi}_{r,N,h} ' \right| \,\,\,\text{or}\,\,\,\,\left| \hat{\Phi}_{r,N,h} ' \right|=\sqrt{1-\epsilon}\left| \Phi_r' \right|=\Phi_{r,\epsilon}'. $$
Therefore, $\hat{\Phi}_{r,N,h}$ verifies Assumption \ref{hyp2}.
\end{enumerate}
\end{proof}

\subsection{A pseudo-differential operator with operator-valued symbol.}
We consider the conjugate operator 
$$\mathcal{N}_h^{[k],\phi}=\mathrm{e}^{\frac{\Phi}{h}}\mathcal{N}_h^{[k]}\mathrm{e}^{-\frac{\phi}{h}},$$
with the same domaine as $\mathcal{N}_h^{[k]}$. It needs
$$ \mathcal{N}_h^{[k],\Phi}=\mathfrak{a}_h^{-1}D_{\tau}a_hD_{\tau}+\mathfrak{a}_h^{-1}\left( hD_{\sigma}-\mathcal{A}_h^{[k],\Phi}(\sigma,\tau)\right)\mathfrak{a}_h^{-1}\left( hD_{\sigma}-\mathcal{A}_h^{[k],\Phi}(\sigma,\tau)\right), $$
with
$$ \mathcal{A}_h^{[k],\Phi}(\sigma,\tau)=-\mathrm{i}\Phi'(\sigma)+\gamma(\sigma)\frac{\tau^{k+1}}{k+1}+h\tilde\delta(\sigma)\frac{\tau^{k+2}}{k+2}c_{\mu}+h^2 c_{\mu}\mathcal{O}(\tau^{k+3}). $$

We recall that for a symbol $a(\sigma,\xi)\in\mathcal{S}\left( \mathbb{R}^2\right)$, the Weyl quantization of $a$ is the operator $\text{Op}_h^\text{W}(a)$ defined, for all $u\in\mathcal{S}\left(\mathbb{R}_{\sigma};\mathcal{S}\left(\mathbb{R}_{\tau}\right)\right)$, by
$$ \text{Op}_h^\text{W}(a)u(\sigma):=\frac{1}{2\pi h}\displaystyle \int \displaystyle \int_{\mathbb{R}^2}\mathrm{e}^{\frac{\mathrm{i}}{h}(\sigma-\tilde{\sigma}).\xi}a(\frac{\sigma+\tilde{\sigma}}{2},\xi)u(\tilde{\sigma})d\tilde{\sigma}d\xi. $$
Classical results of pseudo-differential calculus, for symbols with operator values, are already detailed in \cite[Chapter 2]{keraval}. We consider the real valued function $\mathfrak{g}$ defined by
$$ \mathfrak{g}(\sigma)=\int_0^{\sigma}\gamma(\tilde{\sigma})^{\frac{1}{k+2}} \left( \left( 1-\left(\frac{\gamma_0}{\gamma}\right)^{\frac{1}{k+2}}\right)\xi_0^{[k]}-  \Im \left( \varphi(\tilde{\sigma}) \right)\right) d\tilde{\sigma}. $$

\begin{rmk}
Note that, only in this section, this is not the same function as the one in \eqref{gsigma}. There is an addition of the term $-\gamma_0^{\frac{1}{k+2}}\xi_0^{[k]}$. This new function is more convening for the computations than \eqref{gsigma}.
\end{rmk}

After the gauge transformation $\mathrm{e}^{-\mathrm{i}\frac{\mathfrak{g}(\sigma)}{h}}$, we are led to work with the conjugate operator 
\begin{equation}
    \tilde{\mathcal{N}}_h^{[k],\Phi}=\mathrm{e}^{-\mathrm{i}\frac{\mathfrak{g}(\sigma)}{h}}\mathcal{N}_h^{[k],\Phi}\mathrm{e}^{\mathrm{i}\frac{\mathfrak{g}(\sigma)}{h}}
    \label{N-tilde}
\end{equation}
instead of $\mathcal{N}_h^{[k],\Phi}$. We notice that $\tilde{\mathcal{N}}_h^{[k],\Phi}$ can be written as an $h$-pseudo-differential operator with an operator valued symbol $n_h^{[k]}(\sigma,\xi)$ having an expansion in powers of $h$ :
$$ n_h^{[k]}=n_0+hn_1+h^2n_2+..., $$
with
$$ n_0=D_{\tau}^2+\left(\xi+\mathfrak{w}(\sigma)-\gamma(\sigma)\frac{\tau^{k+1}}{k+1}\right)^2, $$
$$ n_1=-2\tilde{\delta}(\sigma)\frac{\tau^{k+2}}{k+2}c_{\mu}\left( \xi+\mathfrak{w}(\sigma)-\gamma(\sigma)\frac{\tau^{k+1}}{k+1}\right)+2\tau c_{\mu}k\left(\xi+\mathfrak{w}(\sigma)-\gamma(\sigma) \frac{\tau^{k+1}}{k+1}\right)^2, $$
where $\mathfrak{w}(\sigma)=\mathfrak{g}'(\sigma)+\mathrm{i}\Phi'(\sigma)$ and the notation $\mathcal{O}$ is defined in \cite[Notation 3.1]{BHR-purely}.\\

The frequency variable $\xi$ is a priori unbounded. Then, as in \cite{BHR-purely}, $n_h$ can be replaced by a bounded symbol as long as nothing is changed near the minimum. For this, we consider the function defined on $\mathbb{R}$ by
$$ \chi_1(\xi)=\gamma_0^{\frac{1}{k+2}}\xi_0^{[k]}+\chi(\xi-\gamma_0^{\frac{1}{k+2}}\xi_0^{[k]}), $$
where $\chi\in\mathcal{C}^{\infty}\left(\mathbb{R}\right)$ is a function that verifies the following assertions:
\begin{enumerate}
    \item[(i)]the function $\chi$ is a smooth, bounded, increasing and odd function on $\mathbb{R}$.
    \item[(ii)]$\chi(\xi)=\xi$ on $[-1,1]$ and $\underset{\xi \rightarrow +\infty}{\text{lim}}\chi(\xi)=2$.
\end{enumerate}

We will consider
$$ \text{Op}_h^\text{W}(p_h)\,\,\,\,,\,\,\,\,\,\text{where}\,\,\,\,\,p_h(\sigma,\xi)=n_h(\sigma,\chi_1(\xi)). $$
The symbol $p_h$ has the same expansion in powers of $h$, except $\xi$ to replace with the truncation function $\chi_1(\xi)$.
\subsection{Solving the Grushin problem.}
For $z\in\mathbb{C}$, we define
$$ \mathcal{P}_z(\sigma,\xi)= \begin{pmatrix}
p_h-z & .v_{\sigma,\xi} \\
\langle .,v_{\sigma,\xi}\rangle & 0 \\
\end{pmatrix}\in \mathcal{S}\left(\mathbb{R}_{\sigma,\xi}^2,\mathcal{L}\left( \text{Dom}(p_0)\times\mathbb{C},L^2\left(\mathbb{R}\right)\times\mathbb{C}\right)\right), $$
see \cite[Notation 3.2]{BHR-purely}, where
\begin{equation}
    p_0:=\mathcal{M}_{\sigma,\chi_1(\xi)+\mathfrak{w}(\sigma)}^{[k]}=D_{\tau}^2+\left(\chi_1(\xi)+\mathfrak{w}(\sigma)-\gamma(\sigma)\frac{\tau^{k+1}}{k+1}\right)^2,
    \label{po}
\end{equation}
is the principal symbol of $p_h$ and $v_{\sigma,\xi}:=u_{\sigma,\chi_1(\xi)+\mathfrak{w}(\sigma)}^{[k]}$ is the eigenfunction associated with the smallest eigenvalue $\mu^{[k]}(\sigma,\chi_1(\xi)+\mathfrak{w}(\sigma))$ of $p_0$.

$\mathcal{P}_z$ decomposes in the form:
$$ \mathcal{P}_z=\mathcal{P}_{0,z}+h\mathcal{P}_1+h^2\mathcal{P}_2+..., $$
with
$$ \mathcal{P}_{0,z}(\sigma,\xi)=
 \begin{pmatrix}
p_0-z & .v_{\sigma,\xi} \\
\langle.,v_{\sigma,\xi}\rangle & 0 \\
\end{pmatrix}\,\,\,\,,\,\,\,\,\,\mathcal{P}_1=\begin{pmatrix}
p_1 & 0 \\
0 & 0
\end{pmatrix} $$
where
\begin{equation}
    p_1=-2\tilde{\delta}(\sigma)\frac{\tau^{k+2}}{k+2}c_{\mu}\left( \chi_1(\xi)+\mathfrak{w}(\sigma)-\gamma(\sigma)\frac{\tau^{k+1}}{k+1}\right)+2\tau c_{\mu}k\left(\chi_1(\xi)+\mathfrak{w}(\sigma)-\gamma(\sigma) \frac{\tau^{k+1}}{k+1}\right)^2.
    \label{p1}
\end{equation}

Let $z\in\mathbb{C}$ such that $\Re (z)\in(\mu_0^{[k]}-\varepsilon,\mu_0^{[k]}+\varepsilon)$, with $\varepsilon>0$ such that 
\begin{equation}
    \varepsilon<\frac{1}{2}\left(\underset{\xi\in\mathbb{R}}{\text{inf}}\nu_2^{[k]}(\xi)-\nu^{[k]}(\xi_0^{[k]})\right),
    \label{L2}
\end{equation}
where $\nu_2^{[k]}(\xi)$ is the second eigenvalue of the Montgomery operator $\mathfrak{h}_{\xi}^{[k]}$ for $\xi\in\mathbb{R}$. 
\begin{lemma}
For all $(\sigma,\xi)\in\mathbb{R}^2$, $\mathcal{P}_{0,z}(\sigma,\xi)$ is bijective and
$$\mathcal{Q}_{0,z}(\sigma,\xi):=\mathcal{P}_{0,z}^{-1}(\sigma,\xi)=
\begin{pmatrix}
(p_0-z)^{-1}\Pi^{\bot} & .v_{\sigma,\xi} \\
\langle.,v_{\sigma,\xi}\rangle & z-\mu^{[k]}(\sigma,\chi_1(\xi)+\mathfrak{w}(\sigma)) \\
\end{pmatrix},
$$
and
$$\mathcal{Q}_{0,z}(\sigma,\xi)\in S\left( \mathbb{R}_{\sigma,\xi}^2;\mathcal{L}\left(\operatorname{Dom}(p_0)\times\mathbb{C},L^2(\mathbb{R})\times\mathbb{C}\right)\right).$$
Here $\Pi=\Pi_{\sigma,\xi}$ is the orthogonal projection on $v_{\sigma,\xi}$ and $\Pi^{\perp}=Id-\Pi$.
\label{Lem grushin}
\end{lemma}

\begin{proof}
Let $(v,\beta)\in L^2(\mathbb{R}\times\mathbb{C})$ and find $(u,\alpha)\in \text{Dom}(p_0)\times \mathbb{C}$ such that 
\begin{equation}
    \mathcal{P}_{0,z}(\sigma,\xi)\begin{pmatrix} u \\ \alpha \end{pmatrix}=\begin{pmatrix} v \\ \beta \end{pmatrix}.
    \label{L1}
\end{equation}
This equation is equivalent to
$$  (p_0-z)u=v-\alpha v_{\sigma,\xi}\,\,\,\,\text{and}\,\,\,\,\,\langle u,v_{\sigma,\xi}\rangle =\beta. $$

We have 
\begin{align*}
(p_0-z)u^{\perp}&=(p_0-z)(u-\langle u,v_{\sigma,\xi}\rangle v_{\sigma,\xi})\\&=v-\alpha v_{\sigma,\xi}-\beta\left( \mu^{[k]}(\sigma,\chi_1(\xi)+\mathfrak{w}(\sigma))-z\right)v_{\sigma,\xi}.
\end{align*}
The space $\left( \mathbb{C}v_{\sigma,\xi} \right)^{\perp}$ is stable by $p_0-z$, then $p_0-z$ induces an operator 
$$p_0-z:\left( \mathbb{C}v_{\sigma,\xi} \right)^{\perp}\longrightarrow\left( \mathbb{C}v_{\sigma,\xi} \right)^{\perp}.$$
On this space, 
$$\langle (p_0-\Re(z))u,u \rangle \geq \left(\Re\left( \mu_2^{[k]}(\sigma,\chi_1(\xi)+\mathfrak{w}(\sigma))\right)-\Re(z)\right)\left\|u\right\|^2\geq c_0 \left\|u\right\|^2,$$
by the choice of $z$. Indeed, applying \cite[Theorem 1.2]{BHR-holomorphic} to the function $\nu_2^{[k]}$ (see also Remark 1.3 and 1.4 in \cite{BHR-holomorphic}), using \eqref{L2} and the fact that $\left|\Phi'(\sigma)\right|$ is small enough for all $\sigma\in\mathbb{R}$ (according to Assumption \ref{Hyp0} and the choice of $\Phi$ in Proposition \ref{poidss}), we get
\begin{align*}
&\Re\left( \mu_2^{[k]}(\sigma,\chi_1(\xi)+\mathfrak{w}(\sigma))\right)-\Re(z)\\&=\Re \Bigg( \gamma(\sigma)^{\frac{2}{k+2}}\nu_2^{[k]} \bigg( \gamma(\sigma)^{-\frac{1}{k+2}}\chi_1(\xi)+\gamma(\sigma)^{-\frac{1}{k+2}} \mathfrak{g}'(\sigma)+\mathrm{i}\gamma(\sigma)^{-\frac{1}{k+2}}\Phi'(\sigma)  \bigg)-\mu_0^{[k]}-\varepsilon \Bigg)\\&\geq \Re \Bigg( \gamma(\sigma)^{\frac{2}{k+2}}\nu_2^{[k]} \bigg( \gamma(\sigma)^{-\frac{1}{k+2}}\chi_1(\xi)+\gamma(\sigma)^{-\frac{1}{k+2}}\mathfrak{g}'(\sigma) \bigg)-\Phi'(\sigma)^2-\mu_0^{[k]}-\varepsilon \Bigg)
\\&\geq  \gamma_0^{\frac{2}{k+2}}\left( \underset{\xi\in\mathbb{R}}{\text{inf}} \nu_2^{[k]}(\xi)-\underset{\xi\in\mathbb{R}}{\text{inf}} \nu_1^{[k]}(\xi)  \right)-\Phi'(\sigma)^2-\varepsilon
\\&\geq \frac{\gamma_0^{\frac{2}{k+2}}}{2}\left( \underset{\xi\in\mathbb{R}}{\text{inf}} \nu_2^{[k]}(\xi)-\underset{\xi\in\mathbb{R}}{\text{inf}} \nu_1^{[k]}(\xi)  \right)-\Phi'(\sigma)^2\geq c_0,
\end{align*}
where $c_0>0$. Thus, this operator is injective with closed range and, by considering the adjoint, it is bijective. We have:
\begin{align*}
&(p_0-z)u^{\perp}=v-\alpha v_{\sigma,\xi}-\beta\left( \mu^{[k]}(\sigma,\chi_1(\xi)+\mathfrak{w}(\sigma))-z\right)v_{\sigma,\xi}\in \left( \mathbb{C}v_{\sigma,\xi} \right)^{\bot}\\&\implies \langle v,v_{\sigma,\xi}\rangle -\alpha-\beta\left( \mu^{[k]}(\sigma,\chi_1(\xi)+\mathfrak{w}(\sigma))-z\right)=0\\&\implies \alpha=\langle v,v_{\sigma,\xi}\rangle -\beta\left( \mu^{[k]}(\sigma,\chi_1(\xi)+\mathfrak{w}(\sigma))-z\right).
\end{align*}
By the bijectivity of $p_0-z$ on $\left( \mathbb{C}v_{\sigma,\xi} \right)^{\bot}$, we take 
\begin{align*}
u^{\perp}&=(p_0-z)^{-1}\left( v-\alpha v_{\sigma,\xi}-\beta\left( \mu^{[k]}(\sigma,\chi_1(\xi)+\mathfrak{w}(\sigma))-z\right)v_{\sigma,\xi}  \right)\\&=(p_0-z)^{-1}\left( v-\langle v,v_{\sigma,\xi}\rangle v_{\sigma,\xi}\right)\\&=(p_0-z)^{-1}\Pi^{\perp}v.
\end{align*}
Therefore, $u=\langle u,v_{\sigma,\xi}\rangle v_{\sigma,\xi}+u^{\perp} = \beta v_{\sigma,\xi}+(p_0-z)^{-1}\Pi^{\perp}v$.
\end{proof}

The following proposition gives an expression of an approximative inverse of operator $\text{Op}_h^\text{W}\left(\mathcal{P}_z\right)$ with a remainder of order $h$.
\begin{prop}
We have 
\begin{equation}
     \operatorname{Op}_h^{\operatorname{W}}\left( \mathcal{Q}_{0,z}\right) \operatorname{Op}_h^{\operatorname{W}}\left(\mathcal{P}_z\right)=\operatorname{Id}+h\mathcal{O}(\langle \tau \rangle^{2k+3}).
    \label{inverse}
\end{equation}
Moreover, if we denote by
$$ \mathcal{Q}_{0,z}:=\begin{pmatrix}
q_{0,z} & q_{0,z}^+ \\
q_{0,z}^- & q_{0,z}^{\pm}
\end{pmatrix}, $$
then modulo some remainders of order h, we have
\begin{equation}
    \left( \operatorname{Op}_h^{\operatorname{W}}(p_h)-z\right)^{-1}= \operatorname{Op}_h^{\operatorname{W}}q_{0,z}- \operatorname{Op}_h^{\operatorname{W}}q_{0,z}^-\left( \operatorname{Op}_h^{\operatorname{W}}q_{0,z}^{\pm}\right)^{-1} \operatorname{Op}_h^{\operatorname{W}}q_{0,z}^+.
    \label{inverse1}
\end{equation}
\label{prop1}
\end{prop}
\begin{proof}
Using Lemma \ref{Lem grushin}, and composition of pseudo-differential operators, we have $$\operatorname{Op}_h^{\operatorname{W}}\left( \mathcal{Q}_{0,z}\right)\circ \operatorname{Op}_h^{\operatorname{W}}\left(\mathcal{P}_{0,z}\right)=\operatorname{Op}_h^{\operatorname{W}}\left(\mathcal{D}_{0,z}\right),$$
with $\mathcal{D}_{0,z}=\text{Id}+h\tilde{\mathcal{R}}$. By the Calderon-Vaillancourt theorem, $\tilde{\mathcal{R}}$ is a bounded operator, but the bounds depends on the parameter $\mu$. In the terms of $\tilde{\mathcal{R}}$, $\tau^{k+1}$ appears and so we can consider $\operatorname{Op}_h^{\operatorname{W}}\left(\tilde{\mathcal{R}}\right)$ as a bounded operator for the topology $L^2\left(\langle\tau\rangle^{k+1} d\tau d\sigma\right)$.

On the other hand, we see that
$$ \operatorname{Op}_h^{\operatorname{W}}\left( \mathcal{Q}_{0,z}\right)\circ\left( \operatorname{Op}_h^{\operatorname{W}}\left(\mathcal{P}_z\right)-  \operatorname{Op}_h^{\operatorname{W}}\left(\mathcal{P}_{0,z}\right)\right) $$
is of order $h$ for the topology of $L^2\left(\langle\tau\rangle^{2k+3} d\tau d\sigma\right)$. This power of $\tau^{2k+3}$ comes from the terms of $\mathcal{P}_1$ in \eqref{p1}. Therefore, \eqref{inverse} is proved.\\
The proof of \eqref{inverse1} was already established in \cite[Proposition 3.1.7]{keraval}.
\end{proof}
\subsection{Tangential coercivity estimates}
The goal of this subsection is to prove the following Theorem which gives tangential elliptic estimate for the truncated operator $\operatorname{Op}_h^{\operatorname{W}}\left(p_h\right)$. We recall that $\Phi$ is a non-negative Lipchitzian function, verifying Assumption \ref{hyp2}
\begin{theorem}
Let $c_0>0$ and $\chi_0\in\mathcal{C}_c^{\infty}(\mathbb{R})$ which equals $1$ in the neighborhood of $0$. There exist $c,h_0,R_0>0$ such that, for all $R>R_0$, there exists $C_R>0$ such that
for all $h\in(0,h_0)$ and all $z\in\mathbb{C}$ such that $\left|z-\mu_0^{[k]}\right|\leq c_0h$, and for all $\psi\in \operatorname{Dom}\left(\operatorname{Op}_h^{\operatorname{W}}(p_h)\right)$,
$$ cR^2h\left\| \psi \right\|\leq \left\| \left( \operatorname{Op}_h^{\operatorname{W}}(p_h)-z\right)\psi \right\|+C_Rh\left\| \chi_0\left(\frac{\sigma-s_r}{Rh^{1/2}}\right)\psi\right\|+h\left\|\tau^{2k+3}\psi\right\|. $$
\label{Theo}
\end{theorem}
The procedure for proving this Theorem is the same as the one followed in \cite[Theorem 4.2]{BHR-purely}, but what differs here is the eikonal equation. We first prove the following proposition which will be the main ingredient in the proof of Theorem \ref{Theo}.
\begin{prop}
Let $c_0>0$. There exist $C,h_0>0$ such that, for all $z\in\mathbb{C}$ with $\left|z-\mu_0^{[k]}\right|\leq c_0 h$, and for all $\psi\in \operatorname{Dom}\left(\operatorname{Op}_h^{\operatorname{W}}(p_h)\right)$,
\begin{equation}
    \displaystyle \int_{\mathbb{R}^2} \mathfrak{E}_{\Phi}(\sigma)\left|\psi\right|^2d\sigma d\tau-Ch\left\|\psi\right\|^2\leq -\Re \langle \operatorname{Op}_h^{\operatorname{W}}q_{0,z}^{\pm}\psi,\psi\rangle,
    \label{k3}
\end{equation}
where
\begin{equation}
    \mathfrak{E}_{\Phi}(\sigma):=\Re \left( \gamma(\sigma)^{\frac{2}{k+2}}\nu^{[k]}\left( \xi_0^{[k]}-\Im\left(\varphi(\sigma) \right)+\mathrm{i}\gamma(\sigma)^{-\frac{1}{k+2}}\Phi'(\sigma)\right)-\gamma_0^{\frac{2}{k+2}}\nu^{[k]}(\xi_0^{[k]})\right).
    \label{e_Phi}
\end{equation}
Moreover, for some $c>0$ and all $R>0$, there exists $C_R>0$ such that
\begin{equation}
    cR^2h\left\|\psi\right\|\leq \left\|\operatorname{Op}_h^{\operatorname{W}}q_{0,z}^{\pm}\psi\right\|+C_Rh\left\|\chi_0\left(\frac{\sigma-s_r}{Rh^{1/2}}\right)\psi\right\|.
    \label{k2}
\end{equation}
\label{prop2}
\end{prop}
\begin{proof}
By Lemma \ref{Lem grushin}, we have 
\begin{align*}
    -q_{0,z}^{\pm}&=\mu^{[k]}(\sigma,\chi_1(\xi)+\mathfrak{w}(\sigma))-z\\&=\gamma(\sigma)^{\frac{2}{k+2}}\nu^{[k]}\left(\gamma(\sigma)^{-\frac{1}{k+2}}\chi_1(\xi)+\gamma(\sigma)^{-\frac{1}{k+2}}\mathfrak{w}(\sigma)\right)-\gamma_0^{\frac{2}{k+2}}\nu^{[k]}(\xi_0^{[k]})+\mathcal{O}(h)\\&= \gamma(\sigma)^{\frac{2}{k+2}}\nu^{[k]}\left(\xi_0^{[k]}+\gamma(\sigma)^{-\frac{1}{k+2}}\chi(\xi-\gamma_0^{\frac{1}{k+2}}\xi_0^{[k]})-\Im \varphi(\sigma)+\mathrm{i}\gamma(\sigma)^{-\frac{1}{k+2}}\Phi'(\sigma)\right) \\&-\gamma_0^{\frac{2}{k+2}}\nu^{[k]}(\xi_0^{[k]})+\mathcal{O}(h).
\end{align*}
Then, $-\Re \left( q_{0,z}^{\pm}\right)$ is written in the form
\begin{equation}
    -\Re \left( q_{0,z}^{\pm}\right)=\mathfrak{E}_{\Phi}(\sigma)+\gamma(\sigma)^{\frac{2}{k+2}}r_{\Phi}(\sigma,\xi)+\mathcal{O}(h), 
    \label{k1}
\end{equation}
with
\begin{align*}
    r_{\Phi}(\sigma,\xi)&=\nu^{[k]}\left(\xi_0^{[k]}+\gamma(\sigma)^{-\frac{1}{k+2}}\chi(\xi-\gamma_0^{\frac{1}{k+2}}\xi_0^{[k]})-\Im \varphi(\sigma)+\mathrm{i}\gamma(\sigma)^{-\frac{1}{k+2}}\Phi'(\sigma)\right)\\&-\nu^{[k]}\left( \xi_0^{[k]}-\Im\left(\varphi(\sigma) \right)+\mathrm{i}\gamma(\sigma)^{-\frac{1}{k+2}}\Phi'(\sigma)\right),
\end{align*}
and the expression of $\mathfrak{E}_{\Phi}(\sigma)$ is given in \eqref{e_Phi}. Using the Taylor expansion for the two terms of $r_{\Phi}$ at $\xi_0^{[k]}$ (for fixed $\sigma$) and the fact that the functions $\varphi(\sigma)$ and $\gamma(\sigma)^{-\frac{1}{k+2}}\Phi'(\sigma)$ are controlled by $\left\|1-\left( \frac{\gamma_0}{\gamma} \right)^{\frac{1}{k+2}} \right\|_{\infty}^{1/2}$, we obtain
\begin{align*}
    r_{\Phi}(\sigma,\xi)&=\nu^{[k]}\left(\xi_0^{[k]}+\gamma(\sigma)^{-\frac{1}{k+2}}\chi(\xi-\gamma_0^{\frac{1}{k+2}}\xi_0^{[k]})\right)-\nu^{[k]}(\xi_0^{[k]})\\&+\mathcal{O}\left( \left\|1-\left( \frac{\gamma_0}{\gamma} \right)^{\frac{1}{k+2}} \right\|_{\infty}^{1/2} \right)\text{min}\left(1,\left|\xi-\gamma_0^{\frac{1}{k+2}}\xi_0^{[k]}\right|\right).
\end{align*}
Furthermore, since $\left(\nu^{[k]}\right)'(\xi_0^{[k]})=0$ and $\left(\nu^{[k]}\right)''(\xi_0^{[k]})>0$, there exists a constant $c_1>0$ (independent of $\sigma$) such
that
$$ \gamma(\sigma)^{\frac{2}{k+2}}\left( \nu^{[k]}\left(\xi_0^{[k]}+\gamma(\sigma)^{-\frac{1}{k+2}}\chi(\xi-\gamma_0^{\frac{1}{k+2}}\xi_0^{[k]})\right)-\nu^{[k]}(\xi_0^{[k]}) \right)\geq c_1 \text{min}\left(1,\left|\xi-\gamma_0^{\frac{1}{k+2}}\xi_0^{[k]}\right|^2\right). $$
Using the fact that $\left\|1-\left( \frac{\gamma_0}{\gamma} \right)^{\frac{1}{k+2}} \right\|_{\infty}$ is small enough and from the Young inequality, we get
$$\gamma(\sigma)^{\frac{2}{k+2}} r_{\Phi}+\mathcal{O}(h)\geq -Ch, $$
and \eqref{k1} becomes
$$ -\Re \left(q_{0,z}^{\pm}\right)\geq \mathfrak{E}_{\Phi}(\sigma)-Ch. $$
We apply the Fefferman-Phong inequality (see \cite[Theorem 3.2]{inegalite}) to get
$$ \displaystyle \int_{\mathbb{R}^2} \mathfrak{E}_{\Phi}(\sigma)\left|\psi\right|^2d\sigma d\tau-Ch\left\|\psi\right\|^2\leq -\Re \langle \operatorname{Op}_h^{\operatorname{W}}q_{0,z}^{\pm}\psi,\psi\rangle. $$
Using Assumption \ref{hyp2}, the fonction $\Phi$ verifies 
\begin{align*}
    \displaystyle \int_{\mathbb{R}^2} \mathfrak{E}_{\Phi}(\sigma)\left|\psi\right|^2d\sigma d\tau &\geq \displaystyle \int_{ \left|\sigma-s_r\right|\geq Rh^{1/2} } \mathfrak{E}_{\Phi}(\sigma)\left|\psi\right|^2d\sigma d\tau \\&\geq cR^2h\displaystyle \int_{ \left|\sigma-s_r\right|\geq Rh^{1/2} } \left|\psi\right|^2d\sigma d\tau\\&=cR^2h\left\| \psi\right\|^2-cR^2h\displaystyle \int_{ \left|\sigma-s_r\right|\leq Rh^{1/2} } \left|\psi\right|^2d\sigma d\tau.
\end{align*}
So using \eqref{k3}, we get
$$ (cR^2-C)h\left\|\psi\right\|^2\leq \left\|\operatorname{Op}_h^{\operatorname{W}}q_{0,z}^{\pm}\right\|\left\|\psi\right\|+cR^2h\displaystyle \int_{ \left|\sigma-s_r\right|\leq Rh^{1/2} } \left|\psi\right|^2d\sigma d\tau, $$
and for $R$ large enough, \eqref{k3} is well established.
\end{proof}
The proof of Theorem \ref{Theo} is then the same as the one of Theorem 4.2 in \cite{BHR-purely}, but here with the use of the two Propositions \ref{prop1} and \ref{prop2}.
\section{Removing the frequency cutoff}
\label{section6}
The goal of this section is to prove that Theorem \ref{Theo} remains true when we replace the truncate operator $\operatorname{Op}_h^{\operatorname{W}}p_h$ by the operator without frequency cutoff $\mathcal{N}_h^{[k],\Phi}$ defined in \eqref{N-tilde}. For this purpose, we want to prove the following Theorem.
\begin{theorem}
Let $c_0>0$. Under Assumption \ref{hyp2}, there exist $c,h_0>0$ such that for all $h\in(0,h_0)$ and all $z\in\mathbb{C}$ such that $\left|z-\mu_0^{[k]}\right|\leq c_0h$, and all $\psi\in\operatorname{Dom}\left(\mathcal{N}_h^{[k],\Phi}\right)$,
$$ ch\left\|\psi\right\|\leq \left\|\langle\tau\rangle^{2k+3}\left(\mathcal{N}_h^{[k],\Phi}-z\right)\psi\right\|+h\left\|\chi_0\left(\frac{\sigma-s_r}{Rh^{1/2}}\right)\psi\right\|. $$
\label{theo2}
\end{theorem}
In all what follows, we shall consider the smooth function $\mathbb{R}\ni\xi\mapsto\chi_2(\xi)$, such that 
\begin{equation}
    \chi_2\left(2\gamma_0^{\frac{1}{k+2}}\xi_0^{[k]}-\xi\right)=\chi_2(\xi)\,\,\,\,\text{for all}\,\,\,\,\,\xi\in\mathbb{R},
    \label{sy}
\end{equation}
and $\chi_2(\xi)=0$ in a neighborhood of $\gamma_0^{\frac{1}{k+2}}\xi_0^{[k]}$, $\chi_2(\xi)=1$ on $\left\{ \xi\in\mathbb{R}:\chi_1(\xi)=\xi \right\}^{c}$ so that the support of $\chi_2$ avoids $\gamma_0^{\frac{1}{k+2}}\xi_0^{[k]}$. We will now deal with some lemmas that help us to prove Theorem \ref{theo2}.
\begin{lemma}
Let $c_0>0$. Under Assumption \ref{hyp2}, there exist $C,h_0>0$ such that for all $h\in(0,h_0)$ and all $z\in\mathbb{C}$ such that $\left|z-\mu_0^{[k]}\right|\leq c_0h$, and all $\psi\in\text{Dom}\left(\tilde{\mathcal{N}}_h^{[k],\Phi}\right)$,
\begin{equation}
\left\| D_{\tau}\psi \right\|+\left\| \left( hD_{\sigma}-\gamma(\sigma)\frac{\tau^{k+1}}{k+1}\right) \psi\right\|\leq C\left\| \left( \tilde{\mathcal{N}}_h^{[k],\Phi}-z\right)\psi \right\|+C\left\| \psi\right\|.
    \label{le0}
\end{equation}
\end{lemma}
\begin{proof}
For all $\psi\in Dom(\tilde{\mathcal{N}}_h^{[k],\Phi})$, we have
$$ \langle\tilde{\mathcal{N}}_h^{[k],\Phi}\psi,\psi\rangle \geq c \left\| D_{\tau}\psi\right\|^2+c \left\|\left(hD_{\sigma}+\mathfrak{w}(\sigma)-\gamma(\sigma)\frac{\tau^{k+1}}{k+1}\right)\psi\right\|^2,$$
which implies that
$$ \left\|D_{\tau}\psi\right\|+\left\|\left( hD_{\sigma}+\mathfrak{w}(\sigma)-\gamma(\sigma)\frac{\tau^{k+1}}{k+1}\right)\psi\right\|  \leq C \left\| \tilde{\mathcal{N}}_h^{[k],\Phi}\psi \right\|+C\left\| \psi \right\|. $$
Using the fact that $\left|\mathfrak{w}(\sigma)\right|$ is bounded, \eqref{le0} is well established.
\end{proof}
\begin{lemma}
Let $c_0>0$. Under Assumption \ref{hyp2}, there exist $C,h_0>0$ such that for all $h\in(0,h_0)$ and all $z\in\mathbb{C}$ such that $\left|z-\mu_0^{[k]}\right|\leq c_0h$, and all $\psi\in\operatorname{Dom}\left(\tilde{\mathcal{N}}_h^{[k],\Phi}\right)$,
$$ \left\| \operatorname{Op}_h^{\operatorname{W}}\chi_2\psi \right\|+\left\|\left(  hD_{\sigma}-\gamma(\sigma)\frac{\tau^{k+1}}{k+1} \right) \operatorname{Op}_h^{\operatorname{W}}\chi_2\psi \right\|+\left\|D_{\tau}\operatorname{Op}_h^{\operatorname{W}}\chi_2\psi \right\|\leq C \left\|\left(\tilde{\mathcal{N}}_h^{[k],\Phi}-z\right)\operatorname{Op}_h^{\operatorname{W}}\chi_2\psi\right\|. $$
\label{Lemm1}
\end{lemma}
\begin{proof}
We have:
\begin{align}
    &\Re\langle\left(\tilde{\mathcal{N}}_h^{[k],\Phi}-z\right)\operatorname{Op}_h^{\operatorname{W}}\chi_2\psi,\operatorname{Op}_h^{\operatorname{W}}\chi_2\psi\rangle\\&=\Re\langle\left(D_{\tau}^2+\left(hD_{\sigma}+\mathfrak{w}(\sigma)-\gamma(\sigma)\frac{\tau^{k+1}}{k+1}\right)^2+o(1)-z\right)\operatorname{Op}_h^{\operatorname{W}}\chi_2\psi,\operatorname{Op}_h^{\operatorname{W}}\chi_2\psi\rangle\\&\geq (1+o(1))\langle\left(D_{\tau}^2+\left(hD_{\sigma}+\mathfrak{w}(\sigma)-\gamma(\sigma)\frac{\tau^{k+1}}{k+1}\right)^2\right)\operatorname{Op}_h^{\operatorname{W}}\chi_2\psi,\operatorname{Op}_h^{\operatorname{W}}\chi_2\psi\rangle\\&-\Re( z) \left\| \operatorname{Op}_h^{\operatorname{W}}\chi_2\psi \right\|^2 \\&\geq \left( \mu^{[k]}(\sigma,\xi+\mathfrak{w}(\sigma))-\mu_0^{[k]} +o(1)\right)\left\| \operatorname{Op}_h^{\operatorname{W}}\chi_2\psi\right\|^2.
    \label{align}
\end{align}
We can assume without loss of generality that $\xi\geq\gamma_0^{\frac{1}{k+2}}\xi_0^{[k]}$, and by the symmetry of $\chi_2$ in \eqref{sy}, the results are true when $\xi\leq\gamma_0^{\frac{1}{k+2}}\xi_0^{[k]}$. By the choice of $\Phi$ (see \eqref{poidss}), $\left|\gamma(\sigma)^{-\frac{1}{k+2}}\Phi'(\sigma)\right|$ is small enough for all $\sigma\in\mathbb{R}$, and using \cite[Theorem 1.2]{BHR-holomorphic}, we get
\begin{align*}
   &\Re \nu^{[k]}\left(\xi_0^{[k]}+\gamma(\sigma)^{-\frac{1}{k+2}} \left(\xi-\gamma_0^{\frac{1}{k+2}}\xi_0^{[k]}\right) -\Im\left( \varphi(\sigma) \right)+\mathrm{i}\gamma(\sigma)^{-\frac{1}{k+2}}\Phi'(\sigma)\right)\\&\geq\nu^{[k]}\left(\xi_0^{[k]}+\gamma(\sigma)^{-\frac{1}{k+2}} \left(\xi-\gamma_0^{\frac{1}{k+2}}\xi_0^{[k]}\right) -\Im\left( \varphi(\sigma) \right)\right)-\gamma(\sigma)^{-\frac{2}{k+2}}\Phi'(\sigma)^2,
\end{align*}
and
$$ \mu^{[k]}(\sigma,\xi+\mathfrak{w}(\sigma))\geq\gamma(\sigma)^{\frac{2}{k+2}}\nu^{[k]}\left(\xi_0^{[k]}+\gamma(\sigma)^{-\frac{1}{k+2}} \left(\xi-\gamma_0^{\frac{1}{k+2}}\xi_0^{[k]}\right) -\Im\left( \varphi(\sigma) \right)\right)-\Phi'(\sigma)^2.$$
When $\xi\in\text{Supp}(\chi_2)$, $\xi$ is far from $\gamma_0^{\frac{1}{k+2}}\xi_0^{[k]}$ and for a certain $c>0$ we have 
$$\left|\xi-\gamma_0^{\frac{1}{k+2}}\xi_0^{[k]}\right|\geq c.$$
Using the fact that $\left|\varphi(\sigma)\right|$ is small enough, we get in this case
$$ \left| \gamma(\sigma)^{-\frac{1}{k+2}} \left(\xi-\gamma_0^{\frac{1}{k+2}}\xi_0^{[k]}\right) -\Im\left( \varphi(\sigma) \right) \right|\geq \frac{\gamma_{\infty}^{-\frac{1}{k+2}}c}{2}, \,\,\,\,\,\text{for all}\,\,\,\,\sigma\in\mathbb{R}. $$
Then, there exists a constant $r>0$ such that for all $\sigma\in\mathbb{R}$, we have
$$ \xi_0^{[k]}+\gamma(\sigma)^{-\frac{1}{k+2}} \left(\xi-\gamma_0^{\frac{1}{k+2}}\xi_0^{[k]}\right) -\Im\left( \varphi(\sigma) \right)\in \mathbb{R}\setminus\text{B}(\xi_0^{[k]},r). $$
Since the function $\mathbb{R}\ni\xi\mapsto\nu^{[k]}(\xi)$ admits a unique non-degenerate minimum at $\xi_0^{[k]}$, then there exists $c_1>0$ such that for all $\sigma\in\mathbb{R}$, we have
$$ \nu^{[k]}\left( \xi_0^{[k]}+\gamma(\sigma)^{-\frac{1}{k+2}} \left(\xi-\gamma_0^{\frac{1}{k+2}}\xi_0^{[k]}\right) -\Im\left( \varphi(\sigma) \right) \right)\geq c_1. $$
Therefore,
\begin{align*}
    \mu^{[k]}(\sigma,\xi+\mathfrak{w}(\sigma))-\mu_0^{[k]}&\geq \gamma(\sigma)^{\frac{2}{k+2}}c_1-\mu_0^{[k]}-\Phi'(\sigma)^2\\&\geq\gamma_0^{\frac{2}{k+2}}(c_1-\nu^{[k]}(\xi_0^{[k]}))-\Phi'(\sigma)^2\\&\geq\frac{\gamma_0^{\frac{2}{k+2}}(c_1-\nu^{[k]}(\xi_0^{[k]}))}{2}=C_1>0.
\end{align*}
Going back to \eqref{align}, and we get
$$ \left\|\operatorname{Op}_h^{\operatorname{W}}\chi_2\psi \right\|^2 \leq C \times\Re \langle\left(\tilde{\mathcal{N}}_h^{[k],\Phi}-z\right)\operatorname{Op}_h^{\operatorname{W}}\chi_2\psi,\operatorname{Op}_h^{\operatorname{W}}\chi_2\psi\rangle . $$
Combining this inequality with \eqref{le0}, we get 
$$ \left\| \operatorname{Op}_h^{\operatorname{W}}\chi_2\psi \right\|+\left\|\left(  hD_{\sigma}-\gamma(\sigma)\frac{\tau^{k+1}}{k+1} \right) \operatorname{Op}_h^{\operatorname{W}}\chi_2\psi \right\|+\left\|D_{\tau}\operatorname{Op}_h^{\operatorname{W}}\chi_2\psi \right\|\leq C \left\|\left(\tilde{\mathcal{N}}_h^{[k],\Phi}-z\right)\operatorname{Op}_h^{\operatorname{W}}\chi_2\psi\right\|. $$
\end{proof}
\begin{lemma}
Let $N\in\mathbb{N}$, $c_0>0$. Under Assumption \ref{hyp2}, there exist $C,h_0>0$ such that for all $h\in(0,h_0)$ and all $z\in\mathbb{C}$ such that $\left|z-\mu_0^{[k]}\right|\leq c_0h$, and all $\psi\in\operatorname{Dom}\left(\tilde{\mathcal{N}}_h^{[k],\Phi}\right)$,
\begin{align*}
& \left\| \operatorname{Op}_h^{\operatorname{W}}\chi_2\psi \right\|+\left\|\left(  hD_{\sigma}-\gamma(\sigma)\frac{\tau^{k+1}}{k+1} \right) \operatorname{Op}_h^{\operatorname{W}}\chi_2\psi \right\|+\left\|D_{\tau}\operatorname{Op}_h^{\operatorname{W}}\chi_2\psi \right\| \\&+  \left\|\left(  hD_{\sigma}-\gamma(\sigma)\frac{\tau^{k+1}}{k+1} \right)^2 \operatorname{Op}_h^{\operatorname{W}}\chi_2\psi \right\|+\left\|D_{\tau}^2 \operatorname{Op}_h^{\operatorname{W}}\chi_2\psi \right\| \leq C\left\| \left( \tilde{\mathcal{N}}_h^{[k],\Phi}-z\right)\psi\right\|+O(h^N)\left\| \psi\right\|.
\end{align*}
\label{Lemm2}
\end{lemma}
\begin{proof}
Using Lemma \ref{Lemm1}, the proof of this lemma is exactly as the one of \cite[Lemma 5.4]{BHR-purely}.
\end{proof}
We will control now $hD_{\sigma}$ instead of $hD_{\sigma}-\gamma(\sigma)\frac{\tau^{k+1}}{k+1}$. Since $\gamma$ is bounded, then it suffices to control $\tau^{k+1}$ with the normal Agmon estimates.
\begin{lemma}
Let $c_0>0$. Under Assumption \ref{hyp2}, for all $k\geq 1$, there exist $C,h_0>0$ such that for all $h\in(0,h_0)$ and all $z\in\mathbb{C}$ such that $\left|z-\mu_0^{[k]}\right|\leq c_0h$, and all $\psi\in\operatorname{Dom}\left(\tilde{\mathcal{N}}_h^{[k],\Phi}\right)$,
\begin{equation}
     \left\| \left[ \tilde{\mathcal{N}}_h^{[k]\Phi},\tau^k\right]\psi \right\|\leq C\left\| \langle \tau \rangle^{k-1}\left(\tilde{\mathcal{N}}_h^{[k],\Phi}-z\right)\psi \right\| +C\displaystyle\sum_{j=0}^{k-1}\left\| \tau^j\psi \right\|.
     \label{le2}
\end{equation}
\end{lemma}
\begin{proof}
By calculating the commutator $\left[ \tilde{\mathcal{N}}_h^{[k]\Phi},\tau^k\right]=\left[ \mathfrak{a}_h^{-1}D_{\tau}\mathfrak{a}_hD_{\tau},\tau^k\right]$, and using the fact that $\left[D_{\tau},\tau^k\right]=\frac{k}{\mathrm{i}}\tau^{k-1}$ for $k\geq 1$, we have:
$$ \left[ \tilde{\mathcal{N}}_h^{[k]\Phi},\tau^k\right]=\begin{cases}
\frac{2}{\mathrm{i}}D_{\tau}-\mathfrak{a}_h^{-1}\left(\partial_{\tau}\mathfrak{a}_h\right) & \text{if $k=1$}, \\
\frac{2k}{\mathrm{i}}\tau^{k-1} D_{\tau}-k\tau^{k-1} \mathfrak{a}_h^{-1}(\partial_{\tau}\mathfrak{a}_h)-k(k-1)\tau^{k-2} & \text{if $k\geq 2$}.
\end{cases} $$
For $k=1$, using \eqref{le0}, we have
\begin{align*}
    \left\| \left[ \tilde{\mathcal{N}}_h^{[k]\Phi},\tau\right]\psi \right\|&\leq C\left\| D_{\tau}\psi \right\|+C\left\| \psi\right\| \\&\leq C\left\|\left(\tilde{\mathcal{N}}_h^{[k],\Phi}-z\right)\psi\right\|+\left\| \psi \right\|,
\end{align*}
for all $\psi\in\text{Dom}\left(\tilde{\mathcal{N}}_h^{[k],\Phi}\right)$. Then, \eqref{le2} is true for $k=1$.\\
By induction on $k\geq 1$, we assume that \eqref{le2} is true up to order $k-1$ and show that it is true for order $k$. In effect, for all $\psi\in\text{Dom}\left(\tilde{\mathcal{N}}_h^{[k],\Phi}\right)$, we get
\begin{align*}
    \left\| \left[ \tilde{\mathcal{N}}_h^{[k]\Phi},\tau^k\right]\psi \right\|&\leq C\left\|\tau^{k-1}D_{\tau}\psi\right\|+C\left\|\tau^{k-2}\psi\right\|+C\left\|\psi\right\|\\&\leq C \left\|D_{\tau}\left(\tau^{k-1}\psi\right)\right\|+C\left\|\tau^{k-2}\psi\right\|+C\left\|\psi\right\|,
\end{align*}
and by \eqref{le0}, we get
\begin{align*}
    \left\| \left[ \tilde{\mathcal{N}}_h^{[k]\Phi},\tau^k\right]\psi \right\|&\leq C\left\|\left(\tilde{\mathcal{N}_h^{[k],\Phi}}-z\right)\tau^{k-1}\psi\right\|+C\left\|\tau^{k-1}\psi\right\|+C\left\|\tau^{k-2}\psi\right\|+C\left\|\psi\right\|\\&\leq C\left\| \langle \tau \rangle^{k-1}\left(\tilde{\mathcal{N}}_h^{[k],\Phi}-z\right)\psi \right\|+C\left\| \left[ \tilde{\mathcal{N}}_h^{[k]\Phi},\tau^{k-1}\right]\psi \right\|+C\left\|\tau^{k-1}\psi\right\|\\&+C\left\|\tau^{k-2}\psi\right\|+C\left\|\psi\right\|.
\end{align*}
Using the induction hypothesis for order $k-1$, \eqref{le2} is true for order $k$.
\end{proof}
\begin{lemma}
Let $c_0>0$. Under Assumption \ref{hyp2}, for all $k\geq 1$, there exist $C,h_0>0$ such that for all $h\in(0,h_0)$ and all $z\in\mathbb{C}$ such that $\left|z-\mu_0^{[k]}\right|\leq c_0h$, and all $\psi\in\operatorname{Dom}\left(\tilde{\mathcal{N}}_h^{[k],\Phi}\right)$,
\begin{equation}
    \left\|\tau^k\psi\right\|\leq C\left\| \langle \tau \rangle^k\left(\tilde{\mathcal{N}}_h^{[k],\Phi}-z\right)\psi \right\|+C\left\|\psi\right\|.
    \label{le3}
\end{equation}
\end{lemma}
\begin{proof}
The proof is quite simple by noting that 
$$ \left\|\tau^k\psi\right\|\leq C\left\| \left(\tilde{\mathcal{N}}_h^{[k],\Phi}-z\right)\left(\tau^k\psi\right) \right\|\leq C\left\| \langle \tau \rangle^k\left(\tilde{\mathcal{N}}_h^{[k],\Phi}-z\right)\psi \right\|+C  \left\| \left[ \tilde{\mathcal{N}}_h^{[k]\Phi},\tau^k\right]\psi \right\|,$$
and using \eqref{le2}, we get
$$ \left\|\tau^k\psi\right\|\leq   C\left\| \langle \tau \rangle^k\left(\tilde{\mathcal{N}}_h^{[k],\Phi}-z\right)\psi \right\|+C\displaystyle\sum_{j=0}^{k-1}\left\| \tau^j\psi \right\|. $$
By induction on $k\geq 1$, \eqref{le3} is well established.
\end{proof}
\begin{prop}
Let $N\in\mathbb{N}$, $c_0>0$. Under Assumption \ref{hyp2}, there exist $C,h_0>0$ such that for all $h\in(0,h_0)$ and all $z\in\mathbb{C}$ such that $\left|z-\mu_0^{[k]}\right|\leq c_0h$, and all $\psi\in\operatorname{Dom}\left(\tilde{\mathcal{N}}_h^{[k],\Phi}\right)$,
\begin{align*}
& \left\| \operatorname{Op}_h^{\operatorname{W}}\chi_2\psi\right\|+\left\| D_{\tau} \operatorname{Op}_h^{\operatorname{W}}\chi_2\psi\right\|+\left\| hD_{\sigma} \operatorname{Op}_h^{\operatorname{W}}\chi_2\psi\right\|+\left\| D_{\tau}^2 \operatorname{Op}_h^{\operatorname{W}}\chi_2\psi\right\| +\left\| (hD_{\sigma})^2 \operatorname{Op}_h^{\operatorname{W}}\chi_2\psi\right\| \\& +\left\| \tau^{k+1}hD_{\sigma} \operatorname{Op}_h^{\operatorname{W}}\chi_2\psi\right\|\leq C\left\| \langle\tau\rangle^{2k+2} \left( \tilde{\mathcal{N}}_h^{[k],\Phi}-z\right)\psi\right\|+\mathcal{O}\left( h^N \right) \left\| \psi \right\|. 
\end{align*}
\label{propp}
\end{prop}
\begin{proof}
By applying \eqref{le3} to $\operatorname{Op}_h^{\operatorname{W}}\chi_2\psi$, we have
$$ \left\|\tau^{k+1}\operatorname{Op}_h^{\operatorname{W}}\chi_2\psi\right\|\leq   C\left\| \langle \tau \rangle^{k+1}\left(\tilde{\mathcal{N}}_h^{[k],\Phi}-z\right)\operatorname{Op}_h^{\operatorname{W}}\chi_2\psi \right\| +C\left\|\operatorname{Op}_h^{\operatorname{W}}\chi_2\psi\right\|,$$
and by calculating the commutator $\left[\tilde{\mathcal{N}}_h^{[k],\Phi},\operatorname{Op}_h^{\operatorname{W}}\chi_2 \right]$ and using Lemma \ref{Lemm2}, we get
\begin{equation}
    \left\|\tau^{k+1}\operatorname{Op}_h^{\operatorname{W}}\chi_2\psi\right\|\leq   C\left\| \langle \tau \rangle^{k+1}\left(\tilde{\mathcal{N}}_h^{[k],\Phi}-z\right)\psi \right\| +\mathcal{O}(h^N) \left\| \psi\right\|.
    \label{i7}
\end{equation}
Likewise, we get
\begin{equation}
    \left\|\tau^{2k+2}\operatorname{Op}_h^{\operatorname{W}}\chi_2\psi\right\|\leq   C\left\| \langle \tau \rangle^{2k+2}\left(\tilde{\mathcal{N}}_h^{[k],\Phi}-z\right)\psi \right\| +\mathcal{O}(h^N) \left\| \psi\right\|.
    \label{i8}
\end{equation}
Since $\gamma$ is bounded, then using \eqref{i7} and Lemma \ref{Lemm2}, we get
\begin{align*}
    \left\| hD_{\sigma}\operatorname{Op}_h^{\operatorname{W}}\chi_2\psi\right\|&\leq \left\| \left( hD_{\sigma}-\gamma(\sigma)\frac{\tau^{k+1}}{k+1}\right) \operatorname{Op}_h^{\operatorname{W}}\chi_2\psi\right\|+\left\| \gamma(\sigma)\frac{\tau^{k+1}}{k+1} \operatorname{Op}_h^{\operatorname{W}}\chi_2\psi\right\| \\& \leq C\left\| \langle \tau \rangle^{k+1}\left(\tilde{\mathcal{N}}_h^{[k],\Phi}-z\right)\psi \right\| +\mathcal{O}(h^N) \left\| \psi\right\|.
\end{align*}
Likewise as the proof of \cite[Proposition 5.6]{BHR-purely}, we have
$$ \left\| \left( hD_{\sigma} \right)^2 \operatorname{Op}_h^{\operatorname{W}}\chi_2\psi\right\|\leq C\left\| \langle\tau\rangle^{2k+2}\left( \mathcal{N}_h^{[k],\Phi}-z\right)\psi\right\| + \mathcal{O}\left(h^N\right)\left\| \psi \right\|,  $$
and
$$ \left\| \tau^{k+1} hD_{\sigma}  \operatorname{Op}_h^{\operatorname{W}}\chi_2\psi\right\|\leq C\left\| \langle\tau\rangle^{2k+2}\left( \mathcal{N}_h^{[k],\Phi}-z\right)\psi\right\| + \mathcal{O}\left(h^N\right)\left\| \psi \right\|. $$
\end{proof}
We now have all the elements to prove Theorem \ref{Theo}. Using the result of Proposition \ref{propp} and like the same strategy from the proof of \cite[Theorem 5.1]{BHR-purely}, we get
$$ ch\left\|\psi\right\|\leq \left\|\langle\tau\rangle^{2k+3}\left(\tilde{\mathcal{N}}_h^{[k],\Phi}-z\right)\psi\right\|+h\left\|\chi_0\left(\frac{\sigma-s_r}{Rh^{1/2}}\right)\psi\right\|, $$
for all $\psi\in\text{Dom}\left(\tilde{ \mathcal{N}}_h^{[k],\Phi} \right)$. But we use the fact that $\tilde{\mathcal{N}}_h^{[k],\Phi}=\mathrm{e}^{-\frac{\mathrm{i}\mathfrak{g}(\sigma)}{h}}\mathcal{N}_h^{[k],\Phi}\mathrm{e}^{\frac{\mathrm{i}\mathfrak{g}(\sigma)}{h}}$, then Theorem \ref{Theo} is well established.

\section{Optimal tangential Agmon estimates}
\label{Section 7}
\subsection{Agmon estimates}
Let us give the optimal Agmon estimates for the eigenfunctions of the two operators $\mathcal{N}_{h,r}^{[k]}$ and $\mathcal{N}_h^{[k]}$. The following corollary is a consequence of Theorem \ref{theo2}.
\begin{cor}\textbf{(Single well)}
Let $c_0>0$. Under Assumption \ref{hyp2}, there exist $C,h_0>0$ such that for all $h\in(0,h_0)$ and all $\lambda$ eigenvalue of $\mathcal{N}_{h,r}^{[k]}$ such that $\left|\lambda-\mu_0^{[k]}\right|\leq c_0h$, and all associated eigenfunction $\Psi\in \operatorname{Dom}\left( \mathcal{N}_{h,r}^{[k]} \right)$,
$$ \displaystyle \int_{\mathbb{R}\times\mathbb{R}} \mathrm{e}^{\frac{2\Phi}{h}} \left|\Psi\right|^2dsdt \leq C\left\| \Psi \right\|_{L^2\left(\mathbb{R}\times\mathbb{R}\right)}^2. $$
\label{single}
\end{cor}
\begin{proof}
By applying Theorem \ref{theo2} with $\psi=\mathrm{e}^{\frac{\Phi}{h}}\Psi$ and $z=\lambda$, we get
$$ c\left\| \mathrm{e}^{\frac{\Phi}{h}}\Psi \right\|\leq  \left\|\chi_0\left( \frac{\sigma-s_r}{Rh^{1/2}} \right)\mathrm{e}^{\frac{\Phi}{h}} \Psi \right\|. $$
Since the function $\frac{\Phi}{h}$ is bounded on Supp$\left(\sigma\mapsto \chi_0\left( \frac{\sigma-s_r}{Rh^{1/2}} \right) \right)$, then
$$ \left\| \mathrm{e}^{\frac{\Phi}{h}}\Psi\right\| \leq C\left\| \Psi \right\|. $$
\end{proof}
We recall that the two Agmon distances for the two single-well operators $\mathcal{N}_{h,r}^{[k]}$ and $\mathcal{N}_{h,l}^{[k]}$, are respectively given by
$$ \Phi_r(\sigma)=\displaystyle\int_{s_r}^{\sigma}\gamma_r(\tilde{\sigma})^{\frac{1}{k+2}}\Re \varphi_r(\tilde{\sigma}) d\tilde{\sigma} \text{  and  } \Phi_l(\sigma)=\displaystyle\int_{s_l}^{\sigma}\gamma_l(\tilde{\sigma})^{\frac{1}{k+2}}\Re \varphi_l(\tilde{\sigma}) d\tilde{\sigma}. $$

We will consider the operator $\mathcal{N}_h^{[k]}$ with two wells defined on $L^2\left( \Gamma \times \mathbb{R}\right)\sim L^2\left( [-L,L) \times \mathbb{R}\right)$. For $\hat{\eta}>0$ small enough, we denote by
$$ \text{B}_r(\hat{\eta}):=\text{B}(s_r,\hat{\eta})=(s_r-\hat{\eta},s_r+\hat{\eta}),\text{ and } \text{B}_l(\hat{\eta}):=\text{B}(s_l,\hat{\eta})=(s_l-\hat{\eta},s_l+\hat{\eta}). $$

We define the two $2L$-periodic functions on $[-L,+L)$ so that
$$\tilde{\Phi}_r(\sigma)=
\begin{cases}
\Phi_r(\sigma) & \text{if} \,\,\,-L\leq \sigma \leq s_l-\hat{\eta}\,\,, \\
\Phi_r(\sigma-2L) & \text{if} \,\,\,s_l+\hat{\eta}\leq\sigma\leq L\,\,,
\end{cases}$$
and
$$\tilde{\Phi}_l(\sigma)=
\begin{cases}
\Phi_l(\sigma+2L) & \text{if} \,\,\,-L\leq \sigma \leq s_r-\hat{\eta}\,\,, \\
\Phi_l(\sigma) & \text{if} \,\,\,s_r+\hat{\eta}\leq\sigma\leq L\,\,,
\end{cases}$$
and that $\tilde{\Phi}_r>\tilde{\Phi}_l$ on $\text{B}_l(\eta)$ and $\tilde{\Phi}_l>\tilde{\Phi}_r$ on $\text{B}_r(\eta)$.

For $\theta\in (0,1)$, we define the function $\phi$ on $\Gamma$ by
$$ \phi=\sqrt{1-\theta}\text{min}\left( \tilde{\Phi}_r,\tilde{\Phi}_l\right). $$
\begin{prop}\textbf{(Double well)}
Let $\epsilon>0$ and $\hat{\eta}>0$ small enough. There exist $C,h_0>0$ such that for all $h\in(0,h_0)$ and all $\lambda$ eigenvalue of $\mathcal{N}_h^{[k]}$ such that $\left|\lambda-\mu_0^{[k]}\right|\leq c_0h$, and all associated eigenfunction $u\in \operatorname{Dom}\left( \mathcal{N}_h^{[k]} \right)$,
$$ \displaystyle \int_{[-L,L)\times\mathbb{R}}\mathrm{e}^{\frac{2\phi}{h}}\left| u \right| ^2d\sigma d\tau \leq C\mathrm{e}^{\frac{\epsilon}{h}}\left\| u\right\| _{L^2\left( [-L,L)\times\mathbb{R}\right)}. $$
\label{double}
\end{prop}
\begin{proof}
Let $\lambda$ be an eigenvalue of $\mathcal{N}_h^{[k]}$ such that $\left|\lambda-\mu_0^{[k]}\right|\leq c_0h$ and $u$ the associated eigenfunction. We denote by $\chi_r$ the smooth cutoff function which is equal to $0$ for $\sigma\in\text{B}_l(\hat{\eta})$ and $1$ for $\sigma\in\Gamma\setminus\text{B}_l(2\hat{\eta})$. The function $\phi$ is defined on $[-L,+L)$, so we will extend $\phi$ so that it is defined on $\mathbb{R}$ and verifies Assumption \ref{hyp2}. Therefore, we consider the function $\psi=\chi_r\mathrm{e}^{\frac{\phi}{h}}u$ as a function on $\mathbb{R}$ and we apply the Theorem \ref{theo2} with $z=\lambda$ to obtain that
$$ ch\left\| \chi_r \mathrm{e}^{\frac{\phi}{h}}u \right\|\leq \left\| \langle \tau \rangle ^{2k+3} \left( \mathcal{N}_h^{[k],\phi}-z\right)\chi_r \mathrm{e}^{\frac{\phi}{h}}u\right\|+h\left\| \chi_0\left( \frac{\sigma-sr}{Rh^{1/2}}\right)\mathrm{e}^{\frac{\phi}{h}}u\right\|. $$
By using that $u$ is an eigenfunction of $\mathcal{N}_h^{[k]}$ associated with $\lambda$, we get
$$ \left\| \langle \tau \rangle ^{2k+3} \left( \mathcal{N}_h^{[k],\phi}-z\right)\chi_r \mathrm{e}^{\frac{\phi}{h}}u\right\|=\left\| \langle \tau \rangle ^{2k+3}\left[ \mathcal{N}_h^{[k]},\chi_r\right]u\right\|. $$
But $\text{Supp}\left(\chi_r'\right)\subset (s_l-2\hat{\eta},s_l-\hat{\eta})\cup(s_l+\hat{\eta},s_l+2\hat{\eta})$, and so for $\hat{\eta}$ small enough, we can assume that $\phi\leq \frac{\epsilon}{2}$ in Supp$\left(\left[ \mathcal{N}_h^{[k]},\chi_r\right]\right)$. Therefore,
$$ ch\left\| \chi_r \mathrm{e}^{\frac{\phi}{h}}u \right\|\leq \mathrm{e}^{\frac{\epsilon}{2h}} \left\| \langle \tau \rangle ^{2k+3} \left[ \mathcal{N}_h^{[k]},\chi_r\right] u \right\|_{L^2\left([-L,L)\times\mathbb{R}\right)} +Ch\left\| u \right\|_{L^2\left([-L,L)\times\mathbb{R}\right)}. $$
By the normal Agmon estimates,
$$ \left\| \chi_r \mathrm{e}^{\frac{\phi}{h}}u \right\|\leq C \mathrm{e}^{\frac{\epsilon}{h}}\left\| u \right\|_{L^2\left([-L,L)\times\mathbb{R}\right)}, $$
and by symmetry, we get
$$ \left\| \chi_l \mathrm{e}^{\frac{\phi}{h}}u \right\|\leq C \mathrm{e}^{\frac{\epsilon}{h}}\left\| u \right\|_{L^2\left([-L,L)\times\mathbb{R}\right)}. $$
Since $\hat{\eta}$ is small enough, then
\begin{align*}
    \left\| \mathrm{e}^{\frac{\phi}{h}}u\right\|_{L^2\left( (-L,+L)\times\mathbb{R} \right)} &\leq \left\| \mathrm{e}^{\frac{\phi}{h}}u\right\|_{L^2\left( (-L,0)\times\mathbb{R} \right)}+\left\| \mathrm{e}^{\frac{\phi}{h}}u\right\|_{L^2\left( (0,+L)\times\mathbb{R} \right)}\\&\leq \left\| \chi_r \mathrm{e}^{\frac{\phi}{h}}u \right\|+\left\| \chi_l \mathrm{e}^{\frac{\phi}{h}}u \right\|\\&\leq  C \mathrm{e}^{\frac{\epsilon}{h}}\left\| u \right\|_{L^2\left([-L,L)\times\mathbb{R}\right)}.
\end{align*}
\end{proof}
\subsection{WKB approximation in the right well}
\label{Section 8}
In order to perform the tunneling analysis, an explicit approximation of the ground state energy of the single well operators must be found. This approximation is a direct consequence of the Theorem \ref{theo2}.\\
For $\hat{\eta}>0$, we denote by $I_{\hat{\eta},r}:=(s_l-2L+\hat{\eta},s_l-\hat{\eta})$. Let $\psi_{h,r}^{[k]}=\chi_{\hat{\eta},r}\Psi_{h,r}^{[k]}$, with
\begin{enumerate}
\item[$\bullet$]$\chi_{\eta,r}$ a smooth cutoff function such that $\chi_{\hat{\eta},r}\equiv 1$ on $I_{2\hat{\eta},r}$ and $\chi_{\hat{\eta},r}\equiv 0$ on $\mathbb{R}\setminus I_{\hat{\eta},r}$, that is to say $\text{supp}\left( \chi_{\hat{\eta},r}\right)\subset I_{\hat{\eta},r}$.
\item[$\bullet$]$\Psi_{h,r}^{[k]}$ is the WKB expansions (already defined in Theorem \ref{BKW}), such that $\left\| \Psi_{h,r}^{[k]}\right\|=1$.
\item[$\bullet$]Let $\Pi_r$ the orthogonal projection on the first eigenspace $\operatorname{span}\left\{ u_{h,r}^{[k]}\right\}$ for $\mathcal{N}_{h,r}^{[k]}$.
\end{enumerate}
\begin{prop}
We have
$$ \left\| \psi_{h,r}^{[k]}-\Pi_r\psi_{h,r}^{[k]}\right\|_{L^2\left( \mathbb{R}\times\mathbb{R}\right)}=\mathcal{O}(h^{\infty}). $$
\label{akher1}
\end{prop}
\begin{proof}
Using the fact that the spectral gap between the lowest eigenvalues of $\mathcal{N}_{h,r}^{[k]}$ is of order $h$ (see Theorem \ref{BKW}) and that 
$$ \psi_{h,r}^{[k]}-\Pi_r \psi_{h,r}^{[k]}\in \left( \text{vect} \left\{ u_{h,r}^{[k]}\right\}\right)^{\perp}, $$
the Min-Max principle proves that there exists $c>0$ such that
$$ ch\left\| \psi_{h,r}^{[k]}-\Pi_r\psi_{h,r}^{[k]} \right\|\leq \left\| \left( \mathcal{N}_{h,r}^{[k]}-\mu_1^{sw}(h)\right)\psi_{h,r}^{[k]}  \right\|, $$
where $\mu_1^{sw}(h)$ is the smallest eigenvalue of $\mathcal{N}_{h,r}^{[k]}$ associated with $u_{h,r}^{[k]}$. Therefore, by applying Theorem \ref{BKW}, we get
$$ \left\| \psi_{h,r}^{[k]}-\Pi_r\psi_{h,r}^{[k]}\right\|_{L^2\left( \mathbb{R}\times\mathbb{R}\right)}=\mathcal{O}(h^{\infty}). $$
\end{proof}
The following Lemma gives some properties on the weight $\hat{\Phi}_{r,N,h}$ introduced in Proposition \ref{poidss}, and the proof of this Lemma is exactly like that of \cite[Lemma 2.6]{BHR-circle}.
\begin{lemma}
Let $K\subset I_{2\hat{\eta}}$ be a compact set. For all $N\in\mathbb{N}^{\ast}$ there exists $\epsilon_0>0$ such that for all $\epsilon\in(0,\epsilon_0)$, there exist $h_0>0$ and $R>0$ such that, for all $h\in(0,h_0)$, we have
\begin{enumerate}
\item[(1)]$\hat{\Phi}_{r,N,h}\leq \Phi_r$ on $I_{\hat{\eta},r}$.
\item[(2)]$\hat{\Phi}_{r,N,h}=\tilde{\Phi}_{r,N,h}$ on $K$.
\item[(3)]$\hat{\Phi}_{r,N,h}=\sqrt{1-\epsilon}\Phi_r=\Phi_{r,\epsilon}$ on $\operatorname{supp}\chi_{\hat{\eta},r} '$.
\end{enumerate}
\label{na}
\end{lemma}
Using Theorem \ref{theo2}, Lemma \ref{na}, we follow the same proof as \cite[Proposition 6.3]{BHR-purely} (that of \cite[Proposition 2.7]{BHR-circle} as well) and we obtain the following Proposition. 
\begin{prop}
We have
$$ \mathrm{e}^{\frac{\Phi_r}{h}}\left( \Psi_{h,r}^{[k]}-u_{h,r}^{[k]} \right)=\mathcal{O}(h^{\infty}), $$
and
$$ \langle \tau \rangle^{k+1} \mathrm{e}^{\frac{\Phi_r}{h}}\left( \Psi_{h,r}^{[k]}-u_{h,r}^{[k]} \right)=\mathcal{O}(h^{\infty}), $$
in $\mathcal{C}^1\left( K ; L^2\left( \mathbb{R}\times \mathbb{R}\right)\right)$, where $K\subset I_{2\hat{\eta},r}$ is a compact set.
\label{App}
\end{prop}
\section{Interaction matrix and tunneling effect}
\label{section9}
The goal of this section is to estimate the difference between the first two eigenvalues, $\nu_2(h)-\nu_1(h)$, of the operator $\mathcal{N}_h^{[k]}$ which is defined in \eqref{N}. For this, we will follow the same strategy as in \cite{BHR-circle} and \cite{BHR-purely}.
We denote by $\mu_1^{sw}(h)$ the common ``single well" groundstate energy of operators $\mathcal{N}_{h,r,0}^{[k]}$ and $\mathcal{N}_{h,l,0}^{[k]}$. By Corollary \ref{single}, Proposition \ref{double} and the Min-Max principle, we get
$$ \mu_1^{sw}(h)-\tilde{\mathcal{O}}\left( \mathrm{e}^{-\frac{\mathrm{S}}{h}} \right)\leq \nu_1(h)\leq \nu_2(h)\leq \mu_1^{sw}(h)+\tilde{\mathcal{O}}\left( \mathrm{e}^{-\frac{\mathrm{S}}{h}} \right),  $$
where we recall $\mathrm{S}$ is defined in \eqref{S} and $\tilde{\mathcal{O}}$ is defined in \cite[Notation 3.3]{BHR-circle}. 

First, we will construct an orthonormal basis of 
$$ \mathcal{E}=\displaystyle \bigoplus_{i=1}^2\text{Ker}\left(\mathcal{N}_h^{[k]}-\nu_i(h)\right), $$
and we will write the matrix of the operator $\mathcal{N}_h^{[k]}$ in this base. For that, we will start with two functions $u_{h,r}^{[k]}$ and $u_{h,l}^{[k]}$ (the two eigenfunctions of $\mathcal{N}_{h,r,0}^{[k]}$ and $\mathcal{N}_{h,l,0}^{[k]}$ respectively associated with the same first eigenvalue $\mu_1^{sw}(h)$). Inspired from \eqref{aaa1} and \eqref{aaa2}, we define the two functions $\phi_{h,r}^{[k]}$ and $\phi_{h,l}^{[k]}$ by
$$ \phi_{h,r}^{[k]}(\sigma,\tau)=\begin{cases}
\mathrm{e}^{-\mathrm{i}\beta_0\sigma/{h^{k+2}}}u_{h,r}^{[k]}(\sigma,\tau) & \text{if} \,\,\,-L\leq \sigma \leq s_l-\hat{\eta}/2, \\
\mathrm{e}^{-\mathrm{i}\beta_0(\sigma-2L)/{h^{k+2}}}u_{h,r}^{[k]}(\sigma-2L,\tau) & \text{if} \,\,\,s_l+\hat{\eta}/2\leq\sigma\leq L,
\end{cases} $$
and
$$ \phi_{h,l}^{[k]}(\sigma,\tau)=\begin{cases}
\mathrm{e}^{-\mathrm{i}\beta_0(\sigma+2L)/{h^{k+2}}}u_{h,l}^{[k]}(\sigma+2L,\tau) & \text{if} \,\,\,-L\leq \sigma \leq s_r-\hat{\eta}/2, \\
\mathrm{e}^{-\mathrm{i}\beta_0\sigma/{h^{k+2}}}u_{h,l}^{[k]}(\sigma,\tau) & \text{if} \,\,\,s_r+\hat{\eta}/2\leq\sigma\leq L.
\end{cases} $$
We will truncate these two functions so that they are defined on $\Gamma\times\mathbb{R}$, and then we will build from these two functions a basis of $\mathcal{E}$.

For $\alpha\in \left\{l,r \right\}$, we introduce the quasimodes $f_{h,\alpha}^{[k]}$ defined on $\Gamma\times\mathbb{R}$ by
$$ f_{h,\alpha}^{[k]}=\chi_{\eta,\alpha}\phi_{h,\alpha}^{[k]} $$
where $\chi_{\eta,r}$ is the cut-off function introduced in the beginning of Section \ref{Section 8}, $\chi_{\eta,l}=U\chi_{\eta,r}$ is defined by the symmetry operator (see Section \ref{Left well}).

Let $\Pi$ be the orthogonal projection on $\mathcal{E}$, and consider the new quasimodes, for $\alpha\in \left\{l,r \right\}$,
$$ g_{h,\alpha}^{[k]}=\Pi f_{h,\alpha}^{[k]}. $$
By Proposition \ref{double}, we have (see Section 3 of \cite{BHR-circle})
$$ \langle f_{h,\alpha}^{[k]},f_{h,\beta}^{[k]} \rangle=1+\tilde{\mathcal{O}}\left( \mathrm{e}^{-\frac{2S}{h}} \right)\,\,\,\,\,\text{if}\,\,\,\,\,\alpha=\beta, $$
$$ \langle f_{h,\alpha}^{[k]},f_{h,\beta}^{[k]} \rangle=\tilde{\mathcal{O}}\left( \mathrm{e}^{-\frac{S}{h}} \right)\,\,\,\,\,\text{if}\,\,\,\,\,\alpha\neq\beta, $$
and
$$ \left\|g_{h,\alpha}^{[k]}-f_{h,\alpha}^{[k]} \right\|+\left\| \partial_s\left( g_{h,\alpha}^{[k]}-f_{h,\alpha}^{[k]}  \right) \right\|=\tilde{\mathcal{O}}\left( \mathrm{e}^{-\frac{S}{h}}  \right). $$
The base $\left\{ g_{h,l}^{[k]},g_{h,r}^{[k]} \right\}$ is a priori not orthonormal, and by the Gram-Schmidt process, we can transform it to an orthonormal basis $\mathcal{B}_h=\left\{ \tilde{g}_{h,l}^{[k]},\tilde{g}_{h,r}^{[k]} \right\}$ defined by
$$ \tilde{g}_{h,\alpha}^{[k]}=g_{h,\alpha}^{[k]}G^{-\frac{1}{2}}, $$
where $G$ is the Gram-Schmidt matrix $\left( \langle g_{h,\alpha}^{[k]},g_{h,\beta}^{[k]} \rangle \right)_{\alpha,\beta\in\left\{l,r\right\}}$. With this construction, $\mathcal{B}_h$ is an orthonormal basis of $\mathcal{E}$. Let $M$ be the matrix of $\mathcal{N}_h^{[k]}$ relative to the basis $\mathcal{B}_h$. Then, using \cite[Proposition 3.11]{BHR-circle}, we get
\begin{equation}
    \nu_2(h)-\nu_1(h)=2\left| w_{l,r} \right|+\tilde{\mathcal{O}}\left( \mathrm{e}^{-\frac{2S}{h}} \right),\,\,\,\,\,w_{l,r}=\langle r_{h,l}^{[k]},f_{h,r}^{[k]} \rangle,
    \label{a5}
\end{equation}
where $r_{h,\alpha}^{[k]}=\left( \mathcal{N}_h^{[k]}-\mu_1^{sw}(h) \right)f_{h,\alpha}^{[k]}$ for $\alpha\in\left\{ l,r \right\}$.

The goal is now to estimate the interaction term $w_{l,r}$. By integration by parts (see \cite[Lemma 7.1]{BHR-purely}), we have
\begin{equation}
     w_{l,r}=ih(w_{l,r}^u+w_{l,r}^d),
     \label{a2}
\end{equation}
with
$$ w_{l,r}^u=\displaystyle \int_{\mathbb{R}}\mathfrak{a}_h^{-1}\left( \phi_{h,l}^{[k]}\overline{\mathcal{D}_h\phi_{h,r}^{[k]}}+\mathcal{D}_h\phi_{h,l}^{[k]}\overline{\phi_{h,r}^{[k]}}\right)(0,\tau)d\tau, $$
and
$$  w_{l,r}^d=-\displaystyle \int_{\mathbb{R}}\mathfrak{a}_h^{-1}\left( \phi_{h,l}^{[k]}\overline{\mathcal{D}_h\phi_{h,r}^{[k]}}+\mathcal{D}_h\phi_{h,l}^{[k]}\overline{\phi_{h,r}^{[k]}}\right)(-L,\tau)d\tau, $$
where
$$ \mathcal{D}_h=hD_{\sigma}+h^{-k-1}\beta_0-\gamma(\sigma)\frac{\tau^{k+1}}{k+1}-h\tilde{\delta}(\sigma)c_{\mu}\frac{\tau^{k+2}}{k+2}+h^2c_{\mu}\mathcal{O}(\tau^{k+3}). $$
By the explicit form of $\phi_{h,\alpha}^{[k]}$, we can write
\begin{align*}
    w_{l,r}^u=&\displaystyle \int_{\mathbb{R}}\mathfrak{a}_h^{-1} u_{h,l}^{[k]}\overline{ \left(hD_{\sigma}-\gamma(\sigma)\frac{\tau^{k+1}}{k+1}-h\tilde{\delta}(\sigma)c_{\mu}\frac{\tau^{k+2}}{k+2}+h^2c_{\mu}\mathcal{O}(\tau^{k+3}) \right)u_{h,r}^{[k]}}(0,\tau)d\tau\\&+\displaystyle \int_{\mathbb{R}}\mathfrak{a}_h^{-1}\left( hD_{\sigma}-\gamma(\sigma)\frac{\tau^{k+1}}{k+1}-h\tilde{\delta}(\sigma)c_{\mu}\frac{\tau^{k+2}}{k+2}+h^2c_{\mu}\mathcal{O}(\tau^{k+3}) \right)u_{h,l}^{[k]}\overline{u_{h,r}^{[k]}}(0,\tau)d\tau.
\end{align*}
By Proposition \ref{App}, the explicit expression of the WKB in Theorem \ref{BKW} and the fact that $\Phi_r(0)+\Phi_l(0)=\mathrm{S}_u$, the expression for $w_{l,r}^u$ is given by
\begin{align*}
    &w_{l,r}^u=\displaystyle \int_{\mathbb{R}}\mathfrak{a}_h^{-1} \Psi_{h,l}^{[k]}\overline{ \left(hD_{\sigma}-\gamma(\sigma)\frac{\tau^{k+1}}{k+1}-h\tilde{\delta}(\sigma)c_{\mu}\frac{\tau^{k+2}}{k+2}+h^2c_{\mu}\mathcal{O}(\tau^{k+3}) \right)\Psi_{h,r}^{[k]}}(0,\tau)d\tau\\&+\displaystyle \int_{\mathbb{R}}\mathfrak{a}_h^{-1}\left( hD_{\sigma}-\gamma(\sigma)\frac{\tau^{k+1}}{k+1}-h\tilde{\delta}(\sigma)c_{\mu}\frac{\tau^{k+2}}{k+2}+h^2c_{\mu}\mathcal{O}(\tau^{k+3}) \right)\Psi_{h,l}^{[k]}\overline{\Psi_{h,r}^{[k]}}(0,\tau)d\tau+\tilde{\mathcal{O}}(h^{\infty})\mathrm{e}^{-\frac{\mathrm{S}_u}{h}}.
\end{align*}
Since $\Psi_{h,l}^{[k]}(0,\tau)=U\Psi_{h,r}^{[k]}(0,\tau)$ and $\Psi_{h,r}^{[k]}(\sigma,\tau)=a_{1,h}^{[k]}(\sigma,\tau)\mathrm{e}^{-\frac{\Phi_r(\sigma)}{h}}\mathrm{e}^{\frac{\mathrm{i}\mathfrak{g}_r(\sigma)}{h}}$ (see Theorem \ref{BKW}), we get
\begin{align*}
    \mathrm{e}^{\frac{\mathrm{S}_u}{h}}w_{l,r}^u=&\displaystyle \int_{\mathbb{R}}\mathfrak{a}_h^{-1} Ua_{1,h}^{[k]}\overline{ \left(hD_{\sigma}+\mathrm{i}\Phi_r'(\sigma)+\mathfrak{g}_r'(\sigma)-\gamma(\sigma)\frac{\tau^{k+1}}{k+1}  \right)a_{1,h}^{[k]}}(0,\tau)d\tau\\&+\displaystyle \int_{\mathbb{R}}\mathfrak{a}_h^{-1}\left( hD_{\sigma}-\mathrm{i}\Phi_r'(-\sigma)+\mathfrak{g}_r'(-\sigma)-\gamma(\sigma)\frac{\tau^{k+1}}{k+1} \right)Ua_{1,h}^{[k]}\overline{a_{1,h}^{[k]}}(0,\tau)d\tau+\tilde{\mathcal{O}}(h),
\end{align*}
with 
$$ \mathfrak{g}_r(\sigma)=\displaystyle\int_0^{\sigma}\gamma_r(\tilde{\sigma})^{\frac{1}{k+2}}\left( \xi_0^{[k]}-\Im \varphi_r(\tilde{\sigma}) \right)d\tilde{\sigma}, $$
where the function $\varphi_r$ is defined in Lemma \ref{Lemme}. Therefore,
\begin{align*}
\mathrm{e}^{\frac{\mathrm{S}_u}{h}}w_{l,r}^u &=2\displaystyle\int_{\mathbb{R}} \mathfrak{a}_h^{-1}\overline{ \left(\mathrm{i}\Phi_r'(0)+\mathfrak{g}_r'(0)-\gamma(0)\frac{\tau^{k+1}}{k+1}  \right)}Ua_{1,h}^{[k]} \overline{a_{1,h}^{[k]}}(0,\tau)d\tau+\mathcal{O}(h).
\end{align*}
We recall that by Theorem \ref{BKW}, we have $a_{1,h}^{[k]}=a_{1,0}^{[k]}+\mathcal{O}(h)$, with
$$ a_{1,0}^{[k]}(\sigma,\tau)=f_{1,0}(\sigma)u_{\sigma,\mathfrak{w}_r(\sigma)}^{[k]}\,\,\,\,\,\text{and}\,\,\,\,\,\mathfrak{w}_r(\sigma)=\mathrm{i}\Phi_r'(\sigma)+\mathfrak{g}_r(\sigma). $$
Using the expression of $f_{1,0}$ in Remark \ref{rmk1}, we get
\begin{equation}
    \mathrm{e}^{\frac{\mathrm{S}_u}{h}}w_{l,r}^u =\tilde{f}_{1,0}(0)^2\mathrm{e}^{-2\mathrm{i}\alpha_{1,0}(0)}\overline{\int_{\mathbb{R}}2 \left( \mathfrak{w}_r(0) -\gamma_r(0)\frac{\tau^{k+1}}{k+1}  \right)  \left( u_{0,\mathfrak{w}_r(0)}^{[k]} \right)^2 d\tau}+\mathcal{O}(h).
    \label{a1}
\end{equation}
Using \eqref{FH2}, we have
\begin{align*}
   \int_{\mathbb{R}}2 \left( \mathfrak{w}_r(0) -\gamma_r(0)\frac{\tau^{k+1}}{k+1}  \right)  \left( u_{0,\mathfrak{w}_r(0)}^{[k]} \right)^2 d\tau &= \displaystyle\int_{\mathbb{R}}\left( \left( \partial_{\xi}\mathcal{M}_{x,\xi}^{[k]}\right)_{0,\mathfrak{w}_r(0)}u_{0,\mathfrak{w}_r(0)}^{[k]}(\tau)\right)u_{0,\mathfrak{w}_r(0)}^{[k]}(\tau)d\tau\\&=\partial_{\xi}\mu^{[k]}(0,\mathfrak{w}_r(0)).
\end{align*}
According to Remark \ref{rmk1}, we can write 
$$ \tilde{f}_{1,0}(0)^2=\zeta^{1/2}\pi^{-1/2}\mathrm{A}_u\,\,\,\,\text{and}\,\,\,\,\,-\mathrm{i}\partial_{\xi}\mu^{[k]}(0,\mathfrak{w}_r(0))=\mathfrak{V}_r(0). $$
Therefore, we get
\begin{equation}
   \mathrm{i} \mathrm{e}^{\frac{\mathrm{S}_u}{h}}w_{l,r}^u=\zeta^{1/2}\pi^{-1/2}\overline{\mathfrak{V}_r(0)}\mathrm{A}_u\mathrm{e}^{-2\mathrm{i}\alpha_{1,0}(0)}+\mathcal{O}(h).
    \label{i1}
\end{equation}
By the same method, we can obtain
\begin{equation}
   -\mathrm{i} \mathrm{e}^{\frac{\mathrm{S}_d}{h}}w_{l,r}^d=\zeta^{1/2}\pi^{-1/2}\overline{\mathfrak{V}_r(-L)}\mathrm{A}_d\mathrm{e}^{-2\mathrm{i}\alpha_{1,0}(-L)}\mathrm{e}^{-2\mathrm{i}\beta_0L/h^{k+2}-2\mathrm{i}\mathfrak{g}_r(-L)/h}+\mathcal{O}(h) ,
    \label{i2}
\end{equation}
where $\mathrm{A}_d$ is defined in \eqref{Ad}.

By combining \eqref{a2}, \eqref{i1} and \eqref{i2}, we get
\begin{align*}
    w_{l,r}=&h\zeta^{1/2}\pi^{-1/2}\left(\overline{\mathfrak{V}_r(0)}\mathrm{A}_u\mathrm{e}^{-2\mathrm{i}\alpha_{1,0}(0)}\mathrm{e}^{-\frac{\mathrm{S}_u}{h}}+\overline{\mathfrak{V}_r(-L)}\mathrm{A}_d\mathrm{e}^{-2\mathrm{i}\alpha_{1,0}(-L)}\mathrm{e}^{-2\mathrm{i}\beta_0L/h^{k+2}-2\mathrm{i}\mathfrak{g}_r(-L)/h} \mathrm{e}^{-\frac{\mathrm{S}_d}{h}} \right)\\&+\mathrm{e}^{-\frac{\mathrm{S}}{h}}\mathcal{O}(h^2).
\end{align*}
By multiplying $w_{l,r}$ by $\operatorname{exp}\left(\mathrm{i} \mathfrak{g}_r(-L)/h+\mathrm{i}(\alpha_{1,0}(0)+\alpha_{1,0}(-L))+\mathrm{i}\beta_0L/h^{k+2} \right)$, and using \eqref{a5} and the fact that $h=\hbar^{\frac{1}{k+2}}$, we get
$$ \nu_2(\hbar)-\nu_1(\hbar)=2\left| \hat{w}_{l,r} \right|+\mathrm{e}^{-\frac{\mathrm{S}}{\hbar^{1/(k+2)}}}\mathcal{O}(\hbar^{\frac{2}{k+2}}), $$
with
$$ \hat{w}_{l,r}=\zeta^{1/2}\pi^{-1/2}\hbar^{\frac{1}{k+2}}\left(\overline{\mathfrak{V}_r(0)}\mathrm{A}_u\mathrm{e}^{-\frac{\mathrm{S}_u}{\hbar^{1/(k+2)}}}\mathrm{e}^{\mathrm{i}Lf(\hbar)}+\overline{\mathfrak{V}_r(-L)}\mathrm{A}_d\mathrm{e}^{-\frac{\mathrm{S_d}}{\hbar^{1/(k+2)}}}\mathrm{e}^{-\mathrm{i}Lf(\hbar)} \right), $$
where $ f(\hbar)=\frac{\mathfrak{g}_r(-L)}{\hbar^{1/(k+2)}L}-\alpha_0+\beta_0/\hbar $ and $\alpha_0$ is defined in \eqref{alpha0}.\\
Finally, combining this result with Proposition \ref{R4}, we get
$$ \lambda_2(\hbar)-\lambda_1(\hbar)=2\left| \tilde{w}_{l,r} \right|+\mathrm{e}^{-\frac{\mathrm{S}}{\hbar^{1/(k+2)}}}\mathcal{O}(\hbar^2), $$
with
$$ \tilde{w}_{l,r}=\zeta^{1/2}\pi^{-1/2}\hbar^{\frac{2k+3}{k+2}}\left(\overline{\mathfrak{V}_r(0)}\mathrm{A}_u\mathrm{e}^{-\frac{\mathrm{S_u}}{\hbar^{1/(k+2)}}}\mathrm{e}^{\mathrm{i}Lf(\hbar)}+\overline{\mathfrak{V}_r(-L)}\mathrm{A}_d\mathrm{e}^{-\frac{\mathrm{S_d}}{\hbar^{1/(k+2)}}}\mathrm{e}^{-L\mathrm{i}f(\hbar)} \right), $$
which ends the proof of Theorem \ref{Thm0}.
\\
\begin{center}
    ACKNOWLEDGEMENT
\end{center}
I would like to thank my advisor Fr{\'e}d{\'e}ric H{\'e}rau for his valuable remarks and discussion during the maturation of this paper. I thank also Bernard Helffer and Nicolas Raymond for their active reading and comments which have improved the presentation of this paper.
\bibliographystyle{plain}
\bibliography{bibliography.bib}

\end{document}